\documentclass[10pt]{article}
\usepackage{authblk}
\usepackage{titlesec}
\usepackage{anysize}
\usepackage{comment}
\usepackage{booktabs} 
\usepackage{amsmath,amsthm,amsfonts} 
\usepackage{dsfont}
\usepackage{graphicx}
\usepackage{textcomp}
\usepackage{url}
\usepackage{mathtools}
\usepackage{xcolor}
\usepackage{multirow}
\usepackage{hyperref}
\usepackage{multirow}

\DeclarePairedDelimiter\ceil{\lceil}{\rceil}
\DeclarePairedDelimiter\floor{\lfloor}{\rfloor}

\DeclareMathOperator{\opt}{\text{OPT}}
\DeclareMathOperator{\alg}{\text{ALG}}
\DeclareMathOperator{\bat}{\text{Batch}}
\DeclareMathOperator{\db}{\text{DB}}
\DeclareMathOperator{\bb}{\text{BB}}
\DeclareMathOperator*{\argmax}{arg\,max}
\DeclareMathOperator*{\argmin}{arg\,min}
\DeclareMathOperator*{\e}{\mathbb{E}}
\DeclareMathOperator*{\p}{\mathbb{P}}

\allowdisplaybreaks[4]
\marginsize{2cm}{2cm}{2.5cm}{2.5cm}

\usepackage{algorithm}
\usepackage{algpseudocodex} 

\newtheorem{assumption}{Assumption}[section]

\theoremstyle{definition}
\newtheorem{definition}{Definition}[section]

\newtheorem{theorem}{Theorem}[section]

\theoremstyle{remark}
\newtheorem{remark}{Remark}[subsection]
\newtheorem{lemma}[theorem]{Lemma}

\begin{document}
\title{Online Span Minimization for Flexible Uniform Jobs}

\author[1]{\textbf{Mozhengfu Liu}\thanks{Supported by an Adobe award and NSF 2216970 (IDEAL).}}
\author[2]{\textbf{Samir Khuller}\protect\footnotemark[1]}
\author[3]{\textbf{Xueyan Tang}\thanks{Supported by Singapore MOE AcRF Tier 1 Award RG15/25.}}
\affil[1]{Northwestern University,Evanston, IL, USA, \texttt{mozhengfuliu2027@u.northwestern.edu}}
\affil[2]{Northwestern University, Evanston, IL, USA, \texttt{samir.khuller@northwestern.edu}}
\affil[3]{Nanyang Technological University, Singapore, \texttt{asxytang@ntu.edu.sg}}

\maketitle

\begin{abstract}
Motivated by the critical need for energy-efficient scheduling in cloud computing, this paper investigates Span Minimization, a fundamental variant of the well-studied BusyTime problem. In the general BusyTime problem, $n$ jobs characterized by release times, deadlines, and processing times must be partitioned into bundles of capacity $B$, where the objective is to minimize the total active duration of the virtual machines. Span minimization addresses the specific case of unbounded capacity ($B = \infty$), a problem that serves as a vital precursor for achieving high-performance approximation guarantees in more complex scheduling environments.

While previous research established a deterministic $2$-approximation for interval jobs and a $3$-approximation for the general BusyTime problem, the online landscape of span minimization remains less explored. In this paper, we focus on the online version of span minimization. We demonstrate that randomization can be leveraged to break the known deterministic competitive barrier of $2$. For uniform-length jobs, we derive a randomized competitive upper bound of $\frac{1}{\ln{2}}\approx 1.443$ and a lower bound of $\frac{\sqrt{3}+1}{2}\approx 1.366$. Furthermore, we show that by introducing the ability to restart jobs, we can achieve an optimal competitive ratio equal to the golden ratio ($\phi \approx 1.618$). Our results provide new insights into the power of randomization and flexibility in online energy-aware scheduling.
\end{abstract}

\section{Introduction}

Motivated by the need to reduce energy utilization, we study the problem called ``span minimization", which is closely related to the well-known BusyTime problem \cite{chang2014lp,khandekar2010minimizing,koehler2017busy}.
In the most general form in the BusyTime problem, we are given access to virtual machines with batch capacity $B$ and a collection of $n$ jobs, each job $j$ with a release time, deadline and processing time (denoted by $r_j, d_j, p_j$ respectively). The goal is to partition $n$ jobs into bundles $S_1, S_2, \ldots, S_k$ with the following property. Each bundle contains a collection of jobs that can be feasibly scheduled such that no more than $B$ jobs are running simultaneously. The cost of this bundle is simply the duration for which the machine is on (from the earliest start time of any job in the bundle to the latest completion time of any job in the bundle).
We can call this $c(S_i) = e(S_i) - s(S_i)$ where $e(S_i)$ is the ending time of the set of jobs $S_i$ and $s(S_i)$ is the start time of a set of jobs $S_i$. The total cost is the sum of all bundles' costs ($\sum_i c(S_i)$). Note that we can create an unbounded number of bundles, since each is simply a rented virtual machine (VM) with some number of parallel processors.

A lot of research has been done on this problem. Since the problem is NP-hard \cite{winkler2003wavelength} even on the rather special case of interval (or rigid jobs) where $d_j=p_j+r_j$, significant research was done on approximation algorithms for it. 
The earliest algorithms for the problem only focused on the special case of interval jobs \cite{alicherry2003line,kumar2005approximation, flammini2010minimizing}, 
and the best known approximation is $2$. 
The online setting of scheduling interval jobs is also a well studied topic. 
The tight competitiveness of the non-clairvoyant and clairvoyant settings (whether a job's length is known at its release) are $\mu+\Theta(1)$ \cite{spaa2014, kamali2014, ren2016competitiveness} and $\Theta( \sqrt{\log{\mu}} )$ \cite{azar2017tight, azar2019tight} respectively, where $\mu$ is the ratio of the maximum job length to the minimum job length. 


Subsequently, for the general case (flexible jobs), Khandekar et al \cite{khandekar2010minimizing} developed a $4$ approximation, but this works in two steps.  {\em In the first step} solve the problem of {\em minimizing span}, or assuming unbounded batch capacity ($B = \infty$) to fix the starting times of the jobs. In other words, imagine that you have an infinite capacity machine - how would you schedule jobs respecting starting and ending times to minimize the cost of the schedule (the duration for which the machine is on). This is what we refer to as the ``span" of the set of jobs (see Figure \ref{fig:example}). We use this to place the jobs optimally by setting the start times, this reduces to the special case of BusyTime problem for Interval Jobs (since there is no flexibility any longer as to when the jobs start), and even this problem is NP-hard, but multiple algorithms are known that give a $2$ approximation for it as mentioned earlier. 

\begin{figure}[h]
    \centering
    \includegraphics[width=0.5\linewidth]{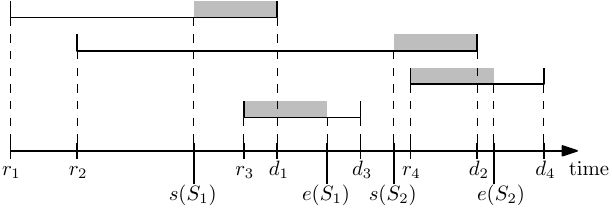}
    \caption{A suboptimal schedule of an instance of four uniform jobs. Moving job 2 earlier can reduce the overall cost.}
    \label{fig:example}
\end{figure}

However, overall this approach yields a 4 approximation for the BusyTime problem as the optimum solution might increase by a factor of 2 when we reduce to an interval instance (as shown by Chang et al \cite{chang2014lp}). Chang et al \cite{chang2014lp} develop a 3 approximation via a different approach, although in the first step we still need to solve the problem of minimizing span for unbounded batch capacity.

In the online setting of the BusyTime model, recently, Albers {\em et al.} \cite{albers2025onlinebusytimescheduling} showed a $2$ competitive online algorithm for uniform (identical length) jobs and a $9$ competitive online algorithm with $p_{max}$ lookahead for jobs with arbitrary lengths, where $p_{max}$ is the maximum job length.  


The only known method for minimizing span is based on dynamic programming and is extremely inefficient (with space complexity $\Omega(n^5)$ and time complexity $\Omega(n^7)$)\cite{khandekar2010minimizing} and is based on earlier work by Baptiste \cite{baptiste2006scheduling} developed for uniform jobs using dynamic programming. This prohibitive running time, is the main bottleneck in all good approximation algorithms for BusyTime Scheduling. Although the algorithms guarantee solutions within a factor of 3 of the optimal cost, as mentioned earlier, the issue is that the algorithms are simply not practical due to the first phase.

On the other hand, a very simple {\em online} (deterministic) algorithm to minimize the span (Algorithm Doubler) \cite{koehler2017busy}  gives a 5 approximation. The algorithm is very simple - it delays jobs as long as possible, and when a job has to be scheduled, we simply confirm running the machine for twice the length of the job; and all jobs that fit in this ``active time" when the machine is on, are scheduled.

In recent work, Liu et al \cite{DBLP:conf/spaa/LiuT24a} show that no deterministic online algorithm can get a bound better than 4.
However, while one question remains, how to close the gap between 4 and 5 - we turn our attention to how randomization could help get past the deterministic lower bound on the competitive ratio.

In practice, there are computing applications where jobs have strict deadlines and uniform processing lengths. Streaming data processing, for example, often requires processing fixed-size data blocks with strict deadlines to complete. Periodic financial reporting also involves running deadline-constrained jobs for similar data sizes.
Understanding the special case when all the jobs are uniform-length is surprisingly nontrivial in the online setting. 
We first show that for this case, there is a tight bound of 2 on the competitive ratio - both the upper and lower bounds is exactly 2. For example Algorithm Doubler obtains a 2 approximation in this case and Albers et al \cite{albers2025onlinebusytimescheduling} claim a lower bound of 2 for all deterministic online algorithms\footnote{As we show in the Appendix \ref{sec:correction}, the proof of this claim is not correct.}. We also show a lower bound of 2 on the competitive ratio of any deterministic online algorithm for uniform-length jobs [Calinescu, personal communication], however the proof is more involved.

Then we show something interesting: the upper bound of 2 can be improved via a simple oblivious randomized algorithm called $\bat^\gamma$. Instead of always stretching the window by doubling it, like in Algorithm Doubler, we pick a random value $\gamma$ chosen carefully, and run each machine by stretching the active time by a factor of $1+\gamma$.

One of the challenges in online scheduling of this type, is when we start the machine to run a job that has reached its latest-possible starting time, then jobs released later have a dilemma - should we start now, or delay (in case they have a later deadline) hoping to run for free with other jobs later on. We propose a {\em restart} model - here we start running a job, but if we learn about new jobs, it is possible to terminate the running job at no cost, but we can re-start it from scratch later on (so we lose all the progress we made). This is different from a model where pre-emption is for free and if half of a  job has completed, only the remaining half needs to be completed. In our {\em restart model}, if a job is not completed, we have to run the entire job from scratch. In the restart model we can develop online algorithms with a (tight) competitive ratio of $\phi \approx 1.618$.


Scheduling problems for uniform-length jobs have been widely studied. For example, just the (offline) feasibility of scheduling $n$ uniform-length jobs on $m$ identical machines is already a non-trivial problem \cite{Simons80,SimonsW89}. Moreover, a very clever and fast $O(n \log n)$ time algorithm for {\em one machine} was reported by Garey et al \cite{GareyJST81}. If a feasible schedule exists, the algorithm finds one. If no feasible schedule exists, the algorithm does not output any schedule. 

For {\em batch scheduling} as well, several problems for uniform-length jobs have \cite{ChangGK14} been studied with the goal of minimizing the ``on'' time for a single batching machine. An efficient algorithm called ``Lazy Activation" is described for the case of scheduling uniform jobs with release times and deadlines. Recently, Davies et al \cite{davies2022balancing} also considers the case of developing polynomial time algorithms for the tradeoff between energy utilization and average waiting time.

\begin{definition}
    For each job $j$, let $r_j\leq s_j< d_j$ denote its release time, starting deadline and deadline respectively. 
    Since we study uniform jobs, by letting each job's processing length be equal to some constant, say $1$, a valid job $j$ satisfies $s_j + 1 = d_j$. 
\end{definition}

The definition of competitiveness is stated as below. 

\begin{definition} [Competitiveness \cite{borodin1998}] \label{def:comp}
    Suppose $\mathcal{A}$ is the class of all deterministic online scheduling algorithms. 
    Suppose $\mathcal{X}$ is the class of all input job instances. 
    For any algorithm $a\in \mathcal{A}$ and any input instance $x\in \mathcal{X}$, let $a(x)$ denote the algorithm $a$'s cost with the input $x$. 
    Let $\opt(x)$ denote the offline optimal cost with input $x$.  
    Suppose $A$ is a randomized online algorithm that is a random variable taking values in $\mathcal{A}$. 
    We say $A$ is $\mathbf{c}$-competitive if $\e_{A} \left[ A(x) \right] \leq \mathbf{c}\cdot \opt(x)$ for any input instance $x\in \mathcal{X}$. 
\end{definition}

Our main results of Span Minimization, summarized in Table \ref{tab:results} together with some state-of-art results of the BusyTime problem, are stated as below:

\begin{table*}[h]
    \centering
    \scriptsize
    \begin{tabular}{|c|c|c|c|c|c|}
        \hline
          & uniform, unbounded & arbitrary, unbounded & uniform, bounded & arbitrary, bounded \\
        \hline
        offline & & optimal \cite{khandekar2010minimizing} &  & $3$-approximation \cite{chang2014lp} \\
        \hline
        \multirow{2}{*}{}
         deterministic online & $2$-competitive \cite{ren2017online,koehler2017busy,albers2025onlinebusytimescheduling} & $5$-competitive \cite{koehler2017busy} & $2$-competitive \cite{albers2025onlinebusytimescheduling} & \\
         (clairvoyant) & lower bound of $2$ (Section \ref{sec:lb_2}) & lower bound of $4$ \cite{DBLP:conf/spaa/LiuT24a} & & $\Omega(\sqrt{\log {\mu}})$ \cite{azar2019tight}\\
        \hline
        \multirow{2}{*}{randomized online}
         & $1.443$-competitive (Section \ref{sec:rand_ub}) &  && \\
        & lower bound of $1.366$ (Section \ref{sec:rand_lb}) &&&\\
        \hline
        \multirow{2}{*}{}deterministic online & $1.618$-competitive (Section \ref{sec:restart})  &  & & \\
        with restarts & lower bound of $1.618$ (Section \ref{sec:restart}) &&&\\
        \hline
    \end{tabular}
    \caption{algorithmic results of BusyTime problem in offline/online settings}
    \label{tab:results}
\end{table*}




\begin{itemize}
    \item There is a lower bound of $2$ for all deterministic online algorithms [Calinescu, personal communication]. In Appendix \ref{sec:correction}, we discuss an attempt by Albers {\em et al.} \cite{albers2025onlinebusytimescheduling} on the lower bound of $2$ and show how it is flawed. 
    \item We show a randomized online algorithm whose competitiveness is $\frac{e}{e-1}\approx 1.582$ and then improve it to have competitiveness $\frac{1}{\ln{2}}\approx 1.443$. 
    \item We show that any randomized online algorithm that is able to restart jobs cannot achieve competitiveness lower than $\frac{\sqrt{3}+1}{2}\approx 1.366$. 
    \item We show that the tight competitiveness for any deterministic online algorithm that is able to restart jobs is the golden ratio $\phi\approx 1.618$. 
\end{itemize}

{\bf Outline of Paper:}
In Section \ref{sec:lb}, we show any deterministic online algorithm cannot have competitiveness strictly below $2$. Then, in Section \ref{sec:rand_ub}, we investigate how randomization against oblivious adversaries helps to reduce the competitiveness strictly below $2$ in the online setting. Correspondingly, in Section \ref{sec:rand_lb}, we propose a non-trivial oblivious adversary to shrink the gap of competitiveness. 
Motivated by the oblivious adversary in Section \ref{sec:rand_lb}, we consider the deterministic online setting where algorithms can restart jobs during the online process in Section \ref{sec:restart} where we show its tight competitiveness is the golden ratio. Please refer to the appendixes for missing proofs. 

\section{Adversary in deterministic online setting} \label{sec:lb}

In Section \ref{sec:lb_3/2}, as a warm-up we show that the competitiveness of any deterministic online algorithm is at least $\frac{3}{2}$. 
Then, we improve the lower bound to be $2$ in Section \ref{sec:lb_2}. 

\subsection{A simple adversary} \label{sec:lb_3/2}

Let $K$ be a large integer. A rigid job $\hat{j}^1$ with $r_{\hat{j}^1} = s_{\hat{j}^1} = 0$ is released initially. For each $h = 0,1,\ldots,K$, job $j_h$ with $r_{j_h} = \frac{h}{K}$ and $s_{j_h} = 3$ is released optionally. At last, a rigid job $\hat{j}^2$ with $r_{\hat{j}^2} = s_{\hat{j}^2} = 3$ may or may not be released.  
Now, we present the adversary as described in Algorithm \ref{alg:det_3/2}. 

\begin{algorithm} 
\caption{Lower bound of 3/2}
\label{alg:det_3/2}
\algrenewcommand\algorithmicprocedure{\textbf{Upon}}
\begin{algorithmic}[1]
\State Release $\hat{j}^1$ and $j_0$ at time $0$; 
\Comment{Assume the online process starts at time $0$}
\Procedure{$j_h$ is started for some $h$}{}\Comment{Adversary waits until $j_h$ is started}
    \State $t^*\gets $ when $j_h$ is started; 
    \If{$t^*\in [0,1)$}
        \State Release job $j_{\floor{t^*\cdot K}+1}$; 
    \ElsIf{$t^*\in [1,2]$}
        \State Release $\hat{j}^2$; 
    \EndIf
\EndProcedure
\end{algorithmic}
\end{algorithm}

First, we show that the adversary is defined in the online manner. It suffices to look at lines 2-7. Note that lines 2-7 happen at the same time $t^*$. 
It suffices to ensure that all the knowledge used by the adversary are at or before time $t^*$ and all the job releases are at or after time $t^*$. 
For the former, note that at time $t^*$, it is feasible to determine the value of $t^*$ and hence line 4 and 6 are feasible. 
For the latter, the release time of $j_{\floor{t^*\cdot K}+1}$ is $\frac{\floor{t^*\cdot K}+1}{K} \geq t^*$, and the release time of $\hat{j}^2$ is $3 = s_{j_h} \geq t^*$. 

It remains to compute the lower bound of competitiveness that Algorithm \ref{alg:det_3/2} gives. 
Take any online scheduling $a$. 
Case 1: $\hat{j}^2$ is released. 
By algorithm definition, there exists some $t^*\in [1,2]$ when $j_{h}$ is started. Then, it follows that the cost of $a$ is at least $3$ that are during $[0,1]$, $[t^*,t^*+1]$ and $[3,4]$. 
For the optimal cost, each job $j_h$ is started at its starting deadline to overlap with $\hat{j}^2$, and we have the cost of $2$. 
The ratio is at least $\frac{3}{2}$. 
Case 2: $\hat{j}^2$ is not released. 
It follows that $j_h$ is never started during $[1,2]$ for any $h$.
Let $j_{h^*}$ be the last $j_h$ released. Note that $h^*$ exists. 
By Algorithm definition (line 4-5), $a$ starts some $j_h$ during $\left[\frac{h^*-1}{K},\frac{h^*}{K}\right)$. It follows that $a$ is processing jobs during $[0,\frac{h^*-1}{K}+1]$. 
By definition of $h^*$, $a$ starts job $j_{h^*}$ during $(2,3]$. It follows that the cost of $a$ during $[2,4]$ is at least $1$. 
Summing up, the cost of $a$ is at least $\frac{h^*-1}{K}+2$. 
For the optimal cost, each job $j_h$ is started at its release time. Then, the optimal cost is $\frac{h^*}{K}+1$. 
The ratio is at least $\frac{\frac{h^*-1}{K}+2}{\frac{h^*}{K}+1}$, where $h^*$ is determined by $a$. 
Note that $\frac{\frac{h^*-1}{K}+2}{\frac{h^*}{K}+1} \geq \frac{\frac{K-1}{K}+2}{2}$.  Eventually, we see that the competitiveness of $a$ is at least $\min\left\{\frac{3}{2},\frac{\frac{K-1}{K}+2}{2}\right\} = \frac{\frac{K-1}{K}+2}{2}$. 
Since $K$ is arbitrarily large, the competitiveness of $a$ approaches $\frac{3}{2}$.

\subsection[Lower bound of 2]{Lower bound of $2$} \label{sec:lb_2}

Now, we show that any deterministic online algorithm has competitiveness at least $2$. 
Let $K$ be a large integer. 
Consider the adversary as the following Algorithm \ref{alg:det_2} with parameters $K$ and $R$. 
Algorithm \ref{alg:det_2} has $R$ iterations. 
Iteration $i$ starts at $t_i$ for $i = 1,2,\ldots,R$, and job $j^i_h$ with $r_{j^i_h} = t_i + \frac{h}{K}$ and $s_{j^i_h} = t_i + 3$ for $h = 0,1,\ldots,K$ is released by the adversary optionally. 

\begin{algorithm} 
\caption{Lower bound of 2}
\label{alg:det_2}
\algrenewcommand\algorithmicprocedure{\textbf{Upon}}
\begin{algorithmic}[1]
\State $t_1 \gets 0$; \Comment{Assume the online process starts at time $0$}
\State Release rigid job $\hat{j}^1$ at time $t_1$;
\For{$i = 1,2,\ldots,R$}
    \State Release job $j^i_0$ at $t_i$; 
    \Procedure{$j^i_h$ is started for some $h$}{}
    \Comment{Adversary waits until some $j^i_h$ is started}
        \State $t^*\gets $ the current time; 
        \If{$t^*\in [t_i,t_i+1)$}
            \State Release job $j^i_{\floor{(t^*-t_i)\cdot K}+1}$; 
        \ElsIf{$t^*\in [t_i+1,t_i+2]$} \label{alg_line:correct_move}
            \State $t_{i+1}\gets t_i + 3$; 
            \State Release rigid job $\hat{j}^{i+1}$ at time $t_{i+1}$; 
        \Else
            \State $t_{i+1}\gets t^*$; 
            \State $\hat{j}^{i+1}\gets j^i_h$; 
        \EndIf
    \EndProcedure
\EndFor
\end{algorithmic}
\end{algorithm}

Consider the following running example when $K = 4$ and $R = 5$. See Figure \ref{fig:det_2} where black and gray rectangles represent the processing of $\hat{j}^i$ and $j^i_h$ respectively. 
At time $t_1$, jobs $\hat{j}^1$ and $j^1_0$ are released. Job $j^1_0$ is started at time $t_1 + 1/K$ and hence $j^1_2$ is released at time $t_1 + 2/K$. Job $j^1_2$ is started at $t_1 + 1.25 \in [t_1 + 1, t_1 + 2]$, and hence rigid job $\hat{j}^2$ is released at time $t_2 = t_1 + 3$. 
For iteration $2$, job $j^2_0$ is released at time $t_2$. Job $j^2_0$ is started at time $t_2$ and hence job $j^2_1$ is released at time $t_2 + 1/K$. Job $j^2_1$ is started at time $t_2 + 1/K$ and hence job $j^2_2$ is released at time $t_2 + 2/K$. Job $j^2_2$ is started at time $t_2 + 5/2 > t_2 + 2$, and hence $j^2_2$ itself becomes $\hat{j}^3$ and $t_3 = t_2 + 5/2$. 
For iteration $3$, job $j^3_0$ is released at time $t_3$. Job $j^3_0$ is started at time $t_3 + 2.75 > t_3 + 2$, and hence $j^3_0$ becomes $\hat{j}^4$ and $t_4 = t_3 + 2.75$. 
For iteration $4$, job $j^4_0$ is released at time $t_4$. Job $j^4_0$ is started at time $t_4 + 3/K$, and hence job $j^4_4$ is released at time $t_4 + 1$. Job $j^4_4$ is started at time $t_4 + 2.75 > t_4 + 2$, and hence job $j^4_4$ becomes $\hat{j}^5$ and $t_5 = t_4 + 2.75$. 
For iteration $5$, job $j^5_0$ is released at time $t_5$. Job $j^5_0$ is started at time $t_5 + 3/K$, and hence job $j^5_4$ is released at time $t_5+1$. Job $j^5_4$ is started at $t_5+1\in [t_5+1,t_5+2]$, and hence rigid job $\hat{j}^6$ is released at time $t_5+3$ and $t_6 = t_5 + 3$. 

\begin{figure*}[t]
  \centering
  \includegraphics[width=0.8\linewidth]{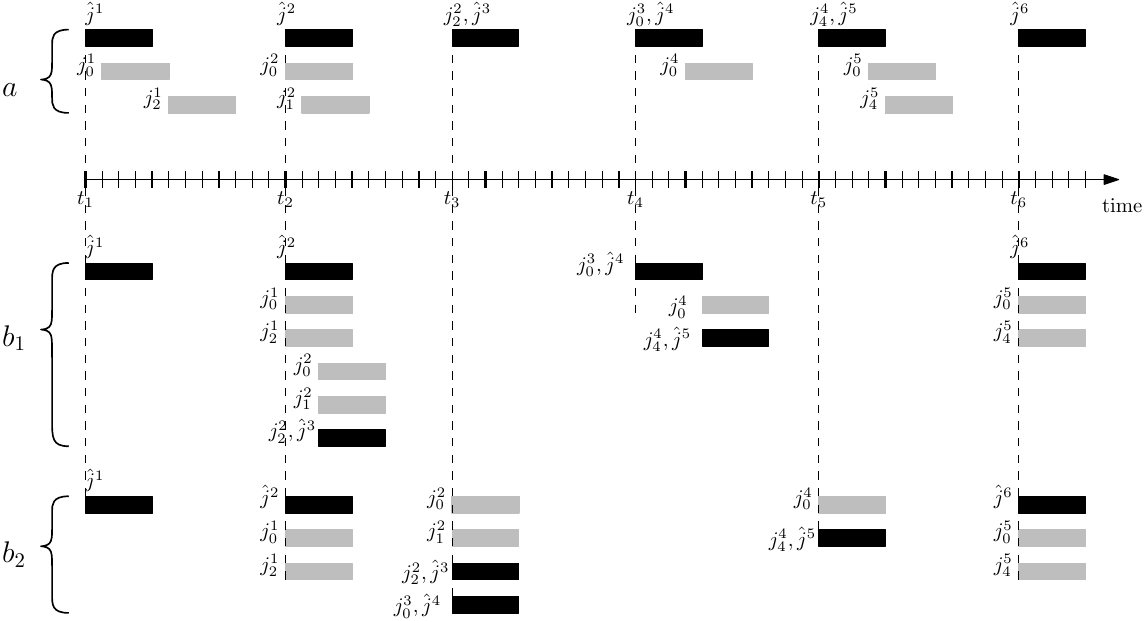}
  \caption{A running example of the lower bound of $2$ when $K = 4$ and $R = 5$. $b_1$ and $b_2$ are two options for the optimal scheduling. }
  \label{fig:det_2}
\end{figure*}

Now, we do the analysis. Take any online scheduling algorithm $a$. 
Let $b$ denote an offline algorithm to be determined later whose cost serves as an upper bound of the optimal cost. 
Consider the partitioning $S_1:=\{i: \hat{j}^{i+1}\text{~is~rigid}\}$ and $S_2:=\{i:\hat{j}^{i+1}\text{~is~not~rigid}\}$. 
Consider the cost of $a$ during $[t_i+1,t_{i+1}+1]$ for each $i$. 
For each $i$ in $S_1$, by definition, $a$ starts some $j^i_h$ at time $t^*\in [t_i+1,t_i+2]$. It follows that the cost of $a$ (during $[t_i+1,t_{i+1}+1]$) is at least $2$ that are during $[t^*,t^*+1]$ and $[t_{i+1},t_{i+1}+1]$. 
For the optimal cost, $b$ starts all the jobs $j^i_h$ at $t_{i+1} = t_{i}+3 = s_{j^i_h}$. 
It follows that the cost of $b$ (during $[t_i+1,t_{i+1}+1]$) is $1$. 

For each $i\in S_2$, let $h_i^*:=\max_{j^i_h\text{~released}} h$. 
Take any $i_1 \leq i_2$ such that $i_1-1,i_2+1\notin S_2$ and $i\in S_2$ for each $i_1\leq i \leq i_2$. 
Partition the index set $\{i_1,i_1+1,\ldots,i_2\}$ into $E:=\{i_1+2k: k \in \mathbb{N}, i_1+2k\leq i_2\}$ and $O:=\{i_1+2k+1: k \in \mathbb{N}, i_1+2k+1\leq i_2\}$. 
Consider the following two scheduling of jobs $\{j^i_h: i_1\leq i \leq i_2, j^i_h\text{~released}\}$. 
\begin{itemize}
    \item $b_1$ starts all jobs $\{j^i_h: j^i_h\text{~released}\}$ at time $t_i + \frac{h^*_i}{K}$ for each $i \in E$, and all jobs $\{j^i_h: j^i_h\text{~released}\}$ at time $t_{i+1}$ for each $i \in O$. 
    \item $b_2$ starts all jobs $\{j^i_h: j^i_h\text{~released}\}$ at time $t_{i+1}$ for each $i \in E$, and all jobs $\{j^i_h: j^i_h\text{~released}\}$ at time $t_i + \frac{h^*_i}{K}$ for each $i \in O$. 
\end{itemize}
Let $b$ choose the scheduling which gives the lower cost between $b_1$ and $b_2$. Also see Figure \ref{fig:det_2} for how $b_1$ and $b_2$ schedule the jobs in the running example. 

The ultimate goal is to show that the cost of $b$ during $[t_{i_1}+1,t_{i_2+1}+1]$ is roughly at most half of the cost of $a$ during the same period. 
Here, we compute the cost of $a$. 
Observe that job $j^i_{h^*_i}$ is released only if $a$ starts some job $j^i_{\tilde{h}_i}$ during $\left[t_i + \frac{h^*_i-1}{K}, t_i + \frac{h^*_i}{K}\right)$. 
Then, for each $i = i_1,i_1+1,\ldots,i_2$, the cost of $a$ during $[t_i+1, t_{i+1}+1]$ is  at least $\frac{h^*_i-1}{K} + 1$ where $\frac{h^*_i-1}{K}$ is the cost during $[t_i+1, t_i+1+\frac{h^*_i-1}{K}]$ and $1$ is the cost during $[t_{i+1}, t_{i+1}+1]$. 
Note that we have the two cost during $[t_i+1, t_i+1+\frac{h^*_i-1}{K}]$ and during $[t_{i+1}, t_{i+1}+1]$ disjoint because $t_{i+1} \geq t_i + 2 \geq t_i+1+\frac{h^*_i-1}{K}$. 
Summing up all the $i$, the cost of $a$ during $[t_{i_1}+1,t_{i_2+1}+1]$ is at least $\sum_{i = i_1,i_1+1,\ldots,i_2} (\frac{h^*_i-1}{K}+1)$. 

Now, we compute $b$'s cost during $[t_{i_1}+1,t_{i_2+1}+1]$ if $b$ chooses $b_1$ for the scheduling of jobs $\{j^i_h: i_1\leq i \leq i_2, j^i_h\text{~released}\}$. 
Claim that $b$ is not busy during $[t_i + 1 + \frac{h^*_i}{K}, t_{i+2}]$ for each $i\in E$. 
For any job released before $j^i_0$, $b$ starts it at or before time $t_i$. 
For each job $j^i_h$ (including $\hat{j}^{i+1}$), $b$ starts it at time $t_i + \frac{h^*_i}{K}$. 
For any job of the rest, $b$ starts it at or after time $t_{i+2}$ where $\hat{j}^{i+2}$ is rigid when $i = i_2$. 
Therefore, the claim is proved. 
One can check that $\left(\cup_{i\in E} \left[t_i + 1, t_i + 1 + \frac{h^*_i}{K}\right]\right) \cup \left( \cup_{i\in O}[t_{i+1},t_{i+1}+1] \right) \cup \left( \cup_{i\in E} [t_i + 1 + \frac{h^*_i}{K}, t_{i+2}] \right)$ is a partitioning of $[t_{i_1}+1,t_{i_2+2}]$ if $i_2\in E$, while it is also a partitioning of $[t_{i_1}+1,t_{i_2+1}+1]$ if $i_2\in O$. 
It follows from the claim that $b$'s cost during $[t_{i_1}+1,t_{i_2+1}+1]$ is at most $|O| + \sum_{i \in E} \frac{h^*_i}{K}$ where $|O|$ is the costs during $[t_{i+1},t_{i+1}+1]$ for each $i\in O$ and $\frac{h^*_i}{K}$ is the cost during $\left[t_i + 1, t_i + 1 + \frac{h^*_i}{K}\right]$ for each $i\in E$. 

Similarly, we compute $b$'s cost during $[t_{i_1}+1,t_{i_2+1}+1]$ if $b$ chooses $b_2$ for the scheduling of jobs $\{j^i_h: i_1\leq i \leq i_2, j^i_h\text{~released}\}$. 
Claim 1: $b$ is not busy during $[t_{i_1}+1,t_{i_1+1}]$. 
For any job released before $j^{i_1}_0$, $b$ starts it at or before time $t_{i_1}$. 
For each job $j^{i_1}_h$ (including $\hat{j}^{i+1}$), $b$ starts it at time $t_{i_1+1}$. 
For any job of the rest, $b$ starts it at or after time $t_{i_1+1}$. 
Therefore, claim 1 is proved. 
Claim 2: $b$ is not busy during $[t_i + 1 + \frac{h^*_i}{K}, t_{i+2}]$ for each $i\in O$. 
For any job released before $j^i_0$, $b$ starts it at or before time $t_i$. 
For each job $j^i_h$ (including $\hat{j}^{i+1}$), $b$ starts it at time $t_i + \frac{h^*_i}{K}$. 
For any job of the rest, $b$ starts it at or after time $t_{i+2}$. 
Therefore, claim 2 is proved. 
One can check that $[t_{i_1}+1, t_{i_1+1}]\cup \left(\cup_{i\in O} \left[t_i + 1, t_i + 1 + \frac{h^*_i}{K}\right]\right) \cup \left( \cup_{i\in E}[t_{i+1},t_{i+1}+1] \right) \cup \left( \cup_{i\in O} [t_i + 1 + \frac{h^*_i}{K}, t_{i+2}] \right)$ is a partitioning of $[t_{i_1}+1,t_{i_2+2}]$ if $i_2\in O$, while it is also a partitioning of $[t_{i_1}+1,t_{i_2+1}+1]$ if $i_2\in E$. 
It follows that the cost of $b$ during $[t_{i_1}+1,t_{i_2+1}+1]$ is $|E| + \sum_{i \in O} \frac{h^*_i}{K}$ where $|E|$ is the costs during $[t_{i+1},t_{i+1}+1]$ for each $i\in E$ and $\frac{h^*_i}{K}$ is the cost during $\left[t_i + 1, t_i + \frac{h^*_i}{K} + 1\right]$ for each $i\in O$. 

Since $b$ chooses the scheduling with the lower cost, we have the cost of $b$ during $[t_{i_1}+1,t_{i_2+1}+1]$ at most the average of $b_1$'s cost and $b_2$'s cost. It follows that the cost of $b$ during $[t_{i_1}+1,t_{i_2+1}+1]$ is at most $\frac{1}{2}\cdot (|O| + \sum_{i \in E} \frac{h^*_i}{K} + |E| + \sum_{i \in O} \frac{h^*_i}{K}) = \frac{1}{2}\cdot ((i_2 - i_1 + 1) + \sum_{i = i_1,i_1+1,\ldots,i_2} \frac{h^*_i}{K})$. 
Summing up all the segments like $i_1,\ldots,i_2$, we have the cost of $b$ during $\cup_{i\in S_2} [t_i+1, t_{i+1}+1]$ at most $\frac{1}{2}\cdot (|S_2| + \sum_{i\in S_2} \frac{h^*_i}{K})$. 
It follows that the cost of $b$ is at most $1 + |S_1| + \frac{1}{2}\cdot ( |S_2| + \sum_{i\in S_2} \frac{h^*_i}{K} )$ where $1$ is the cost during $[t_1,t_1+1]$, $|S_1|$ are the costs during $\cup_{i\in S_1} [t_i+1,t_{i+1}+1]$, and the remaining are the costs during $\cup_{i\in S_2} [t_i+1,t_{i+1}+1]$. 

Eventually, the cost of $a$ is at least $1 + 2\cdot |S_1| + |S_2| + \sum_{i\in S_2} \frac{h^*_i-1}{K}$ where $1$ is the cost during $[t_1,t_1+1]$, $2\cdot |S_1|$ are the costs during $\cup_{i\in S_1} [t_i+1,t_{i+1}+1]$, $|S_2|$ are the costs during $\cup_{i\in S_2} [t_{i+1},t_{i+1}+1]$, and $\frac{h^*_i-1}{K}$ is the cost during $[t_i + 1, t_i + \frac{h^*_i-1}{K} + 1]\subset [t_i + 1, t_{i+1}]$ for each $i\in S_2$. 

The competitiveness ratio of $a$ is at least $\frac{1 + 2\cdot |S_1| + |S_2| + \sum_{i\in S_2} \frac{h^*_i-1}{K}}{1 + |S_1| + \frac{1}{2}\cdot ( |S_2| + \sum_{i\in S_2} \frac{h^*_i}{K} )} = 2 - \frac{ 1 + \frac{|S_2|}{K}}{1 + |S_1| + \frac{1}{2}\cdot ( |S_2| + \sum_{i\in S_2} \frac{h^*_i}{K} )} \geq 2 - \frac{ 1 + \frac{R}{K}}{\frac{1}{2} R}$, where we have $R = |S_1|+|S_2|$. 
Clearly, the term $\frac{ 1 + \frac{R}{K}}{\frac{1}{2} R}$ goes to zero as $R$ and $K$ go to infinity. 
Therefore, the competitiveness of any online scheduling algorithm is at least $2$, since $K$ and $R$ can be chosen arbitrarily large.

\section{Randomized online algorithm against oblivious adversary} \label{sec:rand_ub}

In this section, some preparation work for algorithm analysis are presented in Section \ref{sec:rand_ub_prep}. 
Section \ref{sec:rand_ub_coin} uses a single Bernoulli variable (flip a coin) to break the barrier of $2$ against any oblivious adversary. 
Then, section \ref{sec:rand_ub_1} minimizes the competitive ratio of the same algorithm by choosing an appropriate continuous distribution. 
Finally, Section \ref{sec:rand_ub_2} makes an improvement via a new algorithmic idea. 

\subsection{Preparations} \label{sec:rand_ub_prep}

\begin{definition} \label{alg:batch}
    For each real number $\gamma\in [0,\infty)$, define $\bat^\gamma$ to be the deterministic online algorithm such that 
    \begin{itemize}
        \item upon job $j$ is released, start job $j$ now if job $j$ can be scheduled fully within $[s_f, s_f + 1 + \gamma]$ for some flag job $f$, and delay it otherwise; 
        \item upon some delayed job $f$ hits its starting deadline $s_f$, mark $f$ as a flag job with respect to $\gamma$ and start all the jobs that can be feasibly started (including $f$) now; 
    \end{itemize}
\end{definition}
\begin{remark}
    $\bat^0$ and $\bat^1$ have been investigated in different generalizations of our uniform job scheduling model \cite{ren2017online,koehler2017busy}. 
\end{remark}
\begin{definition} \label{def:batch}
    For each real number $\gamma\in [0,\infty)$, let $\mathcal{F}^{\gamma}$ denote the set of flag jobs with respect to $\gamma$ after running $\bat^\gamma$ on the whole input job instance. 
    For each $f\in \mathcal{F}^{\gamma}$, We call the interval $[s_f, s_f + 1 + \gamma]$ the $\gamma$-extension associated with the flag job $f$ with respect to $\gamma$. 
    For each flag job $f\in \mathcal{F}^{\gamma}$, let $P^\gamma_f$ denote all the non-flag jobs fitted within the $\gamma$-extension associated with it while running $\bat^\gamma$. 
\end{definition}

Take any input job instance $x\in \mathcal{X}$. 
Take any optimal scheduling of $\cup_{\gamma\in [0,\infty)} \mathcal{F}^\gamma$. 
Suppose $I_i$ with $i\in \mathcal{I}$ is the $i$-th (from left to right) maximal busy intervals (called blocks from now on) in the optimal scheduling of $\cup_{\gamma\in [0,\infty)} \mathcal{F}^\gamma$. 
Suppose $I_i = [I^-_i, I^+_i]$ for $i\in \mathcal{I}$. 
Note that $\sum_{i\in \mathcal{I}} I^+_i - I^-_i$ is an lower bound of the optimal cost of scheduling all the jobs. 
Take any $i\in \mathcal{I}$. 
Suppose $\mathcal{F}^\gamma_i$ is the set of flag jobs with respect to $\gamma$ scheduled within $I_i$ by the optimal scheduling of $\cup_{\gamma\in [0,\infty)} \mathcal{F}^\gamma$. 
Suppose $C^\gamma_i$ is the cost of processing $\cup_{f\in \mathcal{F}^\gamma_i} \{f\}\cup P^\gamma_f$ while running $\bat^\gamma$ on the whole input job instance. 
Note that the total cost of $\bat^\gamma$ is bounded by $\sum_{i\in \mathcal{I}} C^\gamma_i$. 

\begin{lemma} \label{lem:I^-}
    We have $s_f \geq I^-_i$ for each $f\in \mathcal{F}^\gamma_i$. 
\end{lemma}
\begin{proof}
    Note that the taken optimal scheduling starts job $f$ at or after time $I^-_i$. 
    It follows that $f$'s starting deadline is also at or after time $I^-_i$.
\end{proof}

\begin{lemma} \label{lem:key}
    For each $\gamma$, $|\{f\in \mathcal{F}^\gamma_i: s_f\geq I^+_i - 1 - \gamma\}| \leq 1$. 
\end{lemma}
\begin{proof}
    Prove by contradiction. Assume $f_1,f_2\in \{f\in \mathcal{F}^\gamma_i: s_f\geq I^+_i - 1 - \gamma\}$ and $f_1$ hits its starting deadline before $f_2$ when running $\bat^\gamma$. Since $f_2$ is a flag job with respect to $\gamma$, $r_{f_2} > s_{f_1}+\gamma$ must hold, otherwise $f_2$ would be scheduled within the $\gamma$-extension associated with $f_1$. 
    Then, we have $r_{f_2} > s_{f_1}+\gamma \geq (I^+_i - 1 - \gamma) + \gamma = I^+_i - 1$ which contradicts to the fact that the optimal scheduling schedules $f_2$ within $I_i$.  
\end{proof}

\begin{remark} \label{rem:barrier_2}
    If $I^+_i - I^-_i \leq 1 + \gamma$, then $|\mathcal{F}^\gamma_i| \leq 1$ by Lemma \ref{lem:I^-} and \ref{lem:key}. 
    It follows that $C^\gamma_i \leq 1 + \gamma \leq (1 + \gamma)\cdot (I^+_i - I^-_i)$. 
    If $I^+_i - I^-_i > 1 + \gamma$, then if any $\gamma$-extension associated with a flag job in $\mathcal{F}^\gamma_i$ ends after $I^+_i$, then the flag job is unique by Lemma \ref{lem:key}. 
    It follows that $C^\gamma_i \leq (I^+_i - I^-_i) + 1 + \gamma < 2\cdot (I^+_i - I^-_i)$.
    Therefore, we have for any input job instance $x\in \mathcal{X}$, the algorithm's cost $\bat^\gamma(x) = \sum_{i\in \mathcal{I}} C^\gamma_i \leq \sum_{i\in \mathcal{I}} \max\{2,1 + \gamma\}\cdot (I^+_i - I^-_i) = \max\{2,1 + \gamma\}\cdot \opt(x)$. 
    It follows that $\bat^\gamma$ for each $\gamma\in [0,1]$ is $2$-competitive. 
\end{remark}

One may think of a randomized online algorithm acting between $\bat^0$ and $\bat^1$ randomly. The following job instance (see Figure \ref{fig:bat_01}) says it just doesn't work. Let $R$ be a large integer. Consider releasing rigid job $j_i$ with $r_{j_i} = 3\cdot i$, job $j'_i$ with $r_{j'_i} = 3\cdot i + \epsilon$ and $s_{j'_i} = 3\cdot i + 1$, and job $j''_i$ with $r_{j''_i} = 3\cdot i + 1$ and $s_{j''_i} = 3\cdot i + 3$, for $i = 0,1,\ldots,R$, where $\epsilon>0$ is a small real number. Observe that either of $\bat^0$ or $\bat^1$ starts job $j''_i$ at time $3\cdot i + 1$, for each $i$. Since the scheduling for distinct $i$'s are mutually disjoint, the algorithm's cost is $2\cdot R + O(1)$. Observe that if $j'_i$ is started at $3\cdot i + \epsilon$ and $j''_i$ is started at $3\cdot i + 3$ for each $i$ (i.e. $\bat^\gamma$ for any $\gamma\in [\epsilon,1)$), then the cost of the scheduling is $(1 + \epsilon)\cdot R + O(1)$. When $R\rightarrow \infty$ and $\epsilon\rightarrow 0$, the ratio becomes close to $2$.

\begin{figure}[h]
  \centering
  \includegraphics[width=0.7\linewidth]{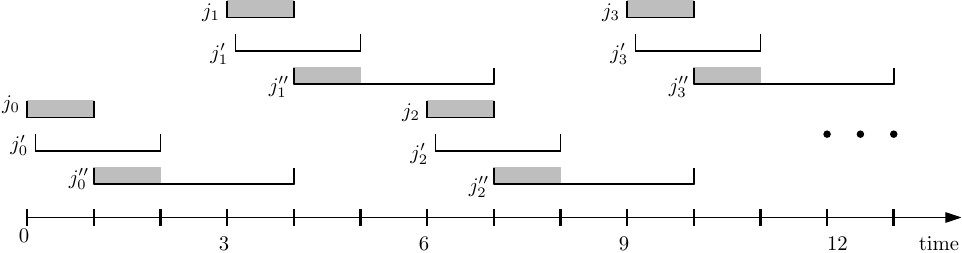}
  \caption{An instance where $\bat^0$ and $\bat^1$ make no difference}
  \label{fig:bat_01}
\end{figure}

\subsection{Barely Randomization: flipping a coin} \label{sec:rand_ub_coin}

It turns out that $\gamma$ using a value in $(0,1)$ makes randomization work. 
Let $\Gamma$ denote the random variable replacing the variable $\gamma$. 
In this subsection, we try a baby algorithm design where $\Gamma$ takes only two values. 
The algorithm initially flips a coin before all jobs are released. If Head, $\Gamma = 0$, else $\Gamma = 1/2$. The goal is to show that the baby algorithm $\bat^\Gamma$ has competitiveness strictly below $2$. 

Let $L_i$ denote $I^+_i - I^-_i$. Consider the following case-by-case analysis: 
\begin{itemize}
    \item Case 1: $L_i \leq 3/2$. 
    Looking at $\bat^{0}$, Lemma \ref{lem:I^-} and \ref{lem:key} imply $C^{0}_i \leq L_i+1 \leq 2\cdot L_i$. 
    Looking at $\bat^{1/2}$, Lemma \ref{lem:I^-} and \ref{lem:key} imply $|\mathcal{F}^{1/2}_i|\leq 1$. Thus, $C^{1/2}_i \leq 1 + 1/2 \leq 3/2\cdot L_i$. 
    
    \item Case 2: $L_i > 3/2$. 
    Looking at $\bat^{0}$, Lemma \ref{lem:I^-} and \ref{lem:key} imply $C^{0}_i \leq L_i+1 \leq 5/3 \cdot L_i$. 
    Looking at $\bat^{1/2}$, Lemma \ref{lem:I^-} and \ref{lem:key} imply $C^{1/2}_i \leq L_i + 3/2 < 2\cdot L_i$. 
\end{itemize}

Let $\alpha\in (0,1)$ denote the probability of the coin facing Head (choosing $\bat^0$). 
Then, $\alpha\cdot C^{0}_i + (1-\alpha)\cdot C^{1/2}_i \leq \mathbf{c}\cdot L$, where $\mathbf{c} = \max\left\{\alpha\cdot 2 + (1-\alpha)\cdot \frac{3}{2}, \alpha\cdot \frac{5}{3} + (1-\alpha)\cdot 2\right\}$. 
Consequently, the expected algorithm cost is $\e_\Gamma[\bat^\Gamma(x)] = \alpha\cdot \sum_{i\in \mathcal{I}} C^{0}_i + (1-\alpha)\cdot \sum_{i\in \mathcal{I}} C^{1/2}_i \leq \mathbf{c}\cdot \opt(\mathcal{J})$. 
It remains to find the best $\mathbf{c}$ by tuning $\alpha$. 
It turns out that the best $\mathbf{c}$ is $\frac{9}{5}$ achieved when $\alpha = \frac{3}{5}$. 

In fact, it turns out that $0$ and $1/2$ is the best pair of values to be taken (see Appendix \ref{sec:best_coin}).

\subsection{Barely Randomization: any distribution} \label{sec:rand_ub_1}

In this subsection, we try to determine the ``best" distribution for $\Gamma$ among any distribution such that $\Gamma\geq 0$ almost surely. 

If $1+\Gamma \geq L_i$, Lemma \ref{lem:I^-} and \ref{lem:key} imply $|\mathcal{F}^{\Gamma}_i|\leq 1$, and hence $C^{\Gamma}_i \leq 1+\Gamma = \frac{1+\Gamma}{L_i}\cdot L_i$. 
If $1+\Gamma < L_i$, Lemma \ref{lem:I^-} and \ref{lem:key} imply $C^{\Gamma}_i \leq L_i + 1+\Gamma = (1 + \frac{1+\Gamma}{L_i})\cdot L_i$.
Let $\mathbf{c}_L$ denote the random variable $\frac{1+\Gamma}{L}\cdot \mathds{1}\{1+\Gamma \geq L\} + \left(1 + \frac{1+\Gamma}{L}\right)\cdot \mathds{1}\{1+\Gamma < L\}$, for each $L\geq 1$. 
Observe that $\e_\Gamma[\bat^\Gamma(x)] = \sum_{i\in \mathcal{I}} \e_{\Gamma}[C^\Gamma_i] \leq \sum_{i\in \mathcal{I}} \e_{\Gamma}\left[(1 + \Gamma)\cdot \mathds{1}\{1+\Gamma \geq L_i\} + (L_i + 1 + \Gamma)\cdot \mathds{1}\{1+\Gamma < L_i\}\right] \\
= \sum_{i\in \mathcal{I}} \e_\Gamma[\mathbf{c}_{L_i}]\cdot L_i \leq \max_{L\geq 1} \e_\Gamma[\mathbf{c}_L]\cdot \opt(x)$ where $x$ is the input job instance. 

The goal is simply to minimize $\max_{L\geq 1} \e_\Gamma[\mathbf{c}_L]$ via choosing the ``best" distribution for $\Gamma$. 
We have $\e_\Gamma\left[\mathbf{c}_L(\Gamma)\right] = \p\left[1+\Gamma < L\right] + \frac{1}{L}\cdot\e_\Gamma\left[1+\Gamma\right]$. 

Consider the density function $f_\Gamma(\gamma):=\\
\begin{cases}
    \frac{a}{(1+\gamma)^2} \quad\text{if~} 0\leq \gamma \leq \frac{1}{a-1}\\
    0 \quad\text{otherwise}
\end{cases}$, 
where parameter $a>1$ is to be adjusted later. 
We have $\p[1+\Gamma < L] = \int^{L-1}_0 \frac{a}{(1+\gamma)^2}\,\text{d}\gamma = \left[-\frac{a}{1+\gamma}\right]^{L-1}_0 = a - \frac{a}{L}$ for each $L\in \left[1,\frac{a}{a-1}\right]$.  
It follows that  $\p[1+\Gamma < L] \leq a - \frac{a}{L}$ for each $L \geq 1$. 
We have $\e[1+\Gamma] = \int^{\frac{1}{a-1}}_0 \frac{a}{1+\gamma}\,\text{d}\gamma = a\ln(\frac{a}{a-1})$. 
Then, $\e_\Gamma\left[\mathbf{c}_L(\Gamma)\right] = \p[1+\Gamma < L] + \frac{1}{L}\cdot\e_\Gamma[1+\Gamma] \leq a - \frac{a}{L} + \frac{a\ln(\frac{a}{a-1})}{L}$. 
The idea is to let the two factors of $\frac{1}{L}$ cancel each other, i.e., $\frac{a}{a-1} = e$ or equivalently $a = \frac{e}{e-1}$. 
Therefore, the competitiveness of $\bat^\Gamma$ is $\frac{e}{e-1}\approx 1.582$ when $\Gamma\sim f_\Gamma$ and $a = \frac{e}{e-1}$. 

Even though the chosen density function $f$ looks satisfactory with respect to the objective function $\p[1+\Gamma < L] + \frac{1}{L}\cdot\e_\Gamma[1+\Gamma]$, we observe that the above analysis is not tight. Specifically, when $1+\Gamma < L$, the bound $C^\Gamma_i \leq L + 1 + \Gamma$ is not tight. 
For example, we show that the equality cannot be approached when $L = 2$ and $\Gamma=0.7$. Observe that $\mathcal{F}^{0.7}_i$, i.e. the set of flag jobs with respect to $\Gamma=0.7$ scheduled within the same block of length $L = 2$, contains at most two flag jobs, otherwise three flag jobs with respect to $\Gamma = 0.7$ would constitute an optimal block strictly longer than $2.4$ (see Figure \ref{fig:L&gamma}). Thus, $C^{0.7}_i$ is at most $2\cdot (1+\Gamma) = 3.4$ and cannot approach $L + 1 + \Gamma = 3.7$. 

\begin{figure}[h]
    \centering
    \includegraphics[width=0.7\linewidth]{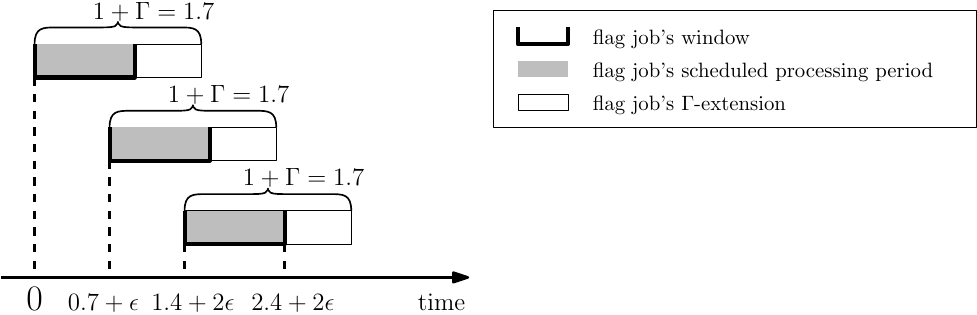}
    \caption{three rigid flag jobs with respect to $\Gamma = 0.7$ and $\epsilon>0$ is small}
    \label{fig:L&gamma}
\end{figure}

The following lemma captures the gap in the above example. 
\begin{lemma} \label{lem:gamma_multiples}
    If $L\in \left(k\cdot (1+\Gamma), 1 + (k+1)\cdot\Gamma\right]$ for some integer $k\geq 1$, then $C^\Gamma_i\leq (k+1)\cdot(1+\Gamma)$. 
\end{lemma}
\begin{proof}
    It suffices to show that $\mathcal{F}^\Gamma_i$ has at most $k+1$ jobs. For the sake of contradiction, assume $\mathcal{F}^\Gamma_i$ has $k+2$ jobs $f^\Gamma_1, f^\Gamma_2, \ldots, f^\Gamma_{k+2}$ from left to right. Observe that $s_{f^\Gamma_i} \geq r_{f^\Gamma_i} > s_{f^\Gamma_{i-1}} + \Gamma$ for each $i = 2,3,\ldots,k+2$. Thus, $r_{f^\Gamma_{k+2}} > s_{f^\Gamma_1} + (k+1)\cdot \Gamma$. 
    Observe that we have $s_{f^\Gamma_1} \geq I^-_i$ and $r_{f^\Gamma_{k+2}} \leq I^+_i - 1$, otherwise flag job $f^\Gamma_1$ or $f^\Gamma_{k+2}$ would not be able to fit into the chosen block of the optimal scheduling. 
    It follows that $ I^+_i -1 > I^-_i + (k+1)\cdot \Gamma$, and hence, $L > 1 + (k+1)\cdot \Gamma$ that is the contradiction. 
\end{proof}

For a pair of values of $L$ and $\Gamma$, Lemma \ref{lem:gamma_multiples} suggests the following two bounds corresponding to two cases respectively: 
\begin{itemize}
    \item If $L\in \left(1 + k\cdot\Gamma, k\cdot (1+\Gamma)\right]$ for some integer $k\geq 2$, $C^\Gamma_i$ may approach $L + 1 + \Gamma$; 
    \item If $L\in \left(k\cdot (1+\Gamma), 1 + (k+1)\cdot\Gamma\right]$ for some integer $k\geq 1$, $C^\Gamma_i$ may approach $(k+1)\cdot(1+\Gamma)$. 
\end{itemize}

However, the following example says the above bounds are still not tight. 
Consider the case that a block of optimal scheduling has length $L = 4.6+\epsilon$, and $\Gamma$ chooses $1.5$, where $\epsilon>0$ is small. 
Observe that $4.6+\epsilon\in \left(1 + 2\cdot 1.5, 2\cdot 2.5\right]$ implies $L + 1 + \Gamma$ can be approached when $\Gamma$ chooses $1.5$ (see Figure \ref{fig:L&multiple_gammas}). 
However, considering the same block with a second $\Gamma$ value of $1.25$. 
Observe that $4.6+\epsilon\in \left(2\cdot 2.25, 1 + 3\cdot 1.25\right]$ implies $3\cdot (1+\Gamma)$ strictly lower than $L + 1 + \Gamma$ can be approached when $\Gamma$ chooses $1.25$. 
We observe that the proposed bounds corresponding to the two values of $\Gamma$, i.e. $1.5$ and $1.25$, cannot be approached simultaneously. 
It is not hard to see that the bound $L + 1 + \Gamma$ when $\Gamma = 1.5$ can only be approached when the first two flag jobs with respect to $1.5$ have roughly the same deadline as $f^{1.5}_1$ and $f^{1.5}_2$ in Figure \ref{fig:L&multiple_gammas} respectively. Simultaneously, when $\Gamma = 1.25$, the existence of $f^{1.5}_1$ and $f^{1.5}_2$ ensures that the first flag job with respect to $1.25$ is started no later than $f^{1.5}_1$ and the second flag job with respect to $1.25$ is started no later than $f^{1.5}_2$. It follows that the cost due to the $1.25$-extension of the first two flag job with respect to $1.25$ is at most $4.35$. Since there are at most three flag job with respect to $1.25$, $C^{1.25}_i$ is at most $6.6$ instead of $6.75 = 3\cdot(1+\Gamma)$ when $C^{1.5}_i$ approaches $7.1 = L + 1 + \Gamma$. 

\begin{figure}[h]
    \centering
    \includegraphics[width=0.8\linewidth]{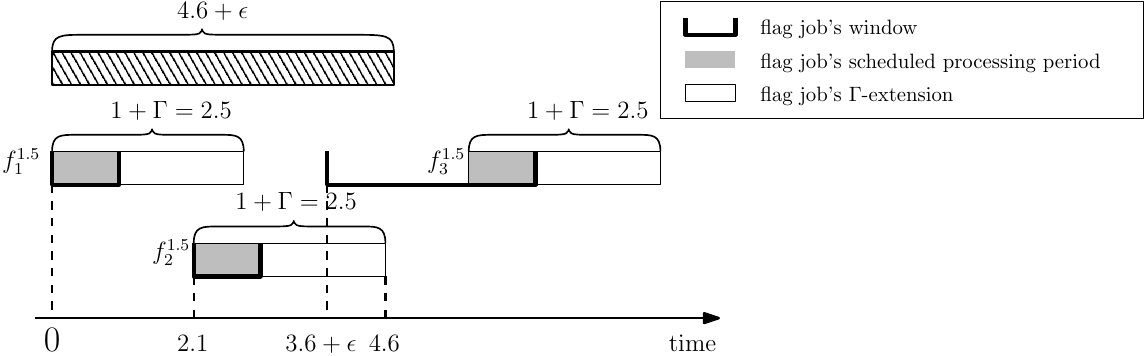}
    \caption{three flag jobs with respect to $\Gamma = 4.8$ and $\epsilon>0$ is small}
    \label{fig:L&multiple_gammas}
\end{figure}

The above discussion suggests there exists a big difficulty in determining the tight competitiveness of $\bat^\Gamma$. 
One may consider improving the competitiveness analysis of $\bat^\Gamma$ using the same density function $f$ for $\Gamma$ by using the above better bounds, but it seems to be a wrong way due to too many cases to be considered. 
Instead of working on improving the analysis of $\bat^\Gamma$, we turn to improve the algorithm design.

\subsection[Disjoint Batches]{Disjoint Batches $\db^{\gamma}$} \label{sec:rand_ub_2}

It seems that the unpleasant feature of $\bat^\gamma$ is that the $\gamma$-extensions for the same $\gamma$ may overlap with each other, so we try to modify $\bat^\gamma$ to force that all $\gamma$-extensions for the same $\gamma$ must be disjoint. 
We propose the following class of deterministic online algorithms $\db^\gamma$ for each $\gamma\geq 0$ as described in Algorithm \ref{alg:db}, where each flag job's $\gamma$-extension is defined differently from Definition \ref{def:batch}.

\begin{algorithm} 
\caption{Disjoint Batches: $\db^\gamma$}
\label{alg:db}
\algrenewcommand\algorithmicprocedure{\textbf{Upon}}
\begin{algorithmic}[1]
\Procedure{job $j$ is released}{}
    \If{starting $j$ now makes $j$'s processing fully contained within some $\gamma$-extension}
         \State start $j$ now; 
    \EndIf
\EndProcedure
\Procedure{a $\gamma$-extension starts}{}
    \State start all the remaining jobs now; 
\EndProcedure
\Procedure{time $s_f$ for some remaining job $f$}{}
    \State $f$ is called a flag job; 
    \If{$s_f$ is within some $\gamma$-extension, say $[t, t + 1 + \gamma]$}
         \State let $[t + 1 + \gamma, t + 2 + 2\gamma]$ be the $\gamma$-extension associated with the flag job $f$; 
    \Else
         \State let $[s_f, s_f + 1 + \gamma]$ be the $\gamma$-extension associated with the flag job $f$; 
    \EndIf
\EndProcedure
\end{algorithmic}
\end{algorithm}

The following lemma suggests $\db^\gamma$ is an overall improvement over $\bat^\gamma$. We keep using the same notations as those in the analysis of $\bat^\gamma$. 
\begin{lemma} \label{lem:db}
    $|\mathcal{F}^\gamma_i|\leq \left\lceil\frac{L_i}{1+\gamma}\right\rceil$, where $L_i = I^+_i - I^-_i$ is the length of the $i$-th block in the taken optimal scheduling of all the flag jobs with respect to $\gamma$ with $\gamma \in [0,\infty)$.  
\end{lemma}
\begin{proof}
    Assume $\mathcal{F}^\gamma_i$ has flag jobs $f^\gamma_1, f^\gamma_2, \ldots, f^\gamma_{\left\lceil\frac{L_i}{1+\gamma}\right\rceil + 1}$ from left to right and we show a contradiction. 
    Note that the $\gamma$-extension associated with $f^\gamma_h$ starts no earlier than $s_{f^\gamma_h}$ for each $h = 1,2,\ldots,\left\lceil\frac{L_i}{1+\gamma}\right\rceil + 1$. 
    Since all the flag jobs in $\mathcal{F}^\gamma_i$ must be feasibly fit into the block $[I^-_i, I^+_i]$, we have $s_{f^\gamma_h} \geq I^-_i$ for each $h = 1,2,\ldots,\left\lceil\frac{L_i}{1+\gamma}\right\rceil + 1$. 
    Suppose $t'$ is the end time of the $\gamma$-extension associated with $f^\gamma_{\left\lceil\frac{L_i}{1+\gamma}\right\rceil}$. 
    Since $f^\gamma_{\left\lceil\frac{L_i}{1+\gamma}\right\rceil + 1}$ is a flag job, we have $r_{f^\gamma_{\left\lceil\frac{L_i}{1+\gamma}\right\rceil + 1}} > t'-1$. 
    Since $f^\gamma_{\left\lceil\frac{L_i}{1+\gamma}\right\rceil + 1}$ can be feasibly scheduled within the taken block $[I^-_i, I^+_i]$, we have $r_{f^\gamma_{\left\lceil\frac{L_i}{1+\gamma}\right\rceil + 1}} \leq I^+_i-1$. 
    It follows that $t' < I^+_i$. 
    If follows that $[s_{f^\gamma_1}, t']$ is a proper subset of $[I^-_i, I^+_i]$. 
    Since all the $\gamma$-extensions are mutually disjoint, we have $t' -s_{f^\gamma_1} \geq \left\lceil\frac{L_i}{1+\gamma}\right\rceil\cdot (1 + \gamma) \geq L_i$. 
    It follows that $L_i > t' -s_{f^\gamma_1} \geq \left\lceil\frac{L_i}{1+\gamma}\right\rceil\cdot (1 + \gamma) \geq L_i$ (contradiction). 
\end{proof}

Note that $C^\gamma_i \leq |\mathcal{F}^\gamma_i|\cdot (1 + \gamma)$. 
It follows that if random variable $\Gamma$ follows a density function $f_\Gamma$, then Lemma \ref{lem:db} says $C^\Gamma_i$ is at most $\left\lceil\frac{L_i}{1+\Gamma}\right\rceil\cdot (1+\Gamma)$. 
Consider the randomized online algorithm $\db^\Gamma$ that generates a random sample of $\Gamma$, denoted by $\gamma$, and then run $\db^\gamma$ for scheduling all the jobs from the input. 
Let $\mathbf{c}_L$ denote the random variable $\left\lceil\frac{L}{1+\Gamma}\right\rceil\cdot \frac{1+\Gamma}{L}$, for each $L\geq 1$. 
Observe that $\e_\Gamma[\db^\Gamma(x)] = \sum_{i\in \mathcal{I}} \e_{\Gamma}[C^\Gamma_i] \leq \sum_{i\in \mathcal{I}} \e_{\Gamma}\left[\left\lceil\frac{L_i}{1+\Gamma}\right\rceil\cdot (1+\Gamma)\right] = \sum_{i\in \mathcal{I}} \e_\Gamma[\mathbf{c}_{L_i}]\cdot L_i \leq \\
\max_{L\geq 1} \e_\Gamma[\mathbf{c}_L]\cdot \opt(x)$ where $x$ is the input job instance. 

It remains to minimize $\max_{L\geq 1} \e_\Gamma[\mathbf{c}_L]$ via finding the best $f_\Gamma$. 

\begin{theorem} \label{thm:upp_best}
    $\db^\Gamma$ with $f_\Gamma(\gamma) = \begin{cases}
    \frac{1}{\ln{2}\cdot (1+\gamma)} \quad \text{if~} 0\leq \gamma\leq 1 \\
    0 \quad \text{otherwise}
\end{cases}$ is $\frac{1}{\ln{2}} \approx 1.44$ competitive. 
\end{theorem}
\begin{proof}
    Consider 
    \begin{eqnarray*}
    &\quad&  \e\left[ \left\lceil\frac{L}{1+\Gamma}\right\rceil\cdot \frac{1+\Gamma}{L} \right] \\
    &=& \e\left[ \sum_{k = 1,2,\ldots, \ceil{L}} \mathds{1}\left\{\frac{L}{k}\leq 1+\Gamma< \frac{L}{k-1}\right\}\cdot \frac{k(1+\Gamma)}{L} \right] \\
    &=& \e\left[ \mathds{1}\left\{\frac{L}{\ceil{L}}\leq 1+\Gamma< \frac{L}{\ceil{L}-1}\right\}\cdot \frac{\ceil{L}(1+\Gamma)}{L} \right] + \\
    &\quad& \e\left[ \mathds{1}\left\{\frac{L}{\ceil{L}-1}\leq 1+\Gamma< \frac{L}{\ceil{L}-2}\right\}\cdot \frac{(\ceil{L}-1)(1+\Gamma)}{L} \right]\\
    &\quad& + \cdots + \e\left[ \mathds{1}\left\{L\leq 1+\Gamma< \infty\right\}\cdot \frac{1+\Gamma}{L} \right].
    \end{eqnarray*}
    
    Consider the following different cases. 
    
    Case 1: $L=1$. Clearly, $\e\left[ \left\lceil\frac{L}{1+\Gamma}\right\rceil\cdot \frac{1+\Gamma}{L} \right] = \e[1+\Gamma] = \frac{1}{\ln{2}}$. 
    
    Case 2: $1<L\leq 2$. We have $\e\left[ \mathds{1}\left\{\frac{L}{2}\leq 1+\Gamma< L\right\}\cdot \frac{2(1+\Gamma)}{L} \right] + 
    \e\left[ \mathds{1}\left\{L\leq 1+\Gamma< \infty\right\}\cdot \frac{1+\Gamma}{L} \right] 
    = \frac{2}{L}\cdot \int^{L-1}_{0} (1+\gamma)\frac{1}{\ln{2}\cdot (1+\gamma)}\,\text{d}\gamma + \frac{1}{L}\cdot \int^{1}_{L-1} (1+\gamma)\frac{1}{\ln{2}\cdot (1+\gamma)}\,\text{d}\gamma 
    = \frac{1}{\ln{2}}\left(\frac{2}{L}(L-1) + \frac{1}{L}(2-L) \right) 
    = \frac{1}{\ln{2}}$. 

    Case 3: $2 < L \leq 3$. We have $\e\left[ \mathds{1}\left\{\frac{L}{3}\leq 1+\Gamma< \frac{L}{2}\right\}\cdot \frac{3(1+\Gamma)}{L} \right] + \e\left[ \mathds{1}\left\{\frac{L}{2}\leq 1+\Gamma< L\right\}\cdot \frac{2(1+\Gamma)}{L} \right] 
    = \frac{3}{L}\cdot \int^{\frac{L}{2}-1}_{0} (1+\gamma)\frac{1}{\ln{2}\cdot (1+\gamma)}\,\text{d}\gamma + \frac{2}{L}\cdot \int^{1}_{\frac{L}{2}-1} (1+\gamma)\frac{1}{\ln{2}\cdot (1+\gamma)}\,\text{d}\gamma 
    = \frac{1}{\ln{2}}\cdot \left( \frac{3}{L}(\frac{L}{2}-1) + \frac{2}{L}(2-\frac{L}{2}) \right) 
    = \frac{1}{\ln{2}}\cdot (\frac{1}{2} + \frac{1}{L}) < \frac{1}{\ln{2}}$, where the last inequality is because $2 < L$. 

    Case 4: $3 < L \leq 4$. We have 
    {\scriptsize
    \begin{equation*}
    \begin{split}
    &\quad \e\left[ \mathds{1}\left\{\frac{L}{4}\leq 1+\Gamma< \frac{L}{3}\right\}\cdot \frac{4(1+\Gamma)}{L} \right] + \e\left[ \mathds{1}\left\{\frac{L}{3}\leq 1+\Gamma< \frac{L}{2}\right\}\cdot \frac{3(1+\Gamma)}{L} \right] 
    + \e\left[ \mathds{1}\left\{\frac{L}{2}\leq 1+\Gamma< L\right\}\cdot \frac{2(1+\Gamma)}{L} \right] \\
    &= \frac{4}{L}\cdot \int^{\frac{L}{3}-1}_{0} (1+\gamma)\frac{1}{\ln{2}\cdot (1+\gamma)}\,\text{d}\gamma 
    + \frac{3}{L}\cdot \int^{\frac{L}{2}-1}_{\frac{L}{3}-1} (1+\gamma)\frac{1}{\ln{2}\cdot (1+\gamma)}\,\text{d}\gamma 
    + \frac{2}{L}\cdot \int^{1}_{\frac{L}{2}-1} (1+\gamma)\frac{1}{\ln{2}\cdot (1+\gamma)}\,\text{d}\gamma \\
    &= \frac{1}{\ln{2}}\cdot \left( \frac{4}{L}(\frac{L}{3} - 1) + \frac{3}{L}(\frac{L}{2}-\frac{L}{3}) + \frac{2}{L}(2-\frac{L}{2}) \right) = \frac{1}{\ln{2}}\cdot \frac{5}{6} < \frac{1}{\ln{2}}.
    \end{split}
    \end{equation*}
    }
    
    Case 5: $L > 4$. We have $\e\left[ \left\lceil\frac{L}{1+\Gamma}\right\rceil\cdot \frac{1 + \Gamma}{L} \right] \leq \e\left[ \left(\frac{L}{1+\Gamma} + 1\right)\cdot \frac{1+\Gamma}{L} \right] = 1 + \frac{\e[1+\Gamma]}{L} < 1 + \frac{\e[1+\Gamma]}{4} = 1 + \frac{1}{4\cdot \ln{2}} < \frac{1}{\ln{2}}$. 
    
\end{proof}

\section{Oblivious adversary against randomized online algorithm} \label{sec:rand_lb}

In this section, we still present a simple oblivious adversary first, which says the competitiveness of any randomized online algorithm is at least $\frac{8}{7}$. Then, we improve the adversary to achieve a better lower bound of $\frac{4}{3}$ via an uniform distribution. Finally, we observe that the uniform distribution can be modified for a better lower bound of $\frac{1 + \sqrt{3}}{2}$. 

The following lemma gives a way of constructing an oblivious adversary against any randomized online algorithm. 
\begin{lemma} [Yao's principle \cite{borodin1998}] \label{lem:yao}
    Following the notations used in Definition \ref{def:comp}, let $X$ be a random input instance on $\mathcal{X}^*\subset \mathcal{X}$, where $\mathcal{X}^*$ is a non-empty finite set and hence $X$ is a discrete random variable with values in $\mathcal{X}^*$. 
    Suppose $\mathcal{A}^*\subset \mathcal{A}$ is a finite set of deterministic online algorithms such that for each $a\in \mathcal{A}$ there exists $l(a)\in \mathcal{A}^*$ such that $l(a)(x) \leq a(x),\forall x\in \mathcal{X}^*$. 
    Then, we have the competitiveness of any randomized online algorithm is at least $\frac{\min_{a^*\in \mathcal{A}^*} \e_X[a^*(X)]}{\e_X[\opt(X)]}$.
\end{lemma}

\subsection{A baby instance}

As how the upper bounds are developed, we consider making a baby step first. 
According to Yao's principle, we turn to find a randomized instance $X$ such that there exists a set of ``best" deterministic online algorithm $\mathcal{A}^*$ such that $\min_{a^*\in \mathcal{A}^*} \frac{\e_X[a^*(X)]}{\e_X[\opt(X)]}$ is strictly greater than $1$.

Consider the following baby instance. There are three jobs defined: 
\begin{itemize}
    \item Job $j_s$ is release at time $0$ and rigid;
    \item Job $j^\delta_m$ is released at time $\delta$ with starting deadline $s_{j^\delta_m}$ at time $2$, where $0 < \delta < 1$ is to be adjusted later;
    \item Job $j_e$ is released at time $2$ and rigid; 
\end{itemize}
Consider the random job instance $X_{\delta,\alpha} := \begin{cases}
    \{j_s,j^\delta_m\} \quad w.p. \quad 1-\alpha \\
    \{j_s,j^\delta_m,j_e\} \quad w.p. \quad \alpha
\end{cases}$ on $\mathcal{X}^*_{\delta,\alpha} = \{\{j_s,j^\delta_m\}, \{j_s,j^\delta_m,j_e\}\}$, where $\delta, \alpha\in (0,1)$ are to be adjusted later. 

Observe that the ``best" deterministic online algorithm will start job $j_m$ either at $r_2 = \delta$ or at $2$. 
Think of the following argument. If some algorithm $a$ starts job $j_m$ during $(\delta, 1)$, then an alternative algorithm $a^*$ which starts job $j_m$ at $\delta$ has a strictly lower scheduling cost during $[0,2]$, and hence a strictly lower scheduling cost overall. If some algorithm $a$ starts job $j_m$ during $[1, 2)$, then an alternative algorithm $a^*$ which starts job $j_m$ at $2$ has a strictly lower scheduling cost during $[1,3]$, and hence a strictly lower scheduling cost overall.
Therefore, for each algorithm $a$, $a$ is mapped to either $\bat^1$ (job $j_m$ is started at $\delta$) or $\bat^0$ ($j_m$ is started at $2$). 

Then, compute 
$\e_X[\bat^1(X_{\delta,\alpha})] = (1-\alpha)\cdot (1+\delta) + \alpha\cdot (2+\delta) = 1+\alpha+\delta$
and
$\e_X[\bat^0(X_{\delta,\alpha})] = (1-\alpha)\cdot 2 + \alpha\cdot 2 = 2$. 
To make the adversary strong, it is intuitive to let the two candidates for the best algorithm have the same expected cost, i.e., $\alpha + \delta = 1$, and we have $\mathcal{A}^*=\{\bat^0,\bat^1\}$. 
To attain the strongest adversary, it remains to attain the lowest possible expected optimal cost, since the expected algorithm's cost is fixed to be $2$. 
We have
$\min_{\delta,\alpha: \alpha + \delta = 1} \e_X[\opt(X_{\delta,\alpha})] = \min_{\delta,\alpha: \alpha + \delta = 1} (1-\alpha)\cdot (1+\delta) + \alpha\cdot 2 = \min_{\alpha} 2 - \alpha(1-\alpha) = \frac{7}{4}$ where the minimum is attained when $\alpha = \frac{1}{2}$. 
Eventually, we obtain a lower bound of $\frac{8}{7}$ for any randomized online algorithm due to $X_{\frac{1}{2}, \frac{1}{2}}$ by Lemma \ref{lem:yao}. 

\subsection{A generalized instance} \label{sec:lb1}

Consider generalizing the baby instance for better lower bounds of the competitiveness. 
The first idea is to replace $j_m$ with a lot of jobs released during $j_s$'s processing. On top of it, the second idea is to have many rounds of the above by setting $j_s$ as the beginning of the current round and $j_e$ as the beginning of the next round. Consider the following definition. 
\begin{definition}
    Define the set of jobs $\sigma^{t}_{K}$ as below: 
    \begin{itemize}
    \item job $j^t_s$ is release at time $t$ and rigid;
    \item job $j^t_k$ is released at time $t+\frac{k}{K}$ with starting deadline $t+k+1$ for $k = 1,2,\ldots,K$, where $K$ is a large integer; 
    \end{itemize}
\end{definition}

For each $r = 0,1,\ldots,R$, let $T_r$ denote $\sum_{r' = 1,2,\ldots,r} \Delta_{r'}$ where $\Delta_{1}, \Delta_2,\ldots,\Delta_R$ are i.i.d. random variables that take values in $\{2,3,\ldots,K+1\}$. 
Consider the oblivious adversary $X_{K,R} := \cup_{r = 0,1,\ldots,R} \sigma^{T_r}_{K}$, where $R$ and $K$ are large integers. 
Clearly, the corresponding finite space $\mathcal{X}^*_{K,R}$ is 
\begin{equation*}
    \left\{\cup_{r = 0,1,\ldots,R} \sigma^{\sum_{r' = 1,2,\ldots,r} \delta_{r'}}_{K}: \delta_1,\delta_2,\ldots,\delta_R\in \{2,3,\ldots,K+1\}\right\}. 
\end{equation*}
See Figure \ref{fig:lb_rand} for a job instance when $K = 4$. 

\begin{figure}[h]
  \centering
  \includegraphics[width=0.7\linewidth]{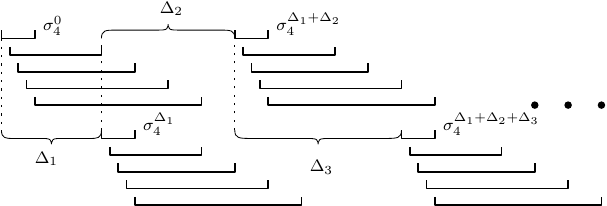}
  \caption{When $K = 4$, $\Delta_1, \Delta_2, \Delta_3$ take values of $3,4,5$ respectively}
  \label{fig:lb_rand}
\end{figure}

Now, we start its analysis. 
As the analysis for the baby instance, we first determine the set of ``best" algorithms.

\begin{lemma} \label{lem:best_alg}
    For any deterministic online algorithm $a$, there exists a deterministic online algorithm $a^*$ with parameters $\Gamma_0, \Gamma_1, \ldots, \Gamma_R$ such that 
\begin{itemize}
    \item $\Gamma_0$ is a constant in $\left\{\frac{1}{K}, \frac{2}{K}, \ldots, 1\right\}$;  
    \item $\Gamma_r$ that takes values in $\left\{\frac{1}{K}, \frac{2}{K}, \ldots, 1\right\}$ is determined by $\Delta_1, \ldots,\Delta_r$, for each $r= 1,2,\ldots,R$; 
    \item $\Gamma_R = \frac{1}{K}$; 
    \item $a^*$ starts jobs $j^{T_r}_1, j^{T_r}_2, \ldots, j^{T_r}_{\Gamma_r\cdot K-1}$ at time $r_{j^{T_r}_{\Gamma_r\cdot K-1}}$, for each $r= 0,1,\ldots,R$; 
    \item $a^*$ starts jobs $j^{T_r}_{\Gamma_r\cdot K}, j^{T_r}_{\Gamma_r\cdot K+1}, \ldots, j^{T_r}_{K}$ at time\\
    $\min\{T_{r+1}, s_{j^{T_r}_{\Gamma_r\cdot K}}\}$, for each $r= 0,1,\ldots,R$; 
    \item $a^*(x) \leq a(x)$ for each $x\in \mathcal{X}^*_{K,R}$; 
\end{itemize}
\end{lemma}

The following Lemma follows from Lemma \ref{lem:best_alg}. 
\begin{lemma} \label{lem:finite}
    There exists a finite set $\mathcal{A}^*_{K,R}$ such that each deterministic online algorithm $a\in \mathcal{A}$ is mapped to $a^*$ for some $a^*\in \mathcal{A}^*_{K,R}$. 
\end{lemma}

\begin{lemma} \label{lem:comp_1}
    If $\p(\Delta_r = \delta_r) = \frac{1}{K}$ for each $\delta_r = 2,3,\ldots,K+1$, then $\e_{X_{K,R}}[a^*(X_{K,R})] = 2 + R\cdot \left(2 - \frac{1}{K}\right)$ for any $a^*\in \mathcal{A}^*_{K,R}$, and $\e_{X_{K,R}}[\opt(X_{K,R})] = 2 + R\cdot\frac{3K-1}{2K}$. 
\end{lemma}

\begin{theorem} \label{thm:lb_1}
    The competitiveness of any randomized online algorithm is at least $\frac{4}{3}$. 
\end{theorem}

\subsection{A better distribution} \label{sec:lb2}

In the previous subsection, we choose the uniform distribution for $\Delta_r$ to make the expected algorithm cost constant for each $a^*\in \mathcal{A}^*_{K,R}$. But actually a modified uniform distribution putting a significant mass at $\Delta_r = 2$ gives a slightly better lower bound of the competitiveness. 

\begin{lemma} \label{lem:comp_2}
    If $\Delta_r = \begin{cases}
    2\quad \text{w.p.} \quad \frac{1+K-\tilde{K}}{K}\\
    \delta_r\quad \text{w.p.} \quad \frac{1}{K} \quad \forall\delta_r = 3,4,\ldots,\tilde{K}+1
\end{cases}$, where $\tilde{K} \in \{K,K-1,\ldots,2\}$, then $\e_{X_{K,R}}[a^*(X_{K,R})] \geq 2 + R\cdot \left(1 + \frac{\tilde{K}-1}{K}\right)$ for any $a^*\in \mathcal{A}^*_{K,R}$, and $\e_{X_{K,R}}[\opt(X_{K,R})] = 2 + R\cdot \left(1 + \frac{\tilde{K}(\tilde{K}-1)}{2K^2}\right)$. 
\end{lemma}

\begin{theorem} \label{thm:low_best}
    The competitiveness of any randomized online algorithm is at least $\frac{1+\sqrt{3}}{2}\approx 1.366$. 
\end{theorem}

\section{Algorithm with restarts} \label{sec:restart}

This section studies deterministic online algorithms that restarts jobs, whose definition will be provided shortly. We are interested in online algorithms with restarts not only because the ability of restarting jobs can improve the algorithm's competitiveness, but also because the lower bound presented in Section \ref{sec:lb1} and \ref{sec:lb2} also applies for any randomized online algorithm that restarts jobs. The goal of this section is to obtain the tight competitiveness of the setting for all deterministic online algorithms that restart jobs.

\subsection{Motivation}

So far, it is not known whether the algorithm $\db^\gamma$ gives the best possible competitiveness. To close the gap between $\frac{1+\sqrt{3}}{2}$ and $\frac{1}{\ln{2}}$, one gives a better lower bound of competitiveness for all algorithms, or, proposes an algorithm that has a lower competitiveness, or even efforts from both sides. 
Unfortunately, we did not manage to achieve any of the above. 
Then, we work around the existing lower bound of $\frac{1+\sqrt{3}}{2}$ by making the compromise that online algorithms are allowed to {\it restart} jobs. 
Precisely, we have the following definition. 
\begin{definition} [Restart a Job]
    Let $a_r$ denote a deterministic online algorithm that can restart jobs. 
    If $a_r$ is processing job $j$ at time $t$ such that $t\in [r_j,s_j]$, then $a_r$ can stop processing job $j$ at time $t$ and restart the {\it whole} processing of job $j$ at any time $t'\in [t,s_j]$. 
\end{definition}

We present a simple example (see Figure \ref{fig:restart_demo}) to show how restarting jobs works. 
Consider an adversary releasing the following four jobs: 
$\hat{j}$ is rigid and released at time $0$;
$j_1$ is released at time $1/3$ such that $s_{j_1} = 3$;
$j_2$ is released at time $2/3$ such that $s_{j_2} = 3$;
$\hat{j}'$ is rigid and released at time $3$. 
Consider $\bat^{1/2}$'s scheduling. 
By definition, $j_1$ is started at its release time and $j_2$ is started at its starting deadline. 
Then, consider $\bat^{1/2}$'s scheduling plus the ability of restarting jobs. 
By definition, $\bat^{1/2}$ starts job $j_1$ at time $1/3$. 
At time $2/3$, the online scheduler observes that $j_1$ should be started at the same time as $j_2$, since $[r_{j_1},s_{j_1}]\supset [r_{j_2},s_{j_2}]$ for minimizing the algorithm's cost. In other words, the online scheduler realizes that $j_1$ is redundant when and only when $j_2$ is released. 
Due to the observation, the online scheduler stops $j_1$'s processing at time $2/3$ and restart $j_1$ at the same time as $j_2$ at time $3$. 
Restarting job $j_1$ saves a cost of $1/3$ for $\bat^{1/2}$. 
\begin{figure}[h]
    \centering
    \includegraphics[width=0.5\linewidth]{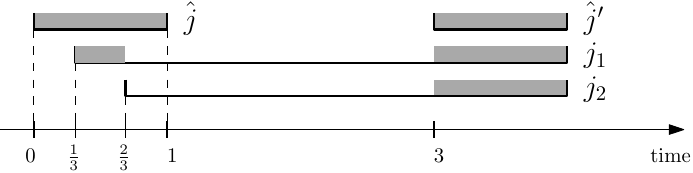}
    \caption{Restart job $j_1$}
    \label{fig:restart_demo}
\end{figure}

An easy observation is that $X_{K,R}$ as defined in Section \ref{sec:lb1} and \ref{sec:lb2} also works for algorithms that restart jobs. 
Precisely, in the following lemma, we show that for scheduling the randomized job instance $X_{K,R}$, the ability of restarting jobs is useless. 
\begin{lemma} \label{lem:best_alg_restart}
    Suppose $a_r$ is any deterministic online algorithm that can restart jobs. There exists a deterministic online algorithm $a$ that never restarts any job such that $a(x) \leq a_r(x)$ for each $x\in \mathcal{X}^*_{K,R}$. 
\end{lemma}

By Lemma \ref{lem:best_alg_restart} and \ref{lem:best_alg}, the lower bounds in Section \ref{sec:lb1} and \ref{sec:lb2} also applies to any randomized online algorithm that restarts jobs. Then, Theorem \ref{thm:low_best} is immediately improved as below. 
\begin{theorem}
    The competitiveness of any randomized online algorithm that restarts jobs is at least $\frac{1+\sqrt{3}}{2}\approx 1.366$.
\end{theorem}

\subsection{Adversary against deterministic online algorithm that restarts jobs} \label{sec:lb_restart}

Lemma \ref{lem:best_alg_restart} says the lower bound in Section \ref{sec:lb1} and \ref{sec:lb2} should give some insight into how online algorithms that restart jobs may fail to achieve the offline optimal cost. 
Indeed, we can simply adapt the same lower bound to the deterministic online setting. 

Let $\phi$ be a real number between $0$ and $1$ to be adjusted later. 
Take any deterministic online algorithm that restarts jobs, denoted by $a_r$. 
By Lemma \ref{lem:best_alg_restart} and \ref{lem:best_alg}, there exists $a^*$ parametrized by $\Gamma_0, \Gamma_1,\ldots,\Gamma_R$ as described in Lemma \ref{lem:best_alg} such that $a^*(x) \leq a_r(x)$ for any $x\in \mathcal{X}_{K,R}$. 
Note that in the deterministic online setting, $\Delta_{r+1}$ is a function of $\Gamma_0, \Gamma_1,\ldots,\Gamma_{r}$ for each $r = 0,1,\ldots,R-1$. 
Consider the adversary such that $\Delta_{r+1} = \begin{cases}
    2 \quad \text{if~} \Gamma_r > \phi \\
    s_{j^{T_r}_{\Gamma_r\cdot K + 1}} - T_r = \Gamma_r\cdot K + 2 \quad \text{if~} \Gamma_r \leq \phi
\end{cases}$ for each $r = 0,1,\ldots,R-1$. 

\begin{lemma} \label{lem:restart_lb}
    The competitiveness of $a^*$ is at least $c\cdot \frac{R}{R + 2} -  \frac{1}{K}$, where $c = \min\left\{1 + \phi, \frac{2+\phi}{1+\phi}\right\}$ and $K,R$ are large integers. 
\end{lemma}

\begin{theorem} \label{thm:restart_lb}
    The competitiveness of any deterministic online algorithm that restarts jobs is at least $\frac{\sqrt{5}+1}{2}$. 
\end{theorem}

\subsection[Best Batch]{Best Batch $\bb^\gamma$}
In this subsection, we present a deterministic online algorithm that restarts jobs, whose competitiveness is $\frac{\sqrt{5}+1}{2}\approx 1.618$ matching the lower bound in Theorem \ref{thm:restart_lb}. 

Without loss of generality, we use the following assumption for simplicity, because an online/offline scheduling algorithm is interesting only if it makes the same schedule for jobs that have the same window. 
\begin{assumption} \label{ass:windows}
    For each pair of distinct jobs $j_1,j_2$ in the input job instance $x\in \mathcal{X}$, we have $[r_{j_1}, d_{j_1}]\neq [r_{j_2},d_{j_2}]$. 
\end{assumption}

It is not hard to see how restarting jobs helps decrease the cost. Figure \ref{fig:restart_demo} is a simple example of doing it. 
To exploit it, consider the following typical situation. 
Upon some flag job $f$ hits its starting deadline, job $f$ is processed during $[s_f, d_f]$. 
All the jobs that are available to be started are started at time $s_f$ for free. 
It remains to decide the scheduling of some job released after time $s_f$. 
Typically, the set of jobs under consideration is $\mathcal{J}^f:=\{j: r_j\in (s_f, d_f]\}$. 
Observe that if jobs $j_1,j_2\in \mathcal{J}^f$ satisfy $[r_{j_1}, d_{j_1}]\supset [r_{j_2},d_{j_2}]$, then an offline scheduling can always process job $j_1$ during the same time as $j_2$, which makes $j_1$ redundant. 
In online settings, the redundancy of each job $j\in \mathcal{J}^f$ cannot be determined at its release time $r_j$, because at time $r_{j_1}$, the online scheduler should not know whether such job $j_2$ will be released (typically $r_{j_1} < r_{j_2}$) in the future and hence should not know whether job $j_1$ is redundant or not at time $r_{j_1}$. 
When the scheduler can restart jobs, a good strategy is to start all the jobs in $\mathcal{J}^f$ at their release times respectively, and then at time $d_f$ the scheduler decides whether to stop job $j$ at time $d_f$ for each $j\in \mathcal{J}^f$. 
Precisely, we have the following typical options. 
Let $\tilde{\mathcal{J}}^f:=\left\{j\in \mathcal{J}^f: \neg([r_{j'},d_{j'}]\subset [r_j,d_j]), \forall j'\in \mathcal{J}^f\setminus \{j\}\right\}$ denote the set of {\it critical} jobs excluding the redundant jobs. 
Then, $\tilde{\mathcal{J}}^f$ can be rewritten as $\{j^f_1,j^f_2,\ldots,j^f_K\}$ such that $r_{j^f_1} < r_{j^f_2} < \cdots < r_{j^f_K}$ and $d_{j^f_1} < d_{j^f_2} < \cdots < d_{j^f_K}$, where such a set of jobs is said to be {\it agreeable}. 
Each $k = 1,2,\ldots,K+1$ corresponds to one option for algorithm to stop jobs $\{j\in \mathcal{J}^f: r_j > r_{j^f_{k-1}}\}$ at time $d_f$, where $r_{j^f_{0}} := s_f$, and then the jobs that are stopped are delayed. 

Now, we describe more about the core ideas behind the algorithm design. 
Assume that all jobs are very flexible (with long windows), because in the case of scheduling jobs with short windows, any feasible scheduling may not be too far from the optimal, and a good scheduling should be easier to obtain intuitively. 
The first observation is that a good algorithm has to stop jobs $\{j\in \mathcal{J}^f: r_j > s_f + 0.618\}$. 
Otherwise, one may have the ratio of the local algorithm cost to the local optimal cost strictly greater than the golden ratio in the following situation. 
There may be a job $f'$ released after $d_f+1$ but with starting deadline before the starting deadline of any job in $\mathcal{J}^f$. 
In this case, all the jobs $\mathcal{J}^f$ becomes redundant because $[r_{f'},d_{f'}]\subset [r_j,d_j]$ for each $j \in \mathcal{J}^f$. 
Consequently, $f'$ becomes the next flag job and the local optimal cost around flag job $f$ is $1$ that is due to flag job $f$ alone. 
Therefore, the local competitiveness is equal to the local algorithm's cost around flag job $f$ which is determined by which jobs are to be stopped/delayed. 
To maintain a competitiveness of $1.618$, one does not want to schedule any job in $\{j\in \mathcal{J}^f: r_j > s_f + 0.618\}$ together with $f$.  
It follows that the ``best" $k$ has to be chosen among $k = 1,2,\ldots,K^f$ where $K^f = \min\{k: r_{j^f_k} > s_f + 0.618\}$. 
Note that intuitively, the more general case is when $r_{j^f_K} > s_f + 0.618$ where $j^f_K$ is the last critical job among $\mathcal{J}^f$, so one may assume the existence of $K^f$ for the current discussion.

Next, we motivate how to choose the ``best" $k$. 
The consequence of stopping jobs $\{j\in \mathcal{J}^f: r_j > r_{j^f_{k-1}}\}$ is that jobs $f,j^f_1,\ldots,j^f_{k-1}$ are scheduled within one block during $[s_d, 1 + r_{j^f_{k-1}}]$ and jobs $j^f_k, \ldots, j^f_{K}$ are delayed. 
A straightforward question is: why not fix the ``best" $k$ to be $K^f$? 
Consider one extreme case when $s_{j^f_1}\approx s_{j^f_2}\approx \cdots \approx s_{j^f_{K^f}}$. 
Intuitively, the optimal scheduling should process $j^f_1,j^f_2,\ldots,j^f_{K^f}$ in the same block, because $j^f_1,\ldots,j^f_{{K^f}-1}$ can overlap almost fully with $j^f_{K^f}$ no matter where $j^f_{K^f}$ is scheduled.
By the previous observation, since $j^f_{K^f}$ has to be delayed ($r_{j^f_{K^f}} > s_f + 0.618$), all the jobs $j^f_1,j^f_2,\ldots,j^f_{K^f}$ should be delayed. 
In this case, the ``best" $k$ is $1$ and all the jobs are stopped at time $d_f$ and delayed. 
Intuitively, there are cases when the ``best" $k$ is not as low as $1$.
Consider another extreme case when $s_{j^f_1} \ll s_{j^f_2} \ll \cdots \ll s_{j^f_{K^f}}$. 
Note that $j^f_{K^f}$ is definitely delayed. 
We consider the case $j^f_{K^f}$ is processed during its rightmost position $[s_{j^f_{K^f}},d_{j^f_{K^f}}]$. 
Observe that any job to be processed strictly right to $[s_{j^f_{K^f}},d_{j^f_{K^f}}]$ (e.g. the flag job next to $j^f_{K^f}$) must be released after $s_{j^f_{K^f}}$.
We can show that in this case $j^f_1, \ldots, j^f_{K^f-1}$ or even $j^f_{K^f}$ make a positive contribution to the local optimal cost around $f$. 
Precisely, the local optimal cost around $f$ should be at least $r_{j^f_{K^f-1}} + 1 - s_f$ significantly more than $1$ because of the additional jobs $j^f_1, \ldots, j^f_{K^f-1}$. 
In this case, the ``best" $k$ is $K^f$ and only jobs released after $r_{j^f_{K^f-1}}$ are stopped at time $d_f$ and delayed. 

However, it still remains to think about how to prove the above idea works in the view of the worst case analysis. 
After staring at algorithm $a^*$ described in Lemma \ref{lem:best_alg}, we make another important observation that helps to obtain the eventual algorithm design. 
Note that for each round $r$, job $\hat{j}^r$ is rigid and becomes a flag job. 
After $a^*$ decides the value of $\Gamma_r$, the processing of $\hat{j}^r$ gets a $\Gamma_r$-extension, that is one block, and the remaining jobs released after time $s_{\hat{j}^r} + \Gamma_r$ are started at the same time, that is the next block. 
In the adversary described in section \ref{sec:lb_restart}, the first block may have any length between $1$ and $2$ depending on the choices of the scheduler (where $1.618$ is a threshold), meanwhile, the second block either does not exist (when $\Gamma_r> 0.618$) or has length of $1$ (when $\Gamma_r\leq 0.618$). 
An important observation is that the two blocks described above are in distinct natures. 
It follows that our algorithm design should deal with the block accommodating $f,j^f_1,\ldots,j^f_{k-1}$ and the block accommodating $j^f_k, \ldots, j^f_K$ differently. 
In particular, after one determines the ``best" value for $k$, when job $j^f_k$ hits its starting deadline and hence is started at time $s_{j^f_k}$, one needs to determine how to schedule the set of jobs released during $(s_{j^f_k}, d_{j^f_k}]$ in some manner different from the scheduling of $f$ and $\mathcal{J}^f$. 

Let $\gamma\in (0,1)$ be a parameter to be adjusted later. Consider the online scheduler that restarts jobs called Best Batch ($\bb^\gamma$) as described in Algorithm \ref{alg:restart}.

\begin{algorithm} 
\caption{Best Batch: $\bb^\gamma$}
\label{alg:restart}
\algrenewcommand\algorithmicprocedure{\textbf{Upon}}
\begin{algorithmic}[1]
\State $\hat{f}\gets NULL$ and $s_{\hat{f}} \gets \infty$; \Comment{initialize the value of variable $\hat{f}$}
\State $p\gets NULL$; \Comment{initialize the value of variable $p$}
\Procedure{a job $j$ is released}{} \label{alg_line:release_begin}
    \If{$r_j \in (s_p, d_p]$}
        \State Start $j$ now (due to $p$); \label{alg_line:start_primary}
    \ElsIf{$r_j \in (d_p, r_{j^p_{K^p}} + 1]$ and $p$ is tight}
        \State Start $j$ now (due to $p$); \label{alg_line:start_tight}
    \ElsIf{$r_j \in (s_{j^p_{k^p}}, s_{j^p_{k^p}} + r_{j^p_{K^p}} - r_{j^p_{k^p-1}}]$ and $j^p_{k^p}$ is a secondary flag job} \Comment{where $r_{j^p_{0}} = s_p$}
        \State Start $j$ now (due to $j^p_{k^p}$); \label{alg_line:start_secondary}
    \ElsIf{$s_j \leq s_{\hat{f}}$}
        \State ${\hat{f}}\gets j$; \label{alg_line:f_gets_1}
    \EndIf
\EndProcedure \label{alg_line:release_end}
\Procedure{time $s_{\hat{f}}$}{} \label{alg_line:flag_begin}
    \State ${\hat{f}}$ is marked as a flag job;  \label{alg_line:flag_mark}
    \If{${\hat{f}}$ is not marked as secondary}
        \State ${\hat{f}}$ is marked as \textit{primary}; \label{alg_line:primary_mark}
        \State $p\gets {\hat{f}}$; 
    \EndIf
    \State Start all the available jobs now (due to ${\hat{f}}$); \label{alg_line:flag_start} 
    \Comment{including any job released at time $s_{\hat{f}}$}
    \State $\hat{f}\gets NULL$ and $s_{\hat{f}} \gets \infty$; \label{alg_line:f_gets_3}
\EndProcedure \label{alg_line:flag_end}
\Procedure{time $d_p$}{}  \label{alg_line:primary_begin}
    \State $\mathcal{J}^p\gets$ the set of all the jobs released during $(s_p, d_p]$; \label{alg_line:A_set_begin}
    \State $\{j^p_1,\ldots,j^p_{K}\}\gets$ the set of \textit{critical} jobs of $\mathcal{J}^p$ such that $r_{j^p_1} < \cdots < r_{j^p_K}$; \label{alg_line:critical_jobs} 
    \Comment{Precisely, $\{j^p_1,\ldots,j^p_{K}\} = \left\{j\in \mathcal{J}^p: \neg([r_{j'},d_{j'}]\subset [r_j,d_j]), \forall j'\in \mathcal{J}^p\setminus \{j\}\right\}$ }
    \If{$r_{j^p_K} > s_p + \gamma$} \label{alg_line:critical_case}
        \State $K^p\gets \min \{k=1,\ldots,K: r_{j^p_k} > s_p + \gamma\}$; \label{alg_line:A_set_end}
        \For {$k = 1,2,\ldots,K^p$}
            \State $C^p_{k}\gets \begin{cases}
                (1 + r_{j^p_{k-1}} - s_p) + (d_{j^p_{K^p}} - s_{j^p_k}) \,\text{~if~} 1 + r_{j^p_{k-1}} < s_{j^p_k}; \\
                d_{j^p_{K^p}} - s_{p} \quad \text{~if~} 1 + r_{j^p_{k-1}} \geq s_{j^p_k}; 
            \end{cases}$ \label{alg_line:key} 
        \EndFor
        \If{$\min_{k = 1,2,\ldots,K^p} C^p_{k} = d_{j^p_{K^p}} - s_p$} \label{alg_line:tight}
            \State Mark $p$ as \textit{tight}; 
        \Else \label{alg_line:separated}
            \State $k^p\gets \argmin_{k = 1,2,\ldots,K^p} C^p_{k}$ (break the ties arbitrarily); 
            \State $j^p_{k^p}$ is marked as \textit{secondary} attached to $p$; 
            \State Stop jobs $\{j\in \mathcal{J}^p: r_j > r_{j^p_{k^p-1}}\}$ (due to $p$); \Comment{where $r_{j^p_{0}} = s_p$}\label{alg_line:stop}
            \If{$s_{j^p_{k^p}} \leq s_{\hat{f}}$} 
                \State ${\hat{f}}\gets {j^p_{k^p}}$; \label{alg_line:f_gets_2}
            \EndIf
        \EndIf
    \EndIf
\EndProcedure \label{alg_line:primary_end}
\end{algorithmic}
\end{algorithm}

Now, we do its analysis. 
It seems the algorithm is so delicate that proofs have to be done by induction and contradiction. 

\begin{lemma} \label{lem:hat_f}
    The value of variable $\hat{f}$ at any time instant $t$ is the job that has the earliest starting deadline among all the jobs that are available to be started at time $t$. 
\end{lemma}

\begin{lemma} \label{lem:feasible}
    Algorithm $\bb^\gamma$ is a well defined online scheduler in the following sense: 
    \begin{itemize}
        \item [(i)] if job $j$ becomes available to be started at time $t$, job $j$ must be (re)-started during $[t,s_j]$; 
        \item [(ii)] if job $j$ is stopped at time $t$, then $t < s_j$; 
    \end{itemize}
\end{lemma}

Observe that all the jobs marked as primary are also marked as flag jobs (line \ref{alg_line:flag_begin}-\ref{alg_line:flag_end}). 
Observe that all the flag jobs are in the time order of being marked as a flag job. 
Let $\mathcal{P}$ denote the set of all primary flag jobs and rewrite $\mathcal{P}$ as $p_1,p_2,\ldots,p_{|\mathcal{P}|}$ such that $p_x$ is the $x$-th job marked as a primary flag job. 
It follows that $s_{p_1} \leq s_{p_2} \leq \cdots \leq s_{p_{|\mathcal{P}|}}$ since each flag job is marked as a flag job at its starting deadline. 
We have the following set of observations about the scheduling made by $\bb^\gamma$. 

\begin{lemma} \label{lem:agreeable}
Take any $x = 1,2,\ldots,|\mathcal{P}|$. The following statements hold: 
\begin{enumerate}
    \item [(i)] if line \ref{alg_line:critical_case} does not hold for $p_x$, then $d_{p_x} < r_{p_{x+1}}$. 
    \item [(ii)] if line \ref{alg_line:critical_case} holds for $p_x$ and $p_x$ is tight (line \ref{alg_line:tight}), then $1 + r_{j^{p_x}_{K^{p_x}}} < r_{p_{x+1}}$. 
    \item [(iii)] if line \ref{alg_line:critical_case} holds for $p_x$, $p_x$ is not tight (line \ref{alg_line:separated}), and the secondary job $j^{p_x}_{k^{p_x}}$ attached to $p_x$ is not a flag job, then the flag job right next to $p_x$ is $p_{x+1}$, and $d_{p_x} < r_{p_{x+1}}$. 
    \item [(iv)] if line \ref{alg_line:critical_case} holds for $p_x$, $p_x$ is not tight (line \ref{alg_line:separated}), and the secondary job $j^{p_x}_{k^{p_x}}$ attached to $p_x$ is a flag job, then the flag job right next to $p_x$ is its secondary job $j^{p_x}_{k^{p_x}}$, $j^{p_x}_{k^{p_x}}$ is completed at its deadline $d_{j^{p_x}_{k^{p_x}}}$, and $s_{j^{p_x}_{k^{p_x}}} + (r_{j^{p_x}_{K^{p_x}}} - r_{j^{p_x}_{k^{p_x}-1}}) < r_{p_{x+1}}$. 
    \item [(v)] no job is stopped during $[r_{p_x}, d_{p_x})$. 
    \item [(vi)] any job released before or at time $s_{p_x}$ is completed before or at time $d_{p_x}$. 
    \item [(vii)] any job released at or after $r_{p_x}$, including $p_x, p_{x+1},\ldots,p_{|\mathcal{P}|}$ and their secondary jobs (if exists), cannot be started due to any of the jobs $p_1,p_2,\ldots,p_{x-1}$ or their secondary flag jobs (if exists). 
    \item [(viii)] primary job $p_x$ is never started until its starting deadline and then processed during $[s_{p_x},d_{p_x}]$. 
\end{enumerate} 
\end{lemma}

\begin{lemma} \label{lem:not_close}\quad
    \begin{itemize}
        \item [(i)] If $p$ is a primary tight job, then $1 + r_{j^{p}_{K^{p}}} > s_{j^p_{K^p}}$. 
        \item [(ii)] If $j^p_{k^p}$ is a secondary flag job attached to $p$, then $s_{j^p_{k^p}} + (r_{j^{p}_{K^p}} - r_{j^p_{k^p-1}}) > \max\{s_{j^p_{K^p}}, 1 + r_{j^{p}_{K^{p}}}\}$. 
    \end{itemize}
\end{lemma}

\begin{lemma} \label{lem:alg_cost}
    Take any primary flag job $p_x$ for $x = 1,2,\ldots,|\mathcal{P}|$. 
\begin{itemize}
    \item [(i)] If line \ref{alg_line:critical_case} does not hold for $p_x$, then the algorithm's cost during $[s_{p_x}, s_{p_{x+1}}]$ is at most $1 + \gamma$. 
    \item [(ii)] If line \ref{alg_line:critical_case} holds for $p_x$ and $p_x$ is tight (line \ref{alg_line:tight}), then the algorithm's cost during $[s_{p_x}, s_{p_{x+1}}]$ is at most $2 + (r_{j^{p_x}_{K^{p_x}}} - s_{p_x})$. 
    \item [(iii)] If line \ref{alg_line:critical_case} holds for $p_x$, $p_x$ is not tight (line \ref{alg_line:separated}), and the secondary job $j^{p_x}_{k^{p_x}}$ attached to $p_x$ does not become a flag job, then the algorithm's cost during $[s_{p_x}, s_{p_{x+1}}]$ is at most $1 + \gamma$. 
    \item [(iv)] If line \ref{alg_line:critical_case} holds for $p_x$, $p_x$ is not tight, and the secondary job $j^{p_x}_{k^{p_x}}$ attached to $p_x$ becomes a flag job, then the algorithm's cost during $[s_{p_x}, s_{p_{x+1}}]$ is at most $2 + r_{j^{p_x}_{K^{p_x}}} - s_{p_x}$. 
\end{itemize}
\end{lemma}

The following definition points out which jobs are taken for constructing a lower bound of the optimal cost. 
\begin{definition}
    For each $p\in \mathcal{P}$, let $\mathcal{A}(p)$ denote the set of critical jobs $\{j^{p}_1,j^{p}_2,\ldots,j^{p}_{K^{p}}\}$ as described in line \ref{alg_line:A_set_begin}-\ref{alg_line:A_set_end} if either the primary job $p$ is tight (line \ref{alg_line:tight}), or, there is a secondary flag job $j^{p}_{k^{p}}$ attached to $p$. 
    Let $\mathcal{A}(p)$ be empty otherwise. 
\end{definition}

\begin{lemma} \label{lem:opt}
    There exists an optimal scheduling, denoted by $O$, of $\cup_{x = 1,2,\ldots, |\mathcal{P}|} \{p_x\}\cup \mathcal{A}(p_x)$ such that $O$ has exactly $|\mathcal{P}|$ blocks, ordered from left to right, and the the set of jobs accommodated within the $x$-th block is exactly $\{p_x\}\cup \mathcal{A}(p_x)$ for $x = 1,2,\ldots, |\mathcal{P}|$. 
\end{lemma}

\begin{theorem} \label{thm:restart}
    When $\gamma = \frac{\sqrt{5}-1}{2}$, the competitiveness of $\bb^\gamma$ is $\frac{\sqrt{5}+1}{2}$. 
\end{theorem}
\begin{proof}
    Consider the following matching between the algorithm's cost and the optimal cost: 
    algorithm's cost during $[s_{p_x}, s_{p_{x+1}}]$ corresponds to the cost of the $x$-th block in the optimal scheduling as described in Lemma \ref{lem:opt}. 
    After a case by case check between Lemma \ref{lem:alg_cost} and \ref{lem:opt}, $1+\gamma$ in the algorithm's cost is matched to $1$ in the optimal cost, and $2 + (r_{j^{p_x}_{K^{p_x}}} - s_{p_x})$ in the algorithm's cost is matched to $1 + (r_{j^{p_x}_{K^{p_x}}} - s_{p_x})$ in the optimal cost. 
    Note that $r_{j^{p_x}_{K^{p_x}}} - s_{p_x} > \gamma$. 
    Since $\gamma$ is chosen to be $ \frac{\sqrt{5}-1}{2}$, we have $\frac{2 + (r_{j^{p_x}_{K^{p_x}}} - s_{p_x})}{1 + (r_{j^{p_x}_{K^{p_x}}} - s_{p_x})} < \frac{2 + \gamma}{1 + \gamma} = \frac{1 + \gamma}{1}$. 
    The local ratio leads to the global ratio, i.e., the competitiveness of $\bb^\gamma$ is $1 + \gamma = \frac{\sqrt{5}+1}{2}$. 
\end{proof}

At last, we show $\bb^\gamma$ is linearly fast given that jobs are sorted.

\begin{lemma} \label{lem:critical_jobs}
    If $\mathcal{J}^p$ is sorted and stored in a doubly linked list $L$ such that $r_{L(l_1)} \leq r_{L(l_2)}$ for any $1 \leq l_1 < l_2 \leq |\mathcal{J}^p|$, then it takes $O(|\mathcal{J}^p|)$ time to output the critical jobs in line \ref{alg_line:critical_jobs}. 
\end{lemma}

\begin{theorem} \label{thm:restart_time}
    If the input job instance of size $n$ is sorted and stored in a doubly linked list $L$ such that $r_{L(l_1)} \leq r_{L(l_2)}$ for any $1 \leq l_1 < l_2 \leq n$, then $\bb^\gamma$ is a $\frac{\sqrt{5}+1}{2}$-approximation algorithm whose time complexity is $O(n)$ for the offline setting . 
\end{theorem}

{\bf Acknowledgements:} We are grateful to Sirie He and Gruia Calinescu for valuable discussions in the early stages of this work.

\onecolumn

\appendix 

\section[Lower bound of 2]{Lower bound of $2$}

\subsection{A correction} \label{sec:correction}

We noticed that Albers {\em et al.} \cite{albers2025onlinebusytimescheduling} attempted to give an lower bound of $2$ for any deterministic online algorithm for uniform job scheduling. However, their attempt of the adversary design and the analysis, in spite of being simpler than ours, still appears to be flawed. 
Albers {\em et al.}'s adversary can be summarized as below where $k$ is a large number: 
\begin{enumerate}
    \item initialize $i$ to be $1$ as the index of the current round; 
    \item release rigid job $f_1$ (corresponding to $\hat{j}_1$ in our adversary) with deadline at time $1$ and $t^*\gets r_{f_1}$; 
    \item release job $j$ at time $t^* + \epsilon$ with deadline at time $3i$ (corresponding to $j^i_k$ in our adversary); 
    \item do nothing until job $j$ is started, and $t^*\gets$ when $j$ is started; 
    \item if $t^*\in [r_{f_i}, d_{f_i})$, then release another job at time $t^* + \epsilon$ with deadline at time $3i$ (corresponding to $j^i_{\floor{(t^*-t_i)\cdot K}+1}$), $j\gets$ the new released job, and repeat 4-6;  
    \item if $t^*\geq d_{f_i}$, then mark job $j$ as $f_{i+1}$, $i\gets i+1$, and repeat 3-6 unless $i = k$; 
\end{enumerate}
Notice that the difference between the above adversary and ours is that even if $t^*\in [d_{f_i}, d_{f_i+1}]$ (in item 6), job $j$ is still marked as $f_{i+1}$ by the above adversary, while in our adversary a rigid job $\hat{j}_{i+1}$ is released at time $d_{\hat{j}_i}+2$ as in line \ref{alg_line:correct_move} of Algorithm \ref{alg:det_2}.  

We present the following algorithm $a$ against the above adversary that has a competitiveness strictly below $2$. 
Start $f_1$ at its release time since it is rigid. 
Take any $r = 0,1,\ldots$ as the index of the $r$-th round. 
The below describes a round of scheduling jobs after $f_{3r+1}$ and before $f_{3(r+1)+1}$ including $f_{3(r+1)+1}$. 
Suppose $f_{3r+1}$ is started at time $t$. 
Adversary releases job $j$ at time $t + \epsilon$ with deadline at time $3(3r+1)$ by definition. 
Algorithm $a$ starts job $j$ at its starting deadline $d_j - 1 = 3(3r+1)-1$. 
Then, the adversary marks $j$ as $f_{3r+2}$. 
Then, the adversary releases $j'$ at time $d_{f_{3r+2}} - 1 + \epsilon$ with deadline at time $3(3r+2)$. 
Algorithm $a$ starts job $j'$ at time $d_{f_{3r+2}} - \epsilon$. 
By definition, the adversary releases another job $j''$ at time $d_{f_{3r+2}}$ with deadline at time $3(3r+2)$. 
Algorithm $a$ starts $j''$ at its release time $r_{j''} = d_{f_{3r+2}} = 3(3r+1)$. 
Then, the adversary marks $j''$ as $f_{3r+3}$. 
Then, the adversary releases job $j'''$ at time $r_{f_{3r+3}} + \epsilon = d_{f_{3r+2}} + \epsilon$ with deadline at time $3(3r+3)$. 
Algorithm $a$ starts job $j'''$ at its starting deadline $d_{j'''}-1 = 3(3r+3) - 1$, and hence the adversary marks $j'''$ as $f_{3(r+1)+1}$. 
See Figure \ref{fig:counter_example} for $a$'s scheduling in the first two rounds until $f_7$. 

\begin{figure}[h]
    \centering
    \includegraphics[width=\linewidth]{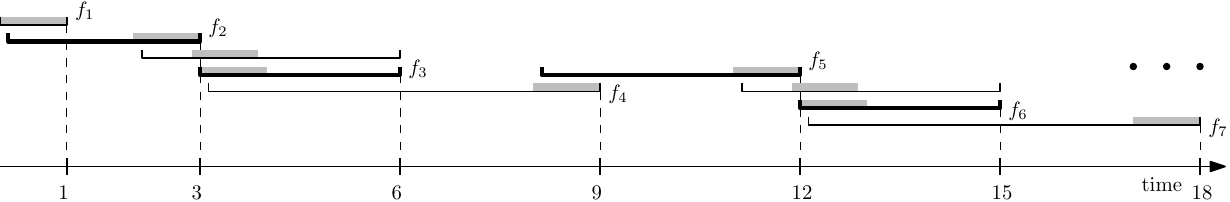}
    \caption{$a$'s scheduling}
    \label{fig:counter_example}
\end{figure}

Since the adversary has $k$ flag jobs, there are $\ceil{\frac{k}{3}}$ rounds. 
Take any $r = 0,1,\ldots,\ceil{\frac{k}{3}}-1$. 
The above scheduling of $a$ ensures that algorithm's cost during $[3(3r)-1, 3(3r+2)]$ is at most $3$ where the processing due to $f_{3r+1}$ has cost of $1$ and the processing of $f_{3r+2}$ and $f_{3r+3}$ has an extra cost of $2$. 
Then, the algorithm's cost is at most $3\ceil{\frac{k}{3}}$. 
Next, we show that the optimal cost is at least $2(\ceil{\frac{k}{3}}-1)$. 
Take jobs $\{f_{3r+2}, f_{3r+3}: r = 0,1,\ldots,\ceil{\frac{k}{3}}-2\}$. 
Note that $d_{f_{3r+2}} \leq r_{f_{3r+3}}$ and $d_{f_{3r+3}} = 3(3r+2) < 3(3r+3) - 1 + \epsilon = r_{f_{3(r+1)+2}}$. 
It follows that all the jobs in $\{f_{3r+2}, f_{3r+3}: r = 0,1,\ldots,\ceil{\frac{k}{3}}-2\}$ have mutually disjoint windows and hence the optimal cost is at least $2(\ceil{\frac{k}{3}}-1)$. 
Therefore, the competitiveness of algorithm $a$ is at most $\frac{3\ceil{\frac{k}{3}}}{2(\ceil{\frac{k}{3}}-1)}$ that drops below $2$ when $k\geq 13$. 

Here, we point out where we think Albers {\em et al.}'s analysis does not work. 
To achieve the lower bound of $2$, the algorithm's cost is represented as the total cost of components where a component is a maximal busy interval and all the maximal busy intervals are mutually disjoint. 
The components should be matched to the flag jobs one to one, otherwise one may not be able to reschedule the non-consecutive components into their adjacent components almost for free. 
After looking into the definitions of component and the adversary, the above statements cannot be achieved simultaneously. 

\section{Randomized online algorithm against oblivious adversary}



\subsection{The best coin} \label{sec:best_coin}

Next, we show that $\{0,1/2\}$ is the best pair of values for $\gamma$. Suppose algorithm flips an unfair coin to choose between $\{\gamma_1, \gamma_2\}$ such that $0\leq \gamma_1 < \gamma_2$. 

\begin{itemize}
    \item Case 1: $L \leq 1 + \gamma_1$. 
    Looking at $\bat^{\gamma_1}$, Lemma \ref{lem:I^-} and \ref{lem:key} imply $|\mathcal{F}^{\gamma_1}_i|\leq 1$, and hence $C^{\gamma_1}_i \leq 1 + \gamma_1 \leq (1+\gamma_1)\cdot L$. 
    Looking at $\bat^{\gamma_2}$, we have the same result:  $C^{\gamma_2}_i \leq 1 + \gamma_2 \leq (1+\gamma_2)\cdot L$. 
    \item Case 2: $1+\gamma_1 < L \leq 1+\gamma_2$. 
    Looking at $\bat^{\gamma_1}$, Lemma \ref{lem:I^-} and \ref{lem:key} imply $C^{\gamma_1}_i \leq L + 1 + \gamma_1 < 2\cdot L$. 
    Looking at $\bat^{\gamma_2}$, $|\mathcal{F}^{\gamma_2}_i|\leq 1$. We have $C^{\gamma_2}_i \leq 1 + \gamma_2 < \frac{1+\gamma_2}{1+\gamma_1}\cdot L$. 
    \item Case 3: $L > 1+\gamma_2$. 
    Looking at $\bat^{\gamma_1}$, Lemma \ref{lem:I^-} and \ref{lem:key} imply $C^{\gamma_1}_i \leq L + 1 + \gamma_1 < (1 + \frac{1+\gamma_1}{1+\gamma_2})\cdot L$. 
    Looking at $\bat^{\gamma_2}$, similarly, we have $C^{\gamma_2}_i \leq L + 1+\gamma_2 < 2\cdot L$. 
\end{itemize}

Let $\alpha\in (0,1)$ denote the probability of choosing $\bat^{\gamma_1}$. 
Then, $\alpha\cdot C^{\gamma_1}_i + (1-\alpha)\cdot C^{\gamma_2}_i \leq \mathbf{c}\cdot L$, where 
\begin{equation*}
    \mathbf{c} = \max\left\{\alpha\cdot (1+\gamma_1) + (1-\alpha)\cdot (1+\gamma_2), \alpha\cdot 2 + (1-\alpha)\cdot \frac{1+\gamma_2}{1+\gamma_1}, \alpha\cdot \left(1 + \frac{1+\gamma_1}{1+\gamma_2}\right) + (1-\alpha)\cdot 2\right\}. 
\end{equation*}

To obtain the best $\mathbf{c}$, consider solving the following optimization problem: 
\begin{eqnarray*}
    &\min& \quad \mathbf{c} \\
    &\text{s.t.}& \quad 0 \leq \alpha \leq 1 \\
    &\quad& \quad 0 \leq \gamma_1 \leq \gamma_2 \\
    &\quad& \quad \mathbf{c} \geq \alpha\cdot (1+\gamma_1) + (1-\alpha)\cdot (1+\gamma_1) \cdot \frac{1+\gamma_2}{1+\gamma_1} \\
    &\quad& \quad \mathbf{c} \geq \alpha\cdot 2 + (1-\alpha)\cdot \frac{1+\gamma_2}{1+\gamma_1} \\
    &\quad& \quad \mathbf{c} \geq \alpha\cdot \left(1 + \frac{1+\gamma_1}{1+\gamma_2}\right) + (1-\alpha)\cdot 2 
\end{eqnarray*}
Observe that the minimum is still achieved when $\gamma_1 = 0$. Then, let $\beta$ be $\frac{1+\gamma_2}{1+\gamma_1}$. Furthermore, the minimum is strictly below $2$ due to the previous subsection, which is achieved only when $\beta < 2$. Then, it suffices to solve the following optimization problem: 
\begin{eqnarray*}
    &\min& \quad \mathbf{c} \\
    &\text{s.t.}& \quad 0 \leq \alpha \leq 1 \\
    &\quad& \quad 1 < \beta < 2 \\
    &\quad& \quad \mathbf{c} \geq \alpha\cdot 2 + (1-\alpha)\cdot \beta = \beta + \alpha\cdot (2 - \beta) \\
    &\quad& \quad \mathbf{c} \geq \alpha\cdot (1 + \frac{1}{\beta}) + (1-\alpha)\cdot 2 = 2 + \alpha\cdot (\frac{1}{\beta}-1) 
\end{eqnarray*}
For a fixed $\beta$, the minimum is $\frac{3 - \beta}{3-\beta-\frac{1}{\beta}}$, achieved when $\alpha = \frac{2 - \beta}{3-\beta-\frac{1}{\beta}}$. 
Eventually, $\frac{3 - \beta}{3-\beta-\frac{1}{\beta}}$ achieves its minimum $\frac{9}{5}$ when $\beta = \frac{3}{2}$.


\section{Oblivious adversary against randomized online algorithm}

\subsection{Proof of Yao's principle}
\begin{proof} [Proof of Lemma \ref{lem:yao}]
    For simplicity, let $c^*$ denote $\frac{\min_{a^*\in \mathcal{A}^*} \e_X[a^*(X)]}{\e_X[\opt(X)]}$. 
    Take any randomized online algorithm $A$. 
    By Definition \ref{def:comp}, we show that if $\e_A[A(x)] \leq c \cdot \opt(x)$ for any input instance $x\in \mathcal{X}$, then $c \geq c^*$. 
    It suffices to show that there exists an input instance $x\in \mathcal{X}^*$ such that $\e_A[A(x)] \geq c^* \cdot \opt(x)$. 
    We show it by contradiction. 
    Assume $\e_A[A(x)] < c^* \cdot \opt(x)$ for any $x \in \mathcal{X}^*$. 
    Note that since $\mathcal{X}^*$ is finite, $\tilde{c}:= \max_{x\in \mathcal{X}^*} \frac{\e_A[A(x)]}{\opt(x)}$ exists and $\tilde{c} < c^*$. 
    Let $\mu$ and $\nu$ be the probability measures taken by random variable $X$ and $A$ respectively. 
    Since $l$ is a map from $\mathcal{A}$ to $\mathcal{A}^*$, random variable $l(A)$ has the probability measure $\nu^*(a^*) = \int_{l^{-1}(a^*)} \,\text{d}\nu(a)$ for each $a^* \in \mathcal{A}^*$. 
    Then, consider the following contradiction: 
    \begin{eqnarray}
        &\quad& \tilde{c}\cdot \e_X[\opt(X)] \nonumber\\
        &=& \e_X[\tilde{c} \cdot \opt(X)] = \int_{\mathcal{X}^*} \tilde{c} \cdot \opt(x) \,\text{d}\mu(x) \nonumber\\
        &\geq& \int_{\mathcal{X}^*}\e_A[A(x)] \,\text{d}\mu(x) = \int_{\mathcal{X}^*}\int_{\mathcal{A}} a(x) \,\text{d}\nu(a)\,\text{d}\mu(x) \label{eq:yao_1}\\
        &\geq& \int_{\mathcal{X}^*}\int_{\mathcal{A}} l(a)(x) \,\text{d}\nu(a)\,\text{d}\mu(x) = \int_{\mathcal{X}^*}\int_{\mathcal{A}^*} a^*(x) \,\text{d}\nu^*(a^*)\,\text{d}\mu(x) \label{eq:yao_2}\\
        &=& \int_{\mathcal{A}^*}\int_{\mathcal{X}^*} a^*(x) \,\text{d}\mu(x) \,\text{d}\nu^*(a^*) = \int_{\mathcal{A}^*} \e_X[a^*(X)] \,\text{d}\nu^*(a^*) \label{eq:yao_3}\\
        &\geq& \int_{\mathcal{A}^*} c^*\cdot \e_X[\opt(X)] \,\text{d}\nu^*(a^*) = c^*\cdot \e_X[\opt(X)], \label{eq:yao_4}
    \end{eqnarray}
    where the inequality in (\ref{eq:yao_1}) is by the definition of $\tilde{c}$, the inequality in (\ref{eq:yao_2}) is by the definition of $l$, the equality in (\ref{eq:yao_2}) is a change of variable, the first equality in (\ref{eq:yao_3}) is valid due to the finiteness of $\mathcal{X}^*$ and $\mathcal{A}^*$, and the inequality in (\ref{eq:yao_4}) is by the definition of $c^*$. 
\end{proof}

\subsection{Proofs of the oblivious adversary}

\begin{proof}[Proof of Lemma \ref{lem:best_alg}]

    Take any deterministic online algorithm $a$. 
    In the following, we define another deterministic online algorithm $a^*$ in Algorithm \ref{alg:a^*} while simulating $a$. 

\begin{algorithm} 
\caption{$a$ mapped to $a^*$}
\label{alg:a^*}
\algrenewcommand\algorithmicprocedure{\textbf{Upon}}
\small
\begin{algorithmic}[1]
\Procedure{time $T_r$ for some $r = 0,1,\ldots,R-1$}{}
    \If{$a$ starts some job during $[T_r+1, T_r+2)$}
        \State $a^*$ starts all the jobs that can be feasibly started at time $T_r+2$; 
    \ElsIf{$a$ starts some job during $[T_r, T_r+1)$}
        \State $T^*\gets$ the last time $a$ starts some job during $[T_r, T_r+1)$;
        \State $a^*$ starts all the jobs that can be feasibly started at time $T_r + \frac{\floor{(T^* - T_r)K}}{K}$; 
    \EndIf
\EndProcedure
\Procedure{some job that has not been scheduled by $a^*$ reaches its starting deadline}{}
    \State $a^*$ starts all the jobs that can be feasibly started now; 
\EndProcedure
\Procedure{time $T_R$}{}
    \State $a^*$ starts jobs $j^{T_R}_1, j^{T_R}_2, \ldots, j^{T_R}_{K}$ at time $T_R + 2$; 
\EndProcedure
\end{algorithmic}
\end{algorithm}

    First, we show that $a^*$ is a well-defined online algorithm. 
    Line 1,9 are valid commands, because counting the number of rounds is an online behavior and a new round starts as soon as a rigid job is released. 
    Take any $r$. At time $T_r$, note that new information about the input set of jobs is revealed only at or after $T_r+2$, because the next round of jobs $\sigma^{T_{r+1}}_{K}$ can only be released at or after $T_r+2$. 
    Then, for line 2-6, at time $T_r$ it is feasible to know how $a$ schedules jobs during $[T_r, T_r + 2)$. 
    At last, line 7-8 ensures that every job is feasibly scheduled by algorithm $a^*$. 
    
    Next, we show that $a^*$ has no higher cost than $a$ in any sample path. 
    First, we show that $a^*$ has no higher cost than $a$ during $[T_R, \infty)$. 
    Note that the cost of $a^*$ during $[T_R, \infty)$ is exactly $2$ that are during $[T_R, T_R+1]$ and $[T_R+2, T_R+3]$. 
    Indeed $2$ is the minimum cost any algorithm can get during $[T_R, \infty)$ because the processing of job $j^{T_R}_s$ and $T^{T_R}_K$ must be fully within $[T_R, \infty)$ and the two jobs' processing cannot have any overlapping. 
    
    It suffices to show that during $[T_r, T_{r+1}]$, $a^*$ has no higher cost than $a$ for each $r = 0,1,\ldots,R-1$.     
    Consider the case that there exists some job $j$ started by $a$ during $[T_r+1, T_r+2)$. 
    Let $\Gamma_r$ be $\frac{1}{K}$ in this case. 
    Clearly, the cost of $a^*$ during $[T_r, T_{r+1}]$ is either $1$ or $2$ corresponding to either $T_{r+1} = T_r+2$ or $T_{r+1} \geq T_r+3$ respectively. 
    The case of $T_{r+1} = T_r+2$ is trivial. 
    In the case of $T_{r+1} \geq T_r+3$, the cost of $a$ during $[T_r, T_{r+1}]$ is at least $2$ because $a$ is processing the whole job $j^{T_r}_s$ during $[T_r,T_r+1]$ and the whole job $j$ during $[T_r+1, T_r+3]$. 
    Therefore, $a^*$ has no higher cost than $a$ in the case that $a$ starts some jobs during $[T_r+1, T_r+2)$. 

    Next, consider the case that $a$ does not start any job during $[T_r+1, T_r+2)$ and $T^*$ doesn't exist (line 5). 
    Let $\Gamma_r$ be $\frac{1}{K}$ in this case. 
    Both $a$ and $a^*$ starts all the jobs $j^{T_r}_{1}, j^{T_r}_{2}, \ldots, j^{T_r}_{K}$ at time $T_r+2 = s_{j^{T_r}_{\Gamma_r\cdot K}}$, and hence, both algorithms have the same cost. 
    
    The last case is that $a$ does not start any job during $[T_r+1, T_r+2)$ and $T^*$ exists (line 5). 
    Let $\Gamma_r$ be $\frac{\floor{(T^*-T_r)K}+1}{K} \in \{\frac{2}{K}, \frac{3}{K}, \ldots, 1\}$, since $\frac{1}{K} \leq T^*-T_r < 1$. 
    By assumption, $a$ starts jobs $j^{T_r}_{\Gamma_r\cdot K}, j^{T_r}_{\Gamma_r\cdot K + 1}, \ldots, j^{T_r}_{K}$ at or after time $T_r+2$. 
    In the case of $T_{r+1} > s_{j^{T_r}_{\Gamma_r\cdot K}}$, by assumption $a$ starts job $j^{T_r}_{\Gamma_r\cdot K}$ during $[T_r+2, T_{r+1}-1]$. 
    It follows that the cost of $a$ during $[T_r, T_{r+1}]$ is at least $2 + (T^*-T_r)$.
    In the case of $T_{r+1} \leq s_{j^{T_r}_{\Gamma_r\cdot K}}$, clearly the cost of $a$ during $[T_r, T_{r+1}]$ is at least $1 + (T^*-T_r)$. 
    On the other hand, by definition, $a^*$ starts $j^{T_r}_1, j^{T_r}_2, \ldots, j^{T_r}_{\Gamma_r\cdot K-1}$ at time $T_r + \Gamma_r - \frac{1}{K}$ and starts $j^{T_r}_{\Gamma_r\cdot K}, j^{T_r}_{\Gamma_r\cdot K+1}, \ldots, j^{T_r}_{K}$ at time $\min\{T_{r+1}, s_{j^{T_r}_{\Gamma_r\cdot K}}\}$. 
    It follows that the cost of $a^*$ during $[T_r, T_{r+1}]$ is either $2 + \Gamma_r - \frac{1}{K}$ or $1 + \Gamma_r - \frac{1}{K}$ corresponding to either $T_{r+1} > s_{j^{T_r}_{\Gamma_r\cdot K}}$ or $T_{r+1}\leq s_{j^{T_r}_{\Gamma_r\cdot K}}$ respectively. 
    Since $T^* - T_r \geq \Gamma_r - \frac{1}{K}$, we have shown that $a^*$ has no higher cost than $a$ in any case. 

    By definition, observe that $\Gamma_r$ is defined based on the information strictly before time $T_r+2$. All the information strictly before time $T_r+2$ are exactly $\Delta_1, \Delta_2, \ldots, \Delta_r$. 
    Therefore, $\Gamma_r$ is a function of $\Delta_1, \Delta_2, \ldots, \Delta_r$. 
    Eventually, algorithm $a^*$ matches all the descriptions in the lemma. 
\end{proof}

\begin{proof} [Proof of Lemma \ref{lem:finite}]
    Let $\mathcal{A}'$ denote the set of deterministic online algorithms each of which fits the description of $a^*$ in Lemma \ref{lem:best_alg} for some $\Gamma_0, \Gamma_1, \ldots, \Gamma_{R-1}$. 
    Consider the equivalence relation $\sim$ among $\mathcal{A}'$ such that $a_1\sim a_2$ iff $a_1$ and $a_2$ make identical scheduling for each $x\in \mathcal{X}^*_{K,R}$. 
    Let $\mathcal{A}^*_{K,R}$ denote the set of representatives of the equivalence classes of $\mathcal{A}'$. 
    It suffices to show that $\mathcal{A}^*_{K,R}$ is finite. 

    Suppose $D_r$ is the set $\{(\delta_1,\delta_2,\ldots,\delta_r): \delta_{r'} = 2,3,\ldots,K+1, r' = 1,2,\ldots,r\}$ and $C$ is the set $\{\frac{1}{K}, \frac{2}{K}, \ldots, 1\}$. 
    Note that there exists a one-one correspondence between $\mathcal{A}^*_{K,R}$ and $\{(\Gamma_0,\Gamma_1,\ldots,\Gamma_{R-1}): \Gamma_r \in C^{D_r}\}$. 
    It suffices to show that $\{(\Gamma_0,\Gamma_1,\ldots,\Gamma_{R-1}): \Gamma_r \in C^{D_r}\}$ is finite, which is clearly true. 
\end{proof}

\begin{proof} [Proof of Lemma \ref{lem:comp_1}]
Consider the cost of algorithm $a^*$ with parameters $\Gamma_0, \ldots,\Gamma_R$ scheduling $X_{K,R}$ in the following parts
\begin{itemize}
    \item for each $r = 1,2,\ldots,R$, the cost during $[T_{r-1}, T_{r}]$ is $1 + \Gamma_{r-1} - \frac{1}{K} + \mathds{1}\{\Delta_{r}-1 > \Gamma_{r-1}\cdot K\}$, where job $j^{T_{r-1}}_{\Gamma_{r-1}\cdot K - 1}$ is started at its release time $T_{r-1} + \Gamma_{r-1} - \frac{1}{K}$ and the processing of job $j^{T_{r-1}}_{\Gamma_{r-1}\cdot K}$ completes before or at time $T_{r}$ only if $T_{r} > T_{r-1} + \Gamma_{r-1}\cdot K + 1 = s_{j^{T_{r-1}}_{\Gamma_{r-1}\cdot K}}$; 
    \item the cost during $[T_R, \infty)$ is $2$ that are during $[T_R, T_R+1]$ and $[T_R+2, T_R+3]$; 
\end{itemize}
Then, we have the expected algorithm cost
\begin{eqnarray*}
    &\quad& \e\left[2 + \sum_{r = 1,2,\ldots,R} \left(1 + \Gamma_{r-1} - \frac{1}{K} + \mathds{1}\{\Delta_{r}-1 > \Gamma_{r-1}\cdot K\}\right)\right] \\
    &=& 2 + R\cdot \left(1 - \frac{1}{K}\right) + \sum_{r = 1,2,\ldots,R} \e\left[ \Gamma_{r-1} + \mathds{1}\{\Delta_{r}-1 > \Gamma_{r-1}\cdot K\}\right]. 
\end{eqnarray*}
Then, we have
\begin{eqnarray*}
    &\quad& \e\left[ \Gamma_{r-1} + \mathds{1}\{\Delta_{r}-1 > \Gamma_{r-1}\cdot K\}\right] \\
    &=& \sum_{\delta_1,\ldots,\delta_r = 2,3,\ldots,K+1} \left(\Gamma_{r-1}(\delta_1,\ldots,\delta_{r-1}) + \mathds{1}\{\delta_{r}-1 > \Gamma_{r-1}(\delta_1,\ldots,\delta_{r-1})\cdot K\}\right)\cdot \frac{1}{K^r}\\
    &=& \sum_{\delta_1,\ldots,\delta_{r-1} = 2,3,\ldots,K+1} \sum_{\delta_r = 2,3,\ldots,K+1} \left(\Gamma_{r-1}(\delta_1,\ldots,\delta_{r-1}) + \mathds{1}\{\delta_{r}-1 > \Gamma_{r-1}(\delta_1,\ldots,\delta_{r-1})\cdot K\}\right)\cdot \frac{1}{K^r} \\
    &=& \sum_{\delta_1,\ldots,\delta_{r-1} = 2,3,\ldots,K+1} \Gamma_{r-1}(\delta_1,\ldots,\delta_{r-1}) \cdot \frac{1}{K^{r-1}} + (K - \Gamma_{r-1}(\delta_1,\ldots,\delta_{r-1})\cdot K)\cdot \frac{1}{K^r} \\
    &=& \sum_{\delta_1,\ldots,\delta_{r-1} = 2,3,\ldots,K+1} \frac{1}{K^{r-1}} = 1. 
\end{eqnarray*}
After substitution, the expected algorithm cost is $2 + R\cdot \left(2 - \frac{1}{K}\right)$. 

Consider the optimal scheduling of $X_{K,R}$ in the following parts:  
\begin{itemize}
    \item for each $r = 1,2,\ldots,R$, the optimal cost during $[T_{r-1}, T_{r}]$ is $1 + \frac{\Delta_{r}-2}{K}$, where jobs $j^{T_{r-1}}_{1}, j^{T_{r-1}}_{2}, \ldots, j^{T_{r-1}}_{\Delta_r-2}$ are started at time $r_{j^{T_{r-1}}_{\Delta_{r}-2}} = T_{r-1} + \frac{\Delta_{r}-2}{K}$ and job $j^{T_{r-1}}_{\Delta_{r}-1}, j^{T_{r-1}}_{\Delta_{r}}, \ldots, j^{T_{r-1}}_{K}$ are started at time $s_{j^{T_{r-1}}_{\Delta_{r}-1}} = T_{r}$; 
    \item the optimal cost during $[T_R, \infty)$ is $2$, where jobs $j^{T_R}_{1}, j^{T_R}_{2}, \ldots, j^{T_R}_{K}$ are started at time $T_R + 2$; 
\end{itemize}

By the chosen distribution, we have the expected optimal cost $\e\left[ 2 + \sum_{r = 1,2,\ldots,R} 1 + \frac{\Delta_r-2}{K} \right] 
    = 2 + R\cdot\left(1-\frac{1}{K}\right) + R\cdot\frac{\e[\Delta_1-1]}{K} 
    = 2 + R\cdot\left(1-\frac{1}{K}\right) + R\cdot\frac{\frac{1+2+\cdots+K}{K}}{K} 
    = 2 + R\cdot\frac{3K-1}{2K}$. 
\end{proof}

\begin{proof} [Proof of Theorem \ref{thm:lb_1}]
    Since $\mathcal{X}_{K,R}$ and $\mathcal{A}^*_{K,R}$ are both finite, by Lemma \ref{lem:best_alg},\ref{lem:yao} and \ref{lem:comp_1}, the competitiveness of any randomized online algorithm is at least $\frac{\e[a^*(X_{K,R})]}{\e[\opt(X_{K,R})]} = \frac{2 + R\cdot \left(2 - \frac{1}{K}\right)}{2 + R\cdot\frac{3K-1}{2K}}$. 
    By letting $K$ go to infinity, we have the competitiveness of any randomized online algorithm is at least $\frac{2 + 2R}{2 + \frac{3}{2}R}$. 
    By letting $R$ go to infinity, we have the competitiveness of any randomized online algorithm is at least $\frac{4}{3}$. 
\end{proof}
\begin{proof} [Proof of Lemma \ref{lem:comp_2}]
Similar to the proof of Lemma \ref{lem:comp_1}, take any algorithm $a^*$ and consider 
{\small 
\begin{eqnarray*}
    &\quad& \e\left[ \Gamma_{r-1} + \mathds{1}\{\Delta_{r}-1 > \Gamma_{r-1}\cdot K\}\right] \\
    &=& \sum_{\delta_1,\ldots,\delta_r = 2,3,\ldots,K+1} (\Gamma_{r-1}(\delta_1,\ldots,\delta_{r-1}) + \mathds{1}\{\delta_{r}-1 > \Gamma_{r-1}(\delta_1,\ldots,\delta_{r-1})\cdot K\})\cdot \p(\Delta_1=\delta_1, \ldots,\Delta_r = \delta_r)\\
    &=& \sum_{\delta_1,\ldots,\delta_{r-1} = 2,3,\ldots,K+1} \sum_{\delta_r = 2,3,\ldots,K+1} \\
    &\quad& (\Gamma_{r-1}(\delta_1,\ldots,\delta_{r-1}) + \mathds{1}\{\delta_{r}-1 > \Gamma_{r-1}(\delta_1,\ldots,\delta_{r-1})\cdot K\})\cdot \p(\Delta_1=\delta_1, \ldots,\Delta_r = \delta_r) \\
    &=& \sum_{\delta_1,\ldots,\delta_{r-1} = 2,3,\ldots,K+1} \p(\Delta_1=\delta_1, \ldots,\Delta_{r-1} = \delta_{r-1})\cdot \\
    &\quad& \sum_{\delta_r = 2,3,\ldots,K+1} (\Gamma_{r-1}(\delta_1,\ldots,\delta_{r-1}) + \mathds{1}\{\delta_{r}-1 > \Gamma_{r-1}(\delta_1,\ldots,\delta_{r-1})\cdot K\})\cdot \p(\Delta_r = \delta_r) \\
    &=& \sum_{\delta_1,\ldots,\delta_{r-1} = 2,3,\ldots,K+1} \p(\Delta_1=\delta_1, \ldots,\Delta_{r-1} = \delta_{r-1})\cdot \left(\Gamma_{r-1}(\delta_1,\ldots,\delta_{r-1}) + \frac{(\tilde{K} - \Gamma_{r-1}(\delta_1,\ldots,\delta_{r-1})\cdot K)^+}{K}\right) \\
    &\geq& \sum_{\delta_1,\ldots,\delta_{r-1} = 2,3,\ldots,K+1} \p(\Delta_1=\delta_1, \ldots,\Delta_{r-1} = \delta_{r-1})\cdot \frac{\tilde{K}}{K} = \frac{\tilde{K}}{K}. 
\end{eqnarray*}
}
After substitution, the expected algorithm cost is $2 + R\cdot \left(1 - \frac{1}{K}\right) + \sum_{r = 1,2,\ldots,R} \e\left[ \Gamma_{r-1} + \mathds{1}\{\Delta_{r}-1 > \Gamma_{r-1}\cdot K\}\right] \geq 2 + R\cdot \left(1 + \frac{\tilde{K}-1}{K}\right)$. 

For the optimal cost, we have $\e[\Delta_r-1] = \frac{1+K-\tilde{K}}{K} + \frac{1}{K}(2+3+\cdots+\tilde{K}) = \frac{1}{K}(K-\tilde{K}+\frac{\tilde{K}(1+\tilde{K})}{2}) = 1 + \frac{\tilde{K}(\tilde{K-1})}{2K}$. 
After substitution, the expected optimal cost is $\e\left[ 2 + \sum_{r = 1,2,\ldots,R} 1 + \frac{\Delta_r-2}{K} \right] = 2 + R\cdot\left(1-\frac{1}{K}\right) + R\cdot\frac{1 + \frac{\tilde{K}(\tilde{K}-1)}{2K}}{K} = 2 + R\cdot \left(1 + \frac{\tilde{K}(\tilde{K}-1)}{2K^2}\right)$. 
\end{proof}
\begin{proof} [Proof of Theorem \ref{thm:low_best}]
    Let $p$ be any number between $0$ and $1$. 
    Let $\tilde{K}$ be $\floor{p\cdot K}$. 
    Since $\mathcal{X}_{K,R}$ and $\mathcal{A}^*_{K,R}$ are both finite, by Lemma \ref{lem:best_alg}, Lemma \ref{lem:yao}, and \ref{lem:comp_2}, the competitiveness of any randomized online algorithm is at least $\frac{2 + R\cdot \left(1 + \frac{\tilde{K}-1}{K}\right)}{2 + R\cdot \left(1 + \frac{\tilde{K}(\tilde{K}-1)}{2K^2}\right)}$. 
    Note that the term $\frac{2 + R\cdot \left(1 + \frac{\tilde{K}-1}{K}\right)}{2 + R\cdot \left(1 + \frac{\tilde{K}(\tilde{K}-1)}{2K^2}\right)} = \frac{2 + R\cdot \left(1 + \frac{\floor{p\cdot K}-1}{K}\right)}{2 + R\cdot \left(1 + \frac{\floor{p\cdot K}(\floor{p\cdot K}-1)}{2K^2}\right)}$ goes to $\frac{2+R\cdot (1+p)}{2+R\cdot (1+\frac{p^2}{2})}$, by letting $K$ go to infinity. 
    Then, we have the competitiveness of any randomized online algorithm is at least $\frac{2+R\cdot (1+p)}{2+R\cdot (1+\frac{p^2}{2})}$. 
    By letting $R$ go to infinity, we have the competitiveness of any randomized online algorithm is at least $\frac{1+p}{1+\frac{p^2}{2}}$. 
    Finally, $\frac{1+p}{1+\frac{p^2}{2}}$ reaches its maximum $\frac{1+\sqrt{3}}{2}$ by letting $p = \sqrt{3}-1$. 
\end{proof}

\section{Algorithm that restarts jobs}

\subsection{Lower bounds}
\begin{proof} [Proof of Lemma \ref{lem:best_alg_restart}]
    Consider defining $a$ in Algorithm \ref{alg:a} while simulating $a_r$. 
    We show that algorithm $a$ is a well defined online algorithm, and then, $a(x) \leq a_r(x)$ for each $x\in \mathcal{X}^*_{K,R}$. 
    
\begin{algorithm} 
\caption{$a_r$ mapped to $a$}
\label{alg:a}
\algrenewcommand\algorithmicprocedure{\textbf{Upon}}
\begin{algorithmic}[1]
\For {$r = 0,1,\ldots,R$}
    \Procedure{time $T_r$}{}
        \State $a$ starts all the jobs that can be feasibly started now;  
        \If{$a_r$ is always processing jobs during $[T_r,T_r+2)$}
            \State  $a$ starts all the jobs $j^{T_r}_1,\ldots,j^{T_r}_K$ at time $T_r+2$; 
        \Else
            \State $t\gets$ the first moment during $[T_r+1,T_r+2)$ when $a_r$ becomes idle;   
            \State $a$ starts all the jobs that can be feasibly started at time $t-1$;  
        \EndIf
    \EndProcedure
    \For{$k = 1,2,\ldots,\Delta_{r+1}-1$}
        \Procedure{time $T_r+k+1$}{}
            \If {$a_r$ finishes processing some job $j$ during $(T_r+k+1, T_r+k+2]$}
                \State $a$ starts all the jobs that can be feasibly started now; 
            \EndIf 
        \EndProcedure
    \EndFor
\EndFor
\end{algorithmic}
\end{algorithm}    

First, we show that $a$ is online. 
Since $a_r$ is an online algorithm, all the scheduling decisions $a_r$ make at time $t$ depends only on what happened at and before time $t$. 
Note that all the decisions $a_r$ make before time $T_r + 2$ depends only on $\Delta_1,\Delta_2,\ldots,\Delta_r$ because $T_{r+1}\geq T_r + 2$ almost surely. 
Since $\Delta_1,\Delta_2,\ldots,\Delta_r$ have all been revealed at time $T_r$, it is feasible for $a$ at time $T_r$ to know how $a_r$ processes jobs during $[T_r,T_r+2)$ (line 4 and 6). 
Similarly, all the available information about the input job instance at time $T_r + k + 1$ equal to that at any time $t\in (T_r+k+1,T_r+k+2)$ for $k = 1,2,\ldots,\Delta_{r+1}-1$. 
Then, it is feasible for $a$ at time $T_r+k+1$ to know all the decisions $a_r$ make during $[0,T_r+k+2)$ and hence to know whether $a_r$ finishes processing a job during $(T_r+k+1,T_r+k+2]$ (line 11). 

Next, we ensure that every job is feasibly scheduled by $a$. 
First, we show that if $a$ schedules any job (line 3,5,8,12), then the job's processing is fully within its window. 
Since all the jobs $j^{T_r}_1,\ldots,j^{T_r}_K$ are released before $T_r+1$ and have starting deadline at or after $T_r+2$, line 5 is feasible. Line 3,8,12 are feasible, clearly. 
Second, we show $a$ does not miss any job. 
At time $T_r$ for each $r$, line 3 starts job $j^{T_r}_s$ at time $T_r$, so it suffices to consider $j^{T_r}_1,\ldots,j^{T_r}_K$ for each $r$. 
Furthermore, due to line 5, it suffices to consider the case that $a_r$ is not always processing jobs during $[T_r,T_r+2)$. 
Suppose $t$ is the first moment during $[T_r+1,T_r+2)$ when $a_r$ becomes idle. 
Take any $k = 1,2,\ldots,K$. Let $t^*+1$ denote when $a_r$ finish processing $j^{T_r}_k$. 
By line 6-8, if $a$ does not start job $j^{T_r}_k$ at time $t-1$, then $j^{T_r}_k$ must be released after $t-1$, and hence $a_r$ cannot finish processing $j^{T_r}_k$ at or before time $t$. Since $a_r$ is not processing any job at time $t$, we have $t^* > t \geq T_r+1$. 
Then, we have $t^*+1 > T_r+2$. 
If $t^*+1\in (T_r + k + 1, T_r + k + 2]$ for some $k = 1,2,\ldots,\Delta_{r+1}-1$, then we have $d_{j^{T_r}_k}\geq T_r+k+2$ because jobs' deadline are all integral. 
Thus, it is feasible for $a$ to start job $j$ at time $T_r+k+1$, and hence $a$ starts job $j^{T_r}_k$ at time $T_r+k+1$ due to line 12. 
If $t^* + 1> T_{r+1}+1$, then line 3 ensures that $a$ must start job $j^{T_r}_k$ at or before time $T_{r+1}$. 

It remains to show $a(x) \leq a_r(x)$ for each $x\in \mathcal{X}^*_{K,R}$. 
Fix $x\in \mathcal{X}^*_{K,R}$. 
It suffices to show that the cost of $a$ is no higher than the cost of $a_r$ during $[T_r,T_{r+1})$ for any $r = 0,1,\ldots,R$ where $T_{R+1} = \infty$. 
If $a_r$ is always processing jobs during $[T_r,T_r+2)$, then the cost of $a_r$ is at least $2$ that are during $[T_r,T_r+2]$ and the cost of $a$ is at most $2$ that are during $[T_r,T_r+1]$ and $[T_r+2,T_r+3]$ due to line 5. 

If $a_r$ is not always processing jobs during $[T_r,T_r+2)$, suppose $t$ is the first moment during $[T_r+1,T_r+2)$ when $a_r$ becomes idle. 
Both $a$ and $a_r$ have the same cost of $t - T_r$ during $[T_r,t]$. 
In this case, suppose $j^{T_r}_k$ is the first job completed by $a_r$ after time $t+1$, and say the completion time of $j^{T_r}_k$ is $t^*+1$. 
By definition, $a$ starts the remaining of $j^{T_r}_1,\ldots,j^{T_r}_K$ at time $\ceil{t^*}$ as in line 12. 
Therefore, if there is a positive cost of $a$ during $[t,T_{r+1}]$, then the cost is exactly $1$ and $\ceil{t^*} + 1 \leq T_{r+1}$. 
Meanwhile, the cost of $a_r$ during $[t,T_{r+1}]$ is at least $1$ due to some job $j^{T_r}_k$ during $[t^*,t^*+1]$. 
\end{proof}

\begin{proof} [Proof of Lemma \ref{lem:restart_lb}]

Partition $\{0,1,\ldots,R-1\}$ into $S_1:=\{r: \Gamma_{r} > \phi\}$ and $S_2:=\{r: \Gamma_{r} \leq \phi\}$. 
For each $r\in S_1$, the cost of $a^*$ during $[T_r,T_{r+1}]$ is $1 + \Gamma_r - \frac{1}{K}$ that are during $[T_r,T_r + \Gamma_r - \frac{1}{K}+1] \subset [T_r, T_r+2] = [T_r,T_{r+1}]$, since $\Delta_{r+1} = 2$ by definition. 
For each $r\in S_2$, the cost of $a^*$ during $[T_r,T_{r+1}]$ is $1 + \Gamma_r - \frac{1}{K} + 1$ that are during $[T_r,T_r + \Gamma_r - \frac{1}{K}+1]$ and $\left[s_{j^{T_r}_{\Gamma_r\cdot K}}, s_{j^{T_r}_{\Gamma_r\cdot K}}+1\right] = \left[s_{j^{T_r}_{\Gamma_r\cdot K}}, s_{j^{T_r}_{\Gamma_r\cdot K+1}}\right] = \left[s_{j^{T_r}_{\Gamma_r\cdot K}}, T_{r+1}\right]$. 
The cost of $a^*$ during $[T_R,\infty)$ is $2$ that are during $[T_R, T_{R}+1]$ and $[T_{R}+2,T_{R}+3]$ due to Lemma \ref{lem:best_alg}. 

Consider the following scheduling denoted by $b$. 
For each $r\in S_1\cup \{R\}$, $b$ starts jobs $j^{T_r}_1,j^{T_r}_2,\ldots,j^{T_r}_K$ at time $s_{j^{T_r}_1} = T_r+2 = T_{r+1}$. 
For each $r\in S_2$, $b$ starts jobs $j^{T_r}_1, j^{T_r}_2, \ldots, j^{T_r}_{\Gamma_r\cdot K}$ at time $r_{j^{T_r}_{\Gamma_r\cdot K}}$ and jobs $j^{T_r}_{\Gamma_r\cdot K+1}, j^{T_r}_{\Gamma_r\cdot K+2}, \ldots, j^{T_r}_{K}$ at time $s_{j^{T_r}_{\Gamma_r\cdot K+1}} = T_{r+1}$. 
Clearly, for each $r\in S_1$, the cost of $b$ during $[T_r,T_{r+1}]$ is $1$ that are during $[T_r,T_r + 1]$ by definition. 
For each $r\in S_2$, the cost of $b$ during $[T_r,T_{r+1}]$ is $1 + \Gamma_r$ that are during $[T_r,T_r + \Gamma_r +1]$. 
The cost of $b$ during $[T_R,\infty)$ is $2$. 

Let $c$ denote $\min\left\{1 + \phi, \frac{2+\phi}{1+\phi}\right\}$. 
Note that for each $r\in S_2$, we have $\frac{2 + \Gamma_r}{1 + \Gamma_r} \geq \frac{2+\phi}{1+\phi}$ since $\Gamma_r \leq \phi$. 
Then, we have the competitiveness of $a^*$ is at least 
\begin{eqnarray*}
    &\quad& \frac{\left(\sum_{r\in S_1} 1 + \Gamma_r - \frac{1}{K}\right) + \left(\sum_{r\in S_2} 2 + \Gamma_r - \frac{1}{K}\right) + 2}{|S_1| + \left(\sum_{r\in S_2} 1 + \Gamma_r\right) + 2} \\
    &\geq{}& \frac{\left(\sum_{r\in S_1} 1 + \Gamma_r \right) + \left(\sum_{r\in S_2} 2 + \Gamma_r \right)}{|S_1| + \left(\sum_{r\in S_2} 1 + \Gamma_r\right) + 2} - \frac{R/K}{|S_1| + |S_2|} \\
    &={}& \frac{\left(\sum_{r\in S_1} 1 + \Gamma_r \right) + \left(\sum_{r\in S_2} 2 + \Gamma_r \right)}{|S_1| + \left(\sum_{r\in S_2} 1 + \Gamma_r\right) + 2} -  \frac{1}{K} \\
    &\geq{}& \frac{\left(\sum_{r\in S_1} c \right) + \left(\sum_{r\in S_2} c\cdot (1 + \Gamma_r) \right)}{|S_1| + \left(\sum_{r\in S_2} 1 + \Gamma_r\right)}\cdot \frac{|S_1| + \left(\sum_{r\in S_2} 1 + \Gamma_r\right)}{|S_1| + \left(\sum_{r\in S_2} 1 + \Gamma_r\right) + 2} -  \frac{1}{K} \\
    &\geq{}& c\cdot \frac{R}{R + 2} -  \frac{1}{K}. 
\end{eqnarray*} 
\end{proof}

\begin{proof} [Proof of Theorem \ref{thm:restart_lb}]
    By Lemma \ref{lem:restart_lb}, the competitiveness of $a_r$ is at least $c\cdot \frac{R}{R + 2} -  \frac{1}{K}$. 
    Since $a_r$ is taken as any deterministic online algorithm that restarts jobs, the competitiveness of any deterministic online algorithm that restarts jobs is at least $c\cdot \frac{R}{R + 2} -  \frac{1}{K}$. 
    Since $R$ and $K$ can be arbitrary large, we have the competitiveness of any deterministic online algorithm that restarts jobs is at least $c$, where $c = \min\left\{1 + \phi, \frac{2+\phi}{1+\phi}\right\}$ reaches maximum $\frac{\sqrt{5}+1}{2}$ when $1 + \phi = \frac{2+\phi}{1+\phi}$. 
\end{proof}

\subsection[Analysis of Best Batch]{Analysis of Best Batch $\bb^\gamma$}

\begin{proof} [Proof of Lemma \ref{lem:hat_f}]
    We show it by induction. 
    The set of available job may change in three cases: a new job $j$ is released (line \ref{alg_line:release_begin}), the set of available jobs are started (line \ref{alg_line:flag_start}), and some jobs are stopped (line \ref{alg_line:stop}). 
    It suffices to show that in each case above, the value of $\hat{f}$ is updated correctly. 
    
    Case 1: upon $j$ is released (line \ref{alg_line:release_begin}). 
    It suffices to consider the case that $j$ is not started as in line \ref{alg_line:start_primary},\ref{alg_line:start_tight}, \ref{alg_line:start_secondary}, i.e., job $j$ becomes available after being released. 
    Assume that the value of $\hat{f}$ upon $j$ is released before line \ref{alg_line:f_gets_1} is the job that has the earliest starting deadline among all the jobs that are available at time $r_j$ except job $j$. 
    We show that the value of $\hat{j}$ upon $j$ is released after line \ref{alg_line:f_gets_1} is the job that has the earliest starting deadline among all the jobs that are available at time $r_j$ including job $j$. 
    Since the new value of variable $\hat{f}$ takes the job with the earlier starting deadline between $j$ and its old value in line \ref{alg_line:f_gets_1}, we are done. 
    
    Case 2: time $d_p$ for some primary job $p$. 
    Assume that $\hat{f}$ upon time $d_p$ before line \ref{alg_line:f_gets_2} is the job that has the earliest starting deadline among all the jobs that are available at time $d_p$ except jobs $\{j\in \mathcal{J}^p: r_j > r_{j^p_{k^p-1}}\}$ that are stopped in line \ref{alg_line:stop}. 
    We show that $\hat{f}$ upon time $d_p$ after line \ref{alg_line:f_gets_2} is the job that has the earliest starting deadline among all the jobs that are available at time $r_j$ including jobs $\{j\in \mathcal{J}^p: r_j > r_{j^p_{k^p-1}}\}$. 
    Note that $j^p_{k^p}$ has the earliest starting deadline among jobs $\{j\in \mathcal{J}^p: r_j > r_{j^p_{k^p-1}}\}$. 
    This is because each job $j$ in $\mathcal{J}^p$ such that $r_j > r_{j^p_{k^p-1}}$ satisfies either $j\in \{j^p_{k^p}, j^p_{k^p+1},\ldots,j^p_{K}\}$ or $[r_j,d_j]\supset [r_{j^p_{k}},d_{j^p_{k}}]$ for some $k = k^p, k^p+1,\ldots, K$. 
    It follows that the new value of $\hat{f}$ as the job with the smaller starting deadline between $j^p_{k^p}$ and its old value has the earliest starting deadline among all the jobs available at time $d_p$ including $\{j\in \mathcal{J}^p: r_j > r_{j^p_{k^p-1}}\}$. 
    
    Case 3: time $s_{\hat{f}}$. 
    $\hat{f}$ takes value $NULL$ at time $s_{\hat{f}}$ as in line \ref{alg_line:f_gets_3} after line \ref{alg_line:flag_start}. 
    Since line \ref{alg_line:flag_start} starts all the available jobs at at time $s_{\hat{f}}$, the set of available jobs becomes empty after line \ref{alg_line:flag_start} at time $s_{\hat{f}}$. 
    Therefore, $\hat{f}$ should not be any job at time $s_{\hat{f}}$ after line \ref{alg_line:f_gets_3}. 
\end{proof}

\begin{proof} [Proof of Lemma \ref{lem:feasible}]
    Here, we show (i). 
    Assume $j$ is not started during $[t, s_j]$ for the sake of contradiction. 
    Then, job $j$ is available to be started during $[t, s_j]$. 
    Take any time instant $t'\in [t, s_j]$. 
    Let $f$ denote the value of $\hat{f}$ at time $t'$. 
    By Lemma \ref{lem:hat_f}, we have $f\neq NULL$ and $s_f \leq s_j$. 
    It follows that at time $s_f$, job $j$ is available to be started. 
    By algorithm definition, line \ref{alg_line:flag_start} says job $j$ must be started at some time during $s_f\in [t,s_j]$ (contradiction). 
    
    For (ii), suppose job $j$ is stopped at time $d_p$ due to primary job $p$. 
    Note that job $j$ must be in $\mathcal{J}^p$.  
    Note that $C^p_{k} \leq d_{j^p_{K^p}} - s_p$ for each $k = 1,2,\ldots,K^p$ by definition of $C^p_{k}$ (line \ref{alg_line:key}). 
    Furthermore, note that $C^p_{k} < d_{j^p_{K^p}} - s_p$ iff $1 + r_{j^p_{k-1}} < s_{j^p_{k}}$ where $r_{j^p_{0}} = s_p$. 
    Since line \ref{alg_line:stop} happens only when $p$ is not tight (line \ref{alg_line:separated}), we have $C^p_{k^p} < d_{j^p_{K^p}} - s_p$, and then, we have $1 + r_{j^p_{k^p-1}} < s_{j^p_{k^p}}$. 
    Note that $d_p \leq 1 + r_{j^p_{k^p-1}} < s_{j^p_{k^p}}$. 
    $j$ is stopped due to $p$ (line \ref{alg_line:stop}) only if $r_j > r_{j^p_{k^p-1}}$. 
    It follows that $j$ is either $j^p_k$ for some $k = k^p, k^p+1, \ldots, K$ or $[r_j,d_j]\supset [r_{j^p_k}, d_{j^p_k}]$ for some $k = k^p, k^p+1, \ldots, K$ as defined in line \ref{alg_line:critical_jobs}. 
    In either of the cases, we have $s_j \geq s_{j^p_{k^p}} > d_p$. 
\end{proof}

\begin{proof} [Proof of Lemma \ref{lem:agreeable}]
    We prove all the statements by induction on the value of variable $x$. 
    Take any $y = 1,2,\ldots,|\mathcal{P}|$. 
    Assume all the statements hold for $x = 1,2,\ldots,y-1$, and we show all the statements hold for $x = y$. 

    Here, we show (viii). 
    Let $\mathcal{F}$ denote the set of flag jobs among $p_1,j^{p_1}_{k^{p_1}}, p_2,j^{p_2}_{k^{p_2}}, \ldots, p_{y-1},j^{p_{y-1}}_{k^{p_{y-1}}}, p_y$. 
    By Induction Hypothesis (IH), the jobs in $\mathcal{F}$ are rewritten as $p_1$,$j^{p_1}_{k^{p_1}}$, $p_2$,$j^{p_2}_{k^{p_2}}$, $\ldots$, $p_{y-1}$, $j^{p_{y-1}}_{k^{p_{y-1}}}$, $p_y$ in the time order of being marked as flag jobs where $j^{p_1}_{k^{p_1}}, j^{p_2}_{k^{p_2}}, \ldots, j^{p_{y-1}}_{k^{p_{y-1}}}$ may or may not exist or marked as a flag job. 
    Note that a job is started only in line \ref{alg_line:flag_start}, \ref{alg_line:start_primary}, \ref{alg_line:start_tight}, \ref{alg_line:start_secondary} due to some flag job. 
    By Induction Hypothesis (IH), one may check all the cases and see that job $p_y$ is not started due to any flag job in $\mathcal{F}\setminus \{p_y\}$. 
    By algorithm's definition, if $p_y$ started due to some flag job $f$, then $p_y$ is started due to some flag job $f$ that is marked as a flag job before $p_y$. 
    Since any flag job $f\notin \mathcal{F}$ is marked as a flag job later than $p_y$,  
    we have $p_y$ is available to be scheduled since its release and never started until time $s_{p_y}$. 
    In line \ref{alg_line:flag_begin}, job $p_y$ is available to be started at time $s_{p_y}$, and hence started at time $s_{p_y}$ as in line \ref{alg_line:flag_start}. 
    By Lemma \ref{lem:feasible} (ii), $p_y$ is never stopped and hence completed at its deadline $d_{p_y}$. 

    Here, we show (v), i.e., no job is not stopped during $[r_{p_y}, d_{p_y})$. 
    Note that jobs are only stopped at time $d_p$ for some primary flag job $p$. 
    By IH, we have $r_{p_y} > d_{p_x}$ for any $x\leq y-1$. 
    On the other hand, we have $d_{p_y} \leq d_{p_x}$ for any $x\geq y+1$. 
    Therefore, no job can be stopped during $[r_{p_y}, d_{p_y})$. 

    Here we show (vi). 
    Since all the jobs are feasibly scheduled, it suffices to look at jobs whose starting deadlines are after time $s_{p_y}$. 
    Take any job $j$ such that $r_j \leq s_{p_y}$ and $s_j > s_{p_y}$. 
    Either job $j$ is available to be started at time $s_{p_y}$ due to $p_y$ or not, as in line \ref{alg_line:flag_start}. 
    If yes, job $j$ is started at time $s_{p_y}$ due to $p_y$ by definition, and then completed at time $d_{p_y}$, because no job is stopped during $[s_{p_y}, d_{p_y})$. 
    If no, then either job $j$ has already been completed before or at time $s_{p_y}$, or, job $j$ has been started before or at time $s_{p_y}$ and is processed at time $s_{p_y}$. 
    Since no job is stopped during $[s_{p_y}, d_{p_y})$, in any of the cases, job $j$ must be completed before or at time $d_{p_y}$. 

    We start to show (vii). 
    Claim 1: $r_{p_x} > s_{p_y}$ for any $x\geq y+1$. 
    At time $s_{p_y}$, variable $\hat{f}$ takes the value $NULL$ in line \ref{alg_line:f_gets_3}. 
    Since all the available jobs are started at time $s_{p_y}$, after being $NULL$, $\hat{f}$ can only be some job released after $s_{p_y}$ either in line \ref{alg_line:f_gets_2} or \ref{alg_line:f_gets_1}, because IH guarantees that $\hat{f}$ cannot be updated to be any secondary job attached to $p_x$ for any $x\leq y-1$. 
    The claim follows, because $\hat{f}$ has to be $p_x$ at time $s_{p_x}$ for any $x\geq y+1$. 
    
    Since any secondary job attached to $p_x$ must be released after $s_{p_x}$, any secondary job attached to $p_x$ for some $x\geq y$ must be released after time $s_{p_y}$. 
    We have that all the jobs $p_y, p_{y+1},\ldots,p_{|\mathcal{P}|}$ and their secondary jobs (if exists) are released at or after $r_{p_y}$. 

    Take any job $j$ such that $r_j \geq r_{p_y}$. 
    Take any primary flag job $p_x$ with $x\leq y-1$. 
    If line \ref{alg_line:critical_case} does not hold for $p_x$, then IH (i) says $d_{p_x} < r_{p_{x+1}}$, and hence, we have $d_{p_x} < r_{p_{x+1}} \leq r_{p_y} \leq r_j$. 
    By definition (line \ref{alg_line:flag_start}, \ref{alg_line:start_primary}), job $j$ cannot be started due to $p_x$, since $p_x$ is not tight. 
    If line \ref{alg_line:critical_case} holds for $p_x$ and $p_x$ is tight, then IH (ii) says $1 + r_{j^{p_x}_{K^{p_x}}} < r_{p_{x+1}} \leq r_{p_y} \leq r_j$. 
    By definition (line \ref{alg_line:flag_start}, \ref{alg_line:start_primary}, \ref{alg_line:start_tight}), job $j$ cannot be started due to $p_x$. 
    If line \ref{alg_line:critical_case} holds for $p_x$ and $p_x$ is not tight, then IH (iii) says $d_{p_x} < r_{p_{x+1}} \leq r_{p_y} \leq r_j$. 
    By definition (line \ref{alg_line:flag_start}, \ref{alg_line:start_primary}), job $j$ cannot be started due to $p_x$.
    Take any secondary flag job $j^{p_{x}}_{k^{p_x}}$ with $x\leq y-1$. 
    By IH (iv), we have $s_{j^{p_x}_{k^{p_x}}} + (r_{j^{p_x}_{K^{p_x}}} - r_{j^{p_x}_{k^{p_x}-1}}) < r_{p_{x+1}}\leq r_{p_y} \leq r_j$. 
    By definition (line \ref{alg_line:flag_start}, \ref{alg_line:start_secondary}), job $j$ cannot be started due to $j^{p_x}_{k^{p_x}}$. 

    Claim 2: $r_{p_{y+1}}\notin (s_{p_y}, d_{p_y}]$. 
    Assume that $r_{p_{y+1}}\in (s_{p_y}, d_{p_y}]$ and we show a contradiction. 
    Note that $p_{y+1}$ is a flag job, and hence, variable $\hat{f}$ must be $p_{y+1}$ at time $s_{p_{y+1}}$. 
    Since $r_{p_{y+1}}\in (s_{p_y}, d_{p_y}]$, line \ref{alg_line:f_gets_1} cannot happen to $p_{y+1}$ because $p_{y+1}$ is started as it is released in line \ref{alg_line:start_primary}. 
    But since $p_{y+1}$ is primary, $p_{y+1}$ is never marked as secondary. 
    It follows that line \ref{alg_line:f_gets_2} cannot happen to $p_{y+1}$ either. 
    Therefore, we have shown that $r_{p_{y+1}}\in (s_{p_y}, d_{p_y}]$ cannot hold. 
    
    It follows from Claim 1 and 2 that (i) holds for $x = y$. 

    Here, we show (ii) holds for $x = y$. 
    By Claim 1 and 2, it remains to show that $r_{p_{y+1}}\in (d_{p_y}, 1 + r_{j^{p_y}_{K^{p_y}}}]$ does not hold. 
    Note that $\hat{f}$ must be updated to be $p_{y+1}$. 
    Since $r_{p_{y+1}}\in (d_{p_y}, 1 + r_{j^{p_y}_{K^{p_y}}}]$, $p_{y+1}$ is started at its release time in line \ref{alg_line:start_tight} and hence $\hat{f}$ cannot be $p_{y+1}$ at $p_{y+1}$'s release as in line \ref{alg_line:f_gets_1}. 
    Since $p_{y+1}$ is primary, $p_{y+1}$ is never marked as secondary, and hence, line  \ref{alg_line:f_gets_2} cannot happen to $p_{y+1}$ either. 
    Therefore, $r_{p_{y+1}} > 1 + r_{j^{p_y}_{K^{p_y}}}$ must hold. 

    Here, we show the first part of (iii) holds for $x = y$. 
    By IH (vi), all the secondary flag jobs attached to $p_x$ with $x\leq y-1$ are marked as flag jobs before $p_y$. 
    The secondary flag jobs attached to $p_x$ must be completed after $p_x$ for each $x\geq y+1$. 
    Thus, the flag job (marked as a flag job) right next to $p_y$ can only be $p_{y+1}$, since $p_y$ does not have a secondary flag job. 
    It follows from Claim 1 and 2 that the second part of (iii) holds for $x = y$. 

    Here, we show the first part of (iv) holds for $x = y$. 
    Note that the flag job right next to $p_y$, say $f$, is either $p_{y+1}$ or $j^{p_y}_{k^{p_y}}$.
    Claim 1 and 2 says $r_{p_{y+1}} > d_{p_y}$. 
    It follows that no job can be stopped during $[s_{p_{y+1}}, d_{p_{y+1}})$, because $d_{p_x} < s_{p_{y+1}}$ for each $x\leq y$ and $d_{p_x}\geq d_{p_{y+1}}$ for each $x\geq y+1$. 
    Assume the contrary that $f$ is $p_{y+1}$ for the sake of contradiction. 
    Then, job $j^{p_y}_{k^{p_y}}$ becomes available to be started at time $d_{p_y}$ and is kept available until time $s_{p_{y+1}}$, because (vii) guarantees that $j^{p_y}_{k^{p_y}}$ and $p_{y+1}$ cannot be started due to any of $p_1,p_2,\ldots,p_{y-1}$ or their secondary flag jobs, and $p_{y+1}$ is the flag job right next to $p_y$. 
    It follows that $s_{j^{p_y}_{k^{p_y}}} \geq s_{p_{y+1}}$, which implies variable $\hat{f}$ cannot be $j^{p_y}_{k^{p_y}}$ since $p_{y+1}$ is released. 
    It contradicts to the fact that $j^{p_y}_{k^{p_y}}$ is a flag job that says $\hat{f}$ is $j^{p_y}_{k^{p_y}}$ at time $s_{j^{p_y}_{k^{p_y}}}$. 

    Here, we show the second part of (iv) holds for $x=y$. 
    By the first part, We have the flag job right next to $p_y$ is $j^{p_y}_{k^{p_y}}$. 
    By (vii), we have $j^{p_y}_{k^{p_y}}$ becomes available since time $d_{p_y}$ and remains available until time $s_{j^{p_y}_{k^{p_y}}}$. 
    It follows that $j^{p_y}_{k^{p_y}}$ is started at time $s_{j^{p_y}_{k^{p_y}}}$ as in line \ref{alg_line:flag_start}. 
    By Lemma \ref{lem:feasible}, job $j^{p_y}_{k^{p_y}}$ is guaranteed to be completed at time $d_{j^{p_y}_{k^{p_y}}}$. 
    
    Here, we show the third part of (iv) holds for $x = y$. 
    By Claim 1 and 2, it suffices to show that $r_{p_{y+1}}\notin (d_{p_y}, s_{j^{p_y}_{k^{p_y}}} + (r_{j^{p_y}_{K^{p_y}}} - r_{j^{p_y}_{k^{p_y}-1}})]$.
    Assume $r_{p_{y+1}}\in (d_{p_y}, s_{j^{p_y}_{k^{p_y}}} + (r_{j^{p_y}_{K^{p_y}}} - r_{j^{p_y}_{k^{p_y}-1}})]$ for the sake of contradiction. 
    (vii) guarantees that $p_{y+1}$ can only be started due to $p_y$ or some flag job after $p_y$. 
    By definition, $p_{y+1}$ cannot be started due to $p_y$ since $r_{p_{y+1}} > d_{p_y}$. 
    Since $j^{p_y}_{k^{p_y}}$ is the flag job right next to $p_y$, we have $p_{y+1}$ is available to be started during $[r_{p_{y+1}},s_{j^{p_y}_{k^{p_y}}}]$. 
    By definition (line \ref{alg_line:flag_start}), job $p_{y+1}$ is started at time $s_{j^{p_y}_{k^{p_y}}}$. 
    Since $p_{y+1}$ is a flag job, variable $\hat{f}$ must be $p_{y+1}$ at time $s_{p_{y+1}}$. 
    At time $s_{j^{p_y}_{k^{p_y}}}$, $\hat{f}$ is reset to be $NULL$ in line \ref{alg_line:f_gets_3}. 
    It follows that $\hat{f}$ is updated to be $p_{y+1}$ during $[s_{j^{p_y}_{k^{p_y}}}, s_{p_{y+1}}]$. 
    Clearly, the updating cannot happen in line \ref{alg_line:f_gets_2} because $p_{y+1}$ is primary and cannot be marked as secondary. It remain to consider the case of line \ref{alg_line:f_gets_1} for the updating. 
    Since the only effective updating is during $[s_{j^{p_y}_{k^{p_y}}}, s_{p_{y+1}}]$, it suffices to consider the case of $r_{p_{y+1}}\in [s_{j^{p_y}_{k^{p_y}}}, s_{j^{p_y}_{k^{p_y}}} + (r_{j^{p_y}_{K^{p_y}}} - r_{j^{p_y}_{k^{p_y}-1}})]$. 
    By algorithm's definition, $p_{y+1}$ is started at its release time as in line \ref{alg_line:start_secondary} and hence cannot update the value of $\hat{f}$ (contradiction). 
\end{proof}

\begin{proof} [Proof of Lemma \ref{lem:not_close}]
    Suppose $p$ is a primary tight job. 
    We have $C^p_{k} = d_{j^p_{K^p}} - s_{p}$ for each $k = 1,2,\ldots,K^p$. 
    It follows that $C^p_{K^p} = d_{j^p_{K^p}} - s_{p}$. 
    It follows that $r_{j^p_{K^p-1}} + 1 \geq s_{j^p_{K^p}}$.
    Then, we have $1 + r_{j^p_{K^p}} > 1 + r_{j^p_{K^p-1}} \geq s_{j^p_{K^p}}$. 

    Suppose $j^p_{k^p}$ is the secondary flag job attached to $p$.
    Since primary job $p$ has a secondary job, $p$ must not be tight. 
    Then, we have $C^p_{k^p} = \min_{k = 1,2,\ldots,K^p} C^p_{k} < d_{j^p_{K^p}} - s_{p}$. 
    It follows that $(r_{j^p_{k^p-1}} + 1 - s_p ) + (d_{j^p_{K^p}} - s_{j^p_{k^p}}) = C^p_{k^p} \leq C^p_{K^p} \leq 2 + r_{j^p_{K^p-1}} - s_p$. 
    After rearranging the terms, we have $s_{j^p_{K^p}} \leq r_{j^p_{K^p-1}} - r_{j^p_{k^p-1}} + s_{j^p_{k^p}}$. 
    It follows that $s_{j^p_{K^p}} <  s_{j^p_{k^p}} + (r_{j^{p}_{K^p}} - r_{j^p_{k^p-1}})$ since $r_{j^{p}_{K^p}} > r_{j^{p}_{K^p-1}}$. 

    Here, we show that $s_{j^p_{k^p}} + (r_{j^{p}_{K^p}} - r_{j^p_{k^p-1}}) > 1 + r_{j^{p}_{K^{p}}}$. 
    Since $p$ is not tight, we have $C^p_{k^p} < d_{j^p_{K^p}} - s_p$, and hence, $1 + r_{j^p_{k^p-1}} < s_{j^{p}_{k^p}}$. 
    It follows that $s_{j^{p}_{k^p}} + (r_{j^{p}_{K^{p}}} - r_{j^p_{k^p-1}}) > 1 + r_{j^{p}_{K^{p}}}$. 
\end{proof}

\begin{proof} [Proof of Lemma \ref{lem:alg_cost}]
    For (i), It suffices to show that algorithm is not busy during $(d_{p_x} + \gamma, s_{p_{x+1}})$. 
    Lemma \ref{lem:agreeable} (vi) says that any job released before or at time $s_{p_x}$ must be completed before $d_{x}$. Since line \ref{alg_line:critical_case} does not hold for $p_x$, we have that all the jobs released during $[s_{p_x}, d_{p_x}]$ (i.e. $\mathcal{J}^{p_x}$) are released during $[s_{p_x}, s_{p_x} + \gamma]$. 
    Lemma \ref{lem:agreeable} guarantees that no job is stopped during $[s_{p_x}, d_{p_{x+1}})$ and $d_{p_x} < r_{p_{x+1}}$. 
    It follows that each job $j$ in $\mathcal{J}^{p_x}$ is processed during $[r_j,r_j + 1]$. 
    Take any job $j$ released after $d_{p_x}$. 
    Lemma \ref{lem:agreeable} (vii) says $j$ is not started due to any job completed before $p_x$. 
    Since $p_x$ is not tight, $j$ is not started due to $p_x$. 
    Then, job $j$ is not started until time $s_{p_{x+1}}$, since $p_{x+1}$ is the flag job right next to $p_x$. 
    It follows that algorithm cannot process any job during $(d_{p_x} + \gamma, s_{p_{x+1}})$. 
    
    For (ii), it suffices to show that algorithm is not busy during $(2 + r_{j^{p_x}_{K^{p_x}}}, s_{p_{x+1}})$. 
    Lemma \ref{lem:agreeable} (vi) says that any job released before or at time $s_{p_x}$ must be completed before $d_{x}$. 
    Since $p_x$ is tight, the algorithm's definition and Lemma \ref{lem:agreeable} (ii) guarantees that any job released during $(s_{p_x}, 1 + r_{j^{p_x}_{K^{p_x}}}] $ is processed during $[r_j,r_j + 1]$. 
    Take any job $j$ released after $1 + r_{j^{p_x}_{K^{p_x}}}$. 
    Lemma \ref{lem:agreeable} (vii) says $j$ is not started due to any job completed before $p_x$. 
    Also clearly, $j$ is not started due to $p_x$. 
    Then, job $j$ is not started until time $s_{p_{x+1}}$, because $p_{x+1}$ is the flag job right next to $p_x$. 
    It follows that algorithm cannot process any job during $(2 + r_{j^{p_x}_{K^{p_x}}}, s_{p_{x+1}})$. 

    For (iii), note that $1 + r_{j^{p_x}_{k^{p_x}-1}} - s_{p_x} \leq 1 + r_{j^{p_x}_{K^{p_x}-1}} - s_{p_x} \leq 1 + \gamma$ by the definition of $K^{p_x}$.    
    It suffices to show that algorithm is not busy during $(1 + r_{j^{p_x}_{k^{p_x}-1}},s_{p_{x+1}})$. 
    Lemma \ref{lem:agreeable} (vi) says that any job released before $s_{p_x}$ is completed before $d_{p_x}$. 
    Lemma \ref{lem:agreeable} guarantees that any job $j$ released during $[s_{p_x}, r_{j^{p_x}_{k^{p_x}-1}}]$ is started at $r_j$ and completed at $r_j+1$ without any interruption. 
    Lemma \ref{lem:agreeable} together with the fact that $p_x$ is not tight also guarantees that any job $j$ released during $(r_{j^{p_x}_{k^{p_x}-1}}, d_{p_x}]$ is started at $r_j$, then stopped at time $d_{p_x}$, and is not started until time $s_{p_{x+1}}$.  
    Lemma \ref{lem:agreeable} together with the fact that $p_x$ is not tight also guarantees that any job $j$ released after time $d_{p_x}$ is not started until time $s_{p_{x+1}}$. 
    Therefore, we have shown that algorithm cannot be processing any job during $(1 + r_{j^{p_x}_{k^{p_x}-1}},s_{p_{x+1}})$. 

    For (iv), first we show that algorithm is not busy during $(1 + r_{j^{p_x}_{k^{p_x}-1}},s_{j^{p_x}_{k^{p_x}}})$ and $(s_{j^{p_x}_{k^{p_x}}} + (r_{j^{p_x}_{K^{p_x}}} - r_{j^{p_x}_{k^{p_x}-1}}), s_{p_{x+1}})$. 
    Lemma \ref{lem:agreeable} (vi) says that any job released before or at $s_{p_x}$ is completed before $d_{p_x}$. 
    Lemma \ref{lem:agreeable} together with the fact that $p_x$ is not tight imply the following statements: 
    (1) any job $j$ released during $(s_{p_x}, r_{j^{p_x}_{k^{p_x}-1}}]$ is started at $r_j$ and completed at $r_j+1$ without any interruption; 
    (2) any job $j$ released during $(r_{j^{p_x}_{k^{p_x}-1}}, d_{p_x}]$ is started at $r_j$, then stopped at time $d_{p_x}$, is not started until time $s_{j^{p_x}_{k^{p_x}}}$, and is processed during $[s_{j^{p_x}_{k^{p_x}}}, d_{j^{p_x}_{k^{p_x}}}]$ without any interruption; 
    (3) any job $j$ released during $(d_{p_x},s_{j^{p_x}_{k^{p_x}}}]$ is not started until time $s_{j^{p_x}_{k^{p_x}}}$, and processed during $[s_{j^{p_x}_{k^{p_x}}}, d_{j^{p_x}_{k^{p_x}}}]$ without any interruption; 
    (4) any job $j$ released during $(s_{j^{p_x}_{k^{p_x}}}, s_{j^{p_x}_{k^{p_x}}} + (r_{j^{p_x}_{K^{p_x}}} - r_{j^{p_x}_{k^{p_x}-1}})]$ is started at its release time, and processed during $[r_j,r_j+1]$ without any interruption.
    (5) any job $j$ released during $(s_{j^{p_x}_{k^{p_x}}} + (r_{j^{p_x}_{K^{p_x}}} - r_{j^{p_x}_{k^{p_x}-1}}), s_{p_{x+1}}]$ is not started until time $s_{p_{x+1}}$. 
    Therefore, we have shown that algorithm is not busy during $(1 + r_{j^{p_x}_{k^{p_x}-1}},s_{j^{p_x}_{k^{p_x}}})$ and $(d_{j^{p_x}_{k^{p_x}}} + (r_{j^{p_x}_{K^{p_x}}} - r_{j^{p_x}_{k^{p_x}-1}}), s_{p_{x+1}})$. 

    At last for (iv), note that the cost during $[s_{p_x}, 1 + r_{j^{p_x}_{k^{p_x}-1}}]$ and $[s_{j^{p_x}_{k^{p_x}}}, d_{j^{p_x}_{k^{p_x}}} + (r_{j^{p_x}_{K^{p_x}}} - r_{j^{p_x}_{k^{p_x}-1}})]$ is exactly $2 + r_{j^{p_x}_{K^{p_x}}} - s_{p_x}$. 
\end{proof}

\begin{proof} [Proof of Lemma \ref{lem:opt}]
    Take any optimal scheduling $O$. 
    Let $O_0$ be $O$. For $b = 1,2,\ldots,|\mathcal{P}|$, let $O_b$ be a new scheduling that starts job $p_b$ at its starting deadline, starts jobs $\mathcal{A}(p_b)$ (if non-empty) at their release times respectively, and has the same scheduling as $O_{b-1}$ for all the rest jobs. 
    By Lemma \ref{lem:agreeable} and \ref{lem:not_close}, we have that $r_{j^{p_{x}}_{K^{p_{x}}}} + 1 < r_{p_{x+1}}$ if $\mathcal{A}(p_{x})$ is non-empty, for each $x = 1,2,\ldots,|\mathcal{P}|-1$.  
    By Lemma \ref{lem:agreeable}, we have $d_{p_{x}} < r_{p_{x+1}}$ if $\mathcal{A}(p_{x})$ is empty, for each $x = 1,2,\ldots,|\mathcal{P}|-1$. 
    It follows that for each $x = 1,2,\ldots,b-1$, the set of jobs accommodated within the $x$-th block of the scheduling $O_{b-1}$ is exactly $\{p_x\}\cup \mathcal{A}(p_x)$. 
    Clearly, $O_{|\mathcal{P}|}$ is the desired optimal scheduling. 
    It remains to show $O_b$ has no higher cost than $O_{b-1}$ for each $b = 1,2,\ldots,|\mathcal{P}|$ with the induction hypothesis that $O_{b-1}$ is optimal. 

    Note that the jobs $\cup_{x = 1,2,\ldots,b-1} \{p_x\}\cup \mathcal{A}(p_x)$ are accommodated within the first $b-1$ blocks of $O_{b-1}$ and the remaining jobs are either $p_b$ or released after $s_{p_b}$. 
    Starting job $p_b$ at $s_{p_b}$ does not increase the overall cost of $O_{b-1}$. 
    Since $O_{b-1}$ is optimal, without loss of generality, assume that $O_{b-1}$ starts $p_b$ at its starting deadline $s_{p_b}$.
    Moreover, it is clear that the $b$-th block of $O_{b-1}$ must accommodate job $p_b$. 

    We show that the $b$-th block of $O_{b-1}$ cannot accommodate any job $j$ in $\cup_{x = b+1,b+2,\ldots, |\mathcal{P}|} \{p_x\}\cup \mathcal{A}(p_x)$. 
    Assume $j\in \cup_{x = b+1,b+2,\ldots, |\mathcal{P}|} \{p_x\}\cup \mathcal{A}(p_x)$ is accommodated within the $b$-th block of $O_{b-1}$ and we show a contradiction. 
    Observe that jobs $\cup_{x = 1,2,\ldots,b-1} \{p_x\}\cup \mathcal{A}(p_x)$ are scheduled within the previous $b-1$ blocks and each job in $\cup_{x = b+1,b+2,\ldots, |\mathcal{P}|} \{p_x\}\cup \mathcal{A}(p_x)$ is either $p_b$ or released after $s_{p_b}$. 
    Then, the $b$-th block of $O_{b-1}$ must start at time $s_{p_b}$ and end at or after time $r_j + 1$. 
    Note that we have $d_{p_b} < r_{p_{b+1}} \leq r_j$ by Lemma \ref{lem:agreeable}. 
    If $\mathcal{A}(p_b)$ is empty, then no job is being processed during $[d_{p_b}, r_{p_{b+1}}]\neq \emptyset$ (contradiction) because the any job from $\cup_{x = b+1,b+2,\ldots, |\mathcal{P}|} \{p_x\}\cup \mathcal{A}(p_x)$ is released at or after time $r_{p_{b+1}}$. 
    If $\mathcal{A}(p_b)$ is non-empty, then no job is being processed during $[1 + r_{j^{p_b}_{K^{p_b}}}, r_{p_{b+1}}]\neq \emptyset$ (contradiction) because it is feasible for jobs $\mathcal{A}(p_b)$ to be started at their release times and any job from $\cup_{x = b+1,b+2,\ldots, |\mathcal{P}|} \{p_x\}\cup \mathcal{A}(p_x)$ is released at or after time $r_{p_{b+1}}$, where we have $1 + r_{j^{p_b}_{K^{p_b}}} < r_{p_{b+1}} \leq r_j$ by Lemma \ref{lem:agreeable} and \ref{lem:not_close}. 
    Therefore, the $b$-th block can only accommodate jobs from $\{p_b\}\cup \mathcal{A}(p_b)$. 

    If $\mathcal{A}(p_b)$ is empty, the we are done, because $O_{b-1}$ and $O_b$ have the same scheduling of all the jobs by their definition and the assumption. 
    It suffices to consider the case that $\mathcal{A}(p_b)$ is non-empty. 
    Consider the special case that the set of jobs accommodated within the $(b+1)$-th block of $O_{b-1}$ is a subset of $\mathcal{A}(p_b)$. 
    In this case, we have the cost of $O_{b-1}$ equal to $C_1 + A + B + C_2$ where $C_1$ is the total cost of the first $b-1$ blocks, $A$ is the cost of the $b$-th block, $B$ is the cost of the $(b+1)$-th block, and $C_2$ is the total cost of the remaining blocks. 
    On the other hand, we have the cost of $O_b$ at most $C_1 + A + (r_{j^{p_{b}}_{K^{p_{b}}}} - s_{p_{b}}) + C_2$ where $A + (r_{j^{p_{b}}_{K^{p_{b}}}} - s_{p_{b}})$ is an upper bound of the cost of the $b$-th block of $O_b$ that processes $\{p_{b}\}\cup \mathcal{A}(p_{b})$. 
    Since $r_{j^{p_{b}}_{K^{p_{b}}}} - s_{p_{b}} \leq 1 \leq B$. We have shown that $O_b$ has no higher cost than $O_{b-1}$ in the special case. 

    It remains to consider the case that the $(b+1)$-th block of $O_{b-1}$ accommodates some job other than $\mathcal{A}(p_b)$. By Definition, all the jobs $\cup_{x = 1,2,\ldots,b-1} \{p_x\}\cup \mathcal{A}(p_x)$ are accommodated within the first $b-1$ blocks of $O_{b-1}$. It follows that the $(b+1)$-th block of $O_{b-1}$ must accommodate some job in $\cup_{x = b+1,b+2,\ldots, |\mathcal{P}|} \{p_x\}\cup \mathcal{A}(p_x)$. 
    Since each job in $\cup_{x = b+1,b+2,\ldots, |\mathcal{P}|} \{p_x\}\cup \mathcal{A}(p_x)$ is either $p_{b+1}$ or released after $s_{p_{b+1}}$, the $(b+1)$-th block of $O_{b-1}$ must accommodate job $p_{b+1}$. 
    Precisely, it suffices to consider the case that the $b$-th block of $O_{b-1}$ accommodates job $p_b$ and a proper subset of $\mathcal{A}(p_b)\neq \emptyset$ only, $O_{b-1}$ starts $p_b$ at time $s_{p_b}$, and the $(b+1)$-th block of $O_{b-1}$ accommodates job $p_{b+1}$.  
    We show a contradiction and hence this case does not exist. 
    Suppose the $b$-th block of $O_{b-1}$ accommodates job $p_b, j^{p_b}_{1}, \ldots, j^{p_b}_{\tilde{k}-1}$ for some $\tilde{k} = 1,2,\ldots,K^{p_b}$ such that jobs $j^{p_b}_{\tilde{k}}$ is accommodated within the $(b+1)$-th block of $O_{b-1}$. 
    Note that $p_b$ cannot be tight, otherwise it would be infeasible for job $j^{p_b}_{\tilde{k}}$ to be scheduled within the $(b+1)$-th block, because $C^{p_b}_{\tilde{k}} = d_{j^{p_b}_{K^{p_b}}} - s_{p_b}$ implies $s_{j^{p_b}_{\tilde{k}}} \leq 1 + r_{j^{p_b}_{\tilde{k}-1}}$. 
    It follows from the definition of $\mathcal{A}(p_b)$ that $p_b$ has a secondary flag job $j^{p_b}_{k^{p_b}}$. 
    It follows from Lemma \ref{lem:agreeable} (iv) that $r_{p_{b+1}} > s_{j^{p_b}_{k^{p_b}}} + (r_{j^{p_b}_{K^{p_b}}} - r_{j^{p_b}_{k^{p_b}-1}})$. 
    It follows that the $(b+1)$-th block of $O_{b-1}$ starts at time $s_{j^{p_b}_{\tilde{k}}}$ and ends at or after time $r_{p_{b+1}} + 1$, where note that we have $s_{j^{p_b}_{\tilde{k}}} \leq s_{j^{p_b}_{K^{p_b}}} < r_{p_{b+1}}$ by Lemma \ref{lem:agreeable} and \ref{lem:not_close}. 
    Observe that $O_{b-1}$ and $O_b$ have the same scheduling/cost outside $[s_{p_b}, r_{p_{b+1}} + 1]$, because $O_{b-1}$ and $O_b$ only differ in the scheduling of $\mathcal{A}(p_b)$, and we have $s_{p_b} < r_{j^{p_b}_{1}}$ and $s_{j^{p_b}_{K^{p_b}}} < r_{p_{b+1}}$. 
    It suffices to show that the cost of $O_{b-1}$ during $[s_{p_b}, r_{p_{b+1}} + 1]$ is strictly greater than the cost of $O_b$ during $[s_{p_b}, r_{p_{b+1}} + 1]$. 
    
    Observe that $O_{b-1}$ is processing jobs during $[s_{p_b}, r_{j^{p_b}_{\tilde{k}-1}} + 1]$ and $[s_{j^{p_b}_{\tilde{k}}}, r_{p_{b+1}} + 1]$. 
    Observe that the cost of $O_{b-1}$ during $[s_{p_b}, r_{j^{p_b}_{\tilde{k}-1}} + 1]$ and $[s_{j^{p_b}_{\tilde{k}}}, r_{p_{b+1}} + 1]$ is equal to the cost of a scheduling, say $S_1$, of jobs $p_b, j^{p_b}_1,\ldots,j^{p_b}_{K^{p_b}},p_{b+1}$ such that $p_b,j^{p_b}_1,\ldots,j^{p_b}_{\tilde{k}-1}$ are scheduled within one block optimally (locally), jobs $j^{p_b}_{\tilde{k}},\ldots,j^{p_b}_{K^{p_b}},p_{b+1}$ are scheduled within another block optimally (locally), where the second block accommodating $j^{p_b}_{\tilde{k}},\ldots,j^{p_b}_{K^{p_b}},p_{b+1}$ may have strictly higher cost than scheduling the same jobs $j^{p_b}_{\tilde{k}},\ldots,j^{p_b}_{K^{p_b}},p_{b+1}$ in two separate blocks. 
    Precisely, in scheduling $S_1$, $p_b$ is started at its starting deadline, $j^{p_b}_1,\ldots,j^{p_b}_{\tilde{k}-1}$ are started at their release time respectively, $j^{p_b}_{\tilde{k}},\ldots,j^{p_b}_{K^{p_b}}$ are started at their starting deadline respectively, and $p_{b+1}$ is started at its release time. 
    Consider another scheduling $S_2$ of jobs $p_b, j^{p_b}_1,\ldots,j^{p_b}_{K^{p_b}},p_{b+1}$ such that $p_b,j^{p_b}_1,\ldots,j^{p_b}_{k^{p_b}-1}$ are scheduled within one block optimally (locally), and $j^{p_b}_{k^{p_b}},\ldots,j^{p_b}_{K^{p_b}},p_{b+1}$ are scheduled within another block optimally (locally). 
    Precisely, in scheduling $S_2$, $p_b$ is started at its starting deadline, 
    jobs $j^{p_b}_1,\ldots,j^{p_b}_{k^{p_b}-1}$ are started at their release times, 
    jobs $j^{p_b}_{k^{p_b}},\ldots,j^{p_b}_{K^{p_b}}$ are started at their starting deadlines, 
    and $p_{b+1}$ is started at its release time. 
    By the definition of the secondary job $j^{p_b}_{k^{p_b}}$, the cost of $S_1$ during $[s_{p_b}, d_{j^{p_b}_{K^{p_b}}}]$ is equal to $C^{p_b}_{\tilde{k}}$ that is at least $C^{p_b}_{k^{p_b}}$ that is the cost of $S_2$ during $[s_{p_b}, d_{j^{p_b}_{K^{p_b}}}]$. 
    It follows that the cost of $S_1$ during $[s_{p_b}, 1 + r_{p_{b+1}}]$ is at least the cost of $S_2$ during $[s_{p_b}, 1 + r_{p_{b+1}}]$, since $s_{j^{p_b}_{K^{p_b}}} < r_{p_{b+1}}$ and both scheduling $S_1$ and $S_2$ are busy during $[d_{j^{p_b}_{K^{p_b}}}, r_{p_{b+1}}+1]$. 
    Note that the cost of the scheduling $S_2$ during $[s_{p_b}, 1 + r_{p_{b+1}}]$ is greater than $(1 + r_{j^{p_b}_{k^{p_b}-1}} - s_{p_b}) + (1 + r_{j^{p_b}_{K^{p_b}}} - r_{j^{p_b}_{k^{p_b}-1}}) = 2 + (r_{j^{p_b}_{K^{p_b}}} - s_{p_b})$ where $1 + r_{j^{p_b}_{k^{p_b}-1}} - s_{p_b}$ is the cost of processing $p_b, j^{p_b}_1,\ldots,j^{p_b}_{k^{p_b}-1}$, and the cost of processing $j^{p_b}_{k^{p_b}},\ldots,j^{p_b}_{K^{p_b}}, p_{b+1}$ in the scheduling $S_2$ is strictly greater than $1 + r_{j^{p_b}_{K^{p_b}}} - r_{j^{p_b}_{k^{p_b}-1}}$ since $r_{p_{b+1}} > s_{j^{p_b}_{k^{p_b}}} + r_{j^{p_b}_{K^{p_b}}} - r_{j^{p_b}_{k^{p_b}-1}}$. 
    On the other hand, the cost of $O_b$ during $[s_{p_b}, r_{p_{b+1}} + 1]$ is at most $2 + (r_{j^{p_b}_{K^{p_b}}} - s_{p_b})$ where $1 + (r_{j^{p_b}_{K^{p_b}}} - s_{p_b})$ is the cost during $[s_{p_b}, 1 + r_{j^{p_b}_{K^{p_b}}}]$ for processing jobs $p_b, j^{p_b}_1,\ldots,j^{p_b}_{K^{p_b}}$ and $1$ is the cost during $[r_{p_{b+1}}, r_{p_{b+1}}+1]$ in the scheduling $O_b$. 
    Eventually, we have shown that the cost of $O_{b-1}$ during $[s_{p_b}, r_{p_{b+1}} + 1]$ is strictly greater than the cost of $O_b$ during $[s_{p_b}, r_{p_{b+1}} + 1]$ (contradiction). 
\end{proof}

\begin{proof} [Proof of Lemma \ref{lem:critical_jobs}]
We present Algorithm \ref{alg:critical} that outputs the critical jobs and takes $O(|\mathcal{J}^p|)$ time.

\begin{algorithm}
\caption{Determine critical jobs}
\label{alg:critical}
\algrenewcommand\algorithmicrequire{\textbf{Input:}}
\algrenewcommand\algorithmicensure{\textbf{Output:}}
\begin{algorithmic}[1]
\Require{$\mathcal{J}^p$ is sorted and stored in a doubly linked list $L$ such that $l_1 < l_2$ iff $r_{L(l_1)} \leq r_{L(l_2)}$ }
\Ensure{a doubly linked list $L^*$ such that $r_{L^*(1)} < r_{L^*(2)}< \cdots < r_{L^*(|L^*|)}$ and $L^*$ contains exactly the set of jobs $\left\{j\in \mathcal{J}^p: \neg([r_{j'},d_{j'}]\subset [r_j,d_j]), \forall j'\in \mathcal{J}^p\setminus \{j\}\right\}$; }
\State $j\gets$ the first element of $L$; \label{alg_line:j_update_1}
\While{$j$ is not the last element in $L$}
    \State $j\gets$ the element next to $j$; 
    \State $j'\gets$ the element preceding $j$; 
    \While{$r_{j'} = r_j \lor s_{j} \leq s_{j'}$} \label{alg_line:inner_while}
    \If{$s_{j} \leq s_{j'}$}
        \State Delete element $j'$ from the linked list $L$; \label{alg_line:delete_j'}
        \If{$j$ is the first element of $L$}
            \State Break;
        \Else
            \State $j'\gets$ the element preceding $j$;
        \EndIf
    \Else \Comment{This implies $r_{j'} = r_j$ and $s_{j'} < s_j$}
        \State Delete element $j$ from the linked list $L$; \label{alg_line:delete_j}
        \State $j\gets j'$; \label{alg_line:j_update_2}\Comment{now $j$ and its preceding element (if any) do not satisfy the inner while condition}
        \State Break;
    \EndIf
    \EndWhile 
\EndWhile
\State \textbf{Return} $L$; 
\end{algorithmic}
\end{algorithm}

Before showing the correctness of Algorithm \ref{alg:critical}, consider proving two claims.  
Claim 1: all the deleted elements are redundant jobs. 
If a job $j'$ is deleted from the linked list in line \ref{alg_line:delete_j'}, then there exists a job $j$ such that $j'$ precedes $j$ and $s_{j'} \geq s_j$. It follows that $[r_{j'}, s_{j'}]\supset [r_j,s_j]$, which says $j'$ is redundant. 
If a job $j$ is deleted from the linked list in line \ref{alg_line:delete_j}, then there exists a job $j'$ preceding $j$ such that $r_{j'} = r_j$ and $s_{j'} < s_j$. It follows that $[r_j, s_j]\supset [r_{j'}, s_{j'}]$, which says $j$ is redundant. 

Claim 2: the output linked list of jobs is agreeable, i.e., $r_{L^*(1)} < r_{L^*(2)} < \cdots < r_{L^*(|L^*|)}$ and $d_{L^*(1)} < d_{L^*(2)} < \cdots < d_{L^*(|L^*|)}$. 
We prove the claim by induction. 
Base case: at the end of the first round of the outer While loop, we show that all the elements in $L$ before $j$ including $j$ are agreeable. 
Suppose $j_1,j_2$ are the first two jobs of $\mathcal{J}^p$. 
Consider the following cases. 
Case 1: $r_{j_1} = r_{j_2}$ and $d_{j_1} < d_{j_2}$.  
Job $j_2$ is deleted in line \ref{alg_line:delete_j}. 
The value of $j$ at the end of the first round of the outer While loop is $j_1$. 
Case 2: $r_{j_1} = r_{j_2}$ and $d_{j_1} > d_{j_2}$.  
Job $j_1$ is deleted in line \ref{alg_line:delete_j'}. 
The value of $j$ at the end of the first round of the outer While loop is $j_2$. 
Case 3: $r_{j_1} < r_{j_2}$ and $d_{j_1} > d_{j_2}$.  
Job $j_1$ is deleted in line \ref{alg_line:delete_j'}. 
The value of $j$ at the end of the first round of the outer While loop is $j_2$. 
Case 4: $r_{j_1} < r_{j_2}$ and $d_{j_1} < d_{j_2}$.  
Both job $j_1$ and $j_2$ are kept in the linked list. 
The value of $j$ at the end of the first round of the outer While loop is $j_2$. 

Inductive case: assume that at the end of the $k$-th round of the outer While loop, all the elements in $L$ before $j$ including $j$ are agreeable (Induction Hypothesis), and we show that at the end of the $(k+1)$-th round of the outer While loop, all the elements in $L$ before $j$ including $j$ are agreeable. 
Suppose $j_1,j_2,\ldots,j_h$ are all the elements in $L$ before $j$ including $j$ at the end of the $k$-th round of the outer While loop such that $r_{j_1} < r_{j_2} < \cdots < r_{j_h}$ and $d_{j_1} < d_{j_2} < \cdots < d_{j_h}$.
Suppose $j_{h+1}$ is the element next to $j_h$ in $L$ at the end of the $k$-th round of the outer While loop. 
Consider the following cases. 
Case 1: $r_{j_h} = r_{j_{h+1}}$ and $d_{j_h} < d_{j_{h+1}}$.  
Job $j_{h+1}$ is deleted in line \ref{alg_line:delete_j}. 
The value of $j$ at the end of the $(k+1)$-th round of the outer While loop is $j_h$. 
Therefore, the elements before $j$ including $j$ at the end of the $(k+1)$-th round of the outer While loop are equal to the elements before $j$ including $j$ at the end of the $k$-th round of the outer While loop. 
The agreeability follows. 
Case 2: $d_{j_h} \geq d_{j_{h+1}}$.  
Job $j_h$ is deleted in the first round of the inner While loop in line \ref{alg_line:delete_j'}. 
Note that we have $r_{j_1} < \cdots < r_{j_{h-1}} < r_{j_h} \leq r_{j_{h+1}}$ by Induction Hypothesis. 
It follows that the inner While loop keeps running and then jobs $\{j_x: d_{j_x} \geq d_{j_{h+1}}, x = 1,2,\ldots,h-1\}$ are deleted in line \ref{alg_line:delete_j'} because $r_{j_x} < r_{j_{h+1}}$ for each $x \leq h-1$. 
It follows that all the remaining elements in $L$ before $j_{h+1}$ are agreeable. 
Note that the value of $j$ at the end of the $(k+1)$-th round of the outer While loop is $j_{h+1}$. 
Case 3: $r_{j_h} < r_{j_{h+1}}$ and $d_{j_h} < d_{j_{h+1}}$.  
In this case, no job is deleted. 
The value of $j$ at the end of the $(k+1)$-th round of the outer While loop is $j_{h+1}$. 
We have the set of jobs $j_1,j_2,\ldots,j_{h+1}$ is agreeable. 

We show that the output of Algorithm \ref{alg:critical} is exactly the set of critical jobs of the input $\mathcal{J}^p$, i.e., \\
$\left\{j\in \mathcal{J}^p: \neg([r_{j'},d_{j'}]\subset [r_j,d_j]), \forall j'\in \mathcal{J}^p\setminus \{j\}\right\}$. 
By Claim 1, the set of critical jobs is contained within $L^*$. 
For the other direction, take any job $j$ in $L^*$. 
If $j$ is not critical, then there exists a critical job $j'\in \mathcal{J}^p$ such that $[r_{j'},d_{j'}]\subset [r_j,d_j]$. 
It follows that $j'\in L^*$. 
The existence of the pair of $j$ and $j'$ contradicts to Claim 2. 

It remains to compute its time complexity. 
Claim 3: the inner While loop takes $O(|\mathcal{J}^p|)$ time. 
Note that each time $r_{j'} = r_j \lor s_{j} \leq s_{j'}$ in line \ref{alg_line:inner_while} holds, some job is deleted from the linked list. Since there are at most $|\mathcal{J}^p|$ jobs in the linked list, the inner While loop is run at most $|\mathcal{J}^p|$ times. For each round of the inner While loop, every line takes a constant time. Therefore, Claim 3 is proven. 
Claim 4: the outer While loop has $|\mathcal{J}^p|$ rounds. 
For each round of the outer While loop, either $j$ is updated to be the element next to $j$ in line \ref{alg_line:j_update_1}, or, $j$ is updated to be the element next to $j$ in line \ref{alg_line:j_update_1} and then the job stored in $j$ is deleted (line \ref{alg_line:delete_j}) and $j$ is updated back to its old value (line \ref{alg_line:j_update_2}). 
In either of the cases, the number of elements after $j$ in the linked list decreases by one. 
Claim 4 follows. 
Therefore, we have shown that Algorithm \ref{alg:critical} takes $O(|\mathcal{J}^p|)$ time.

\end{proof}

\begin{proof} [Proof of Theorem \ref{thm:restart_time}]
    First, we make some observations. 
    Note that there are $|\mathcal{P}|$ primary jobs and they are all distinct from each other. 
    Note that Lemma \ref{lem:agreeable} says that $d_{p_x} < r_{p_{x+1}}$ for each $x = 1,2,\ldots,|\mathcal{P}|-1$. 
    Take any $x = 1,2,\ldots,|\mathcal{P}|-1$. 
    Take any $j_1\in \mathcal{J}^{p_x}$ and $j_2\in \mathcal{J}^{p_{x+1}}$. 
    Then, we have $r_{p_x}\leq s_{p_x} < r_{j_1} \leq d_{p_x} < r_{p_{x+1}}\leq s_{p_{x+1}} < r_{j_2}$. 
    It follows that the sets $\mathcal{J}^{p_1}, \mathcal{J}^{p_2}, \ldots, \mathcal{J}^{p_{|\mathcal{P}|}}, \mathcal{P}$ are mutually disjoint. 
    Note that $\mathcal{J}^{p_x}$ is an interval stored in the linked list $L$ for each $x$, because $\mathcal{J}^{p_x}$ is exactly the set of jobs released during a time interval $(s_{p_x},d_{p_x}]$. 
    For any flag job $f$ (stored in variable $\hat{f}$), note that the jobs that are started at time $s_f$ due to $f$ in line \ref{alg_line:flag_start} will be completed without interruption. 
    
    Claim that for any flag job $f$ (stored in variable $\hat{f}$), the set of jobs started at time $s_{f}$ due to $f$ in line \ref{alg_line:flag_start} is an interval stored in the linked list $L$. 
    We show the claim by contradiction. 
    Assume that $j_1,j_2,j_3$ are three jobs such that $r_{j_1} \leq r_{j_2} \leq r_{j_3}$, and $j_1,j_3$ are started in line \ref{alg_line:flag_start} due to the same flag job $f$, and $j_2$ is not available to be started at time $s_{f}$ in line \ref{alg_line:flag_start}. 
    Note that Lemma \ref{lem:agreeable} guarantees that jobs that are started due to primary job $p$ and then stopped (e.g. $j^p_{k_p}, \ldots, j^p_K$) are started due to the flag job right next to $p$ in line \ref{alg_line:flag_start} and will not be stopped again.
    It follows that $j_2$ is started due to some flag job $f'$ strictly before $f$ and then completed without interruption. 
    Since $j_1$ is available to be started at time $s_{f}$ and $r_{j_1} \leq r_{j_2}$, job $j_1$ is also available to be started due to $f'$ or some flag job before $f'$. 
    It follows that $j_1$ is started due to $f'$ or some flag job before $f'$ and then completed without interruption. 
    It contradicts to the assumption that $j_1$ is started due to $f$. 

    Now, we are ready to show how the algorithm achieves linear time complexity. 
    Suppose $L$ is the doubly linked list that stores all the $n$ jobs in order of non-descending release time in the input. 
    Two pointers $a$ and $b$ are used. 
    Pointer $b$ always points at the job to be released next and moves from the first node to the last node of $L$ one by one. 
    Pointer $a$ always points at the first job available to be started. 

    Consider each time of line \ref{alg_line:flag_start}. 
    Suppose pointer $a$ and $b$ point at jobs $L(x)$ and $L(y)$ with $x < y$ respectively, where $x < y$ because at line \ref{alg_line:flag_start} at least one job ($\hat{f}$) is available to be started at time $s_{\hat{f}}$.
    By the above claim, the interval $[x,y-1]$ is exactly the set of jobs available to be started at time $s_{\hat{f}}$. 
    It follows that starting these jobs only requires $O(y-x)$ time, including moving the two pointers to job $L(y)$ and making the same schedule for each job within the interval.
    Since jobs that are started in line \ref{alg_line:flag_start} are never stopped, any job cannot be started twice in line \ref{alg_line:flag_start}. 
    It follows that line \ref{alg_line:flag_begin}-\ref{alg_line:flag_end} take $O(n)$ time. 

    Consider line \ref{alg_line:release_begin}-\ref{alg_line:release_end}. 
    Each time of line \ref{alg_line:release_begin}, pointer $b$ moves to the next node. 
    Note that values of $\hat{f}, p, j^p_{k_p}, j^p_{k_p-1}, j^p_{K_p}$ are stored separately and can be accessed in constant time. 
    Note that the two pointers point at job $j$ before starting job $j$ in line \ref{alg_line:start_primary}, \ref{alg_line:start_tight} and \ref{alg_line:start_secondary}, and pointer $a$ moves to the job next to $j$ after starting job $j$ in line \ref{alg_line:start_primary},\ref{alg_line:start_tight} and \ref{alg_line:start_secondary}. 
    In the cases other than line \ref{alg_line:start_primary},\ref{alg_line:start_tight} and \ref{alg_line:start_secondary}, either updating value of $\hat{f}$ or do nothing takes a constant time. 
    Summing up, line \ref{alg_line:release_begin}-\ref{alg_line:release_end} take $O(n)$ time. 
    Note that after running through line \ref{alg_line:flag_begin}-\ref{alg_line:flag_end} of job $j$'s release, pointer $b$ moves to the job next to $j$. 
    
    
    At time $d_p$, pointer $a$ points at the first job released after time $d_p$ and pointer $b$ points at the job preceding to the location of pointer $a$. 
    It takes $O(|\mathcal{J}^{p}|)$ time to obtain the jobs $\mathcal{J}^{p}$ in a linked list because the jobs $\mathcal{J}^{p}$, as an interval, are adjacent to the location of the two pointers. 
    Therefore, line \ref{alg_line:A_set_begin} takes $O(|\mathcal{J}^{p}|)$ time. 
    Lemma \ref{lem:critical_jobs} says line \ref{alg_line:critical_jobs} takes $O(|\mathcal{J}^{p}|)$ time. 
    Obtaining the value of $K^p$ in line \ref{alg_line:A_set_end} takes $O(|\mathcal{J}^{p}|)$ time. 
    The for loop for computing each $C^p_k$ in line \ref{alg_line:key} takes $O(|\mathcal{J}^{p}|)$ time, because it takes constant time to obtain the values $r_{j^p_{k-1}}, s_{j^p_k}$ when the pointer is at job $j^p_k$, and 
    the values of $s_p,d_{j^p_{K^p}}$ are stored separately and can be accessed in constant time. 
    Computing $\min_{k = 1,2,\ldots,K^p} C^p_k$ also takes $O(|\mathcal{J}^{p}|)$ time, along which the value of $k^p$ can be obtained and job $j^p_{k^p}$ is marked as secondary. 
    Stopping jobs within $\{j\in \mathcal{J}^p: r_j > r_{j^p_{k^p-1}}\}$ in the linked list (line \ref{alg_line:stop}) takes $O(|\mathcal{J}^{p}|)$ time via moving pointer $a$ leftwards to the first job among the jobs that are stopped. 
    It takes $O(1)$ time for $\hat{f}$ to store $j^p_{k_p}$. 
    Therefore, we have shown that it takes $O(|\mathcal{J}^{p}|)$ time for the round of time $d_p$ for line \ref{alg_line:primary_begin}-\ref{alg_line:primary_end}. 
    Summing up for all $p$, line \ref{alg_line:primary_begin}-\ref{alg_line:primary_end} take $O(n)$ time, since $\mathcal{J}^{p_1}, \mathcal{J}^{p_2}, \ldots, \mathcal{J}^{p_{|\mathcal{P}|}}$ are mutually disjoint.     

    It takes a constant time to access the values of $s_{\hat{f}}, d_p, r_j$ where pointer $b$ points at job $j$. 
    Thus, it takes a constant time to compare the values of $s_{\hat{f}}, d_p, r_j$ to decide which block of programs (line \ref{alg_line:flag_begin}-\ref{alg_line:flag_end}, line \ref{alg_line:release_begin}-\ref{alg_line:release_end}, or line \ref{alg_line:primary_begin}-\ref{alg_line:primary_end}) algorithm enters at each stage. 
    At last, we show the total number of comparisions needed is $O(n)$. 
    It suffices to show for a fixed value of variable $j$, the tuple $(s_{\hat{f}}, d_p, r_j)$ is only updated for a constant number of times. 
    Suppose pointer $b$ points at $L(x)$. 
    It comes down to look at what happens during the time period $(r_{L(x-1)}, r_{L(x)})$. 
    Since no new job is released during $(r_{L(x-1)}, r_{L(x)})$, line \ref{alg_line:flag_begin}-\ref{alg_line:flag_end} can only happen at most once during $(r_{L(x-1)}, r_{L(x)})$. 
    It follows that line \ref{alg_line:primary_begin}-\ref{alg_line:primary_end} can only happen at most twice during $(r_{L(x-1)}, r_{L(x)})$. 
    It follows that the value of $\hat{f}$ and $p$ are updated for constant number of times. 
\end{proof}

\end{document}